\documentclass[12pt]{article}
\usepackage{enumitem,graphicx,xcolor,tabularx}
\usepackage[sort&compress,square,numbers]{natbib}
\usepackage{hyperref}
\usepackage{etoolbox}
\usepackage{algorithm}
\usepackage{algpseudocode}
\makeatletter
\preto{\@verbatim}{\topsep=0pt \partopsep=0pt }
\makeatother

\usepackage{mymathdefs}
\usepackage{cleveref}

\newtheorem{theorem}{Theorem}[section]
\crefname{theorem}{theorem}{theorems}
\newtheorem{prop}[theorem]{Proposition}
\crefname{prop}{proposition}{propositions}
\newtheorem{lemma}[theorem]{Lemma}
\crefname{lemma}{lemma}{lemmas}
\numberwithin{equation}{section}
\numberwithin{algorithm}{section}

\theoremstyle{remark}
\newtheorem*{remark}{Remark}

\setlist{itemsep=0ex, topsep=1ex, parsep=0ex}

\def\laz{\lambda_z}
\def\mean{\operatorname{\mathbb E}}
\def\pr{\operatorname{\mathbb P}}

\def\rx{\epsilon}

\def\Vrx{\Cal E}
\def\const{\operatorname{const.}}
\def\sinc{\operatorname{sinc}}
\def\sign{\operatorname{sign}}
\def\XI{\Cal K}

\makeatletter
\def\cdf{c.d.f\@cleandot}
\def\pdf{p.d.f\@cleandot}
\def\pmf{p.m.f\@cleandot}
\def\lhs{l.h.s\@cleandot}
\def\rhs{r.h.s\@cleandot}
\makeatother

\newcommandx\hAT[2][1=\theta,2=\alpha]{h_{#2}({#1})}
\newcommandx\HAT[2][1=\theta,2=\alpha]{H_{#2}({#1})}
\newcommandx\dHAT[2][1=\theta,2=\alpha]{H'_{#2}({#1})}
\newcommandx\JAT[2][1=\theta,2=\alpha]{J_{#2}({#1})}
\newcommandx\IJAT[2][1=\theta,2=\alpha]{J^{-1}_{#2}({#1})}
\newcommandx\dIJAT[2][1=\theta,2=\alpha]{(J^{-1}_{#2})'({#1})}
\newcommandx\XiAT[2][1=\theta,2=\alpha]{\Xi_{#2}({#1})}
\newcommandx\XIAT[2][1=\theta,2=\alpha]{\XI_{#2}({#1})}
\def\alz{_{\alpha,z}}

\def\Id{\operatorname{Id}}
\def\Dexp{\operatorname{Exp}}
\def\Dunif{\operatorname{Unif}}
\newlist{subenum}{enumerate}{3}
\setlist[subenum,1]{label={\alph*}),leftmargin=5ex}

\def\phlabel#1{\phantomsection\label{#1}}
\newboolean{prooflink}
\setboolean{prooflink}{false}
\def\phref#1#2{\ifthenelse{\boolean{prooflink}}{\hfill\hyperref[#1]{#2}}{}}

\makeatletter
\@addtoreset{ALG@line}{algorithm}
\makeatother

\begin{document}
\begin{center}
  \bf\large Complexity of exact sampling of the first passage
  of a stable subordinator
  \\[1ex]
  \rm\normalsize Zhiyi Chi\footnote{
    Department of Statistics, University of Connecticut, U.S.A.
    Email: zhiyi.chi@uconn.edu
  }
\end{center}
\begin{abstract}
  We consider the exact sampling of the first passage of a stable
  subordinator across a non-increasing regular barrier.  First, the
  sampling is reduced to one from a bivariate distribution
  parameterized by the index $\alpha$ of the subordinator and a scalar
  $z$ independent of the barrier.  Then three algorithms are devised
  for different regions of $(\alpha, z)$, using the
  acceptance-rejection method without numerical inversion or
  integration.  When combined, the algorithms allow the exact sampling
  of the first passage to be done with complexity $O(1+
  |\ln(1-\alpha)|)$.
  

\end{abstract}

\section{Introduction} \label{s:intro}
The first passage event of a stochastic process is an important topic
in probability (cf.\ \cite{asmussen:10:wsp,
  kyprianou:14:sh} and the numerous references therein).  Due to
potential applications, the exact sampling of the first passage of a
stable or tempered stable subordinator was investigate in depth
recently \cite{chi:16:spa, dassios:20:acm, gonzalez:23:arxiva}.  In
\cite{gonzalez:23:arxiva}, a fast algorithm was developed to sample
the first passage of a stable 
subordinator of index $\alpha\in (0,1)$.  When the barrier is
non-increasing and regular enough, e.g.\ differentiable, the
complexity of the algorithm is $O(\delta^{-3} + |\ln \alpha| + \ln
N)$, where $\delta$ denotes $1-\alpha$ from now on and $N$ is the
required number of precision bits of the output.  Although $N$ is
exogenous to the first passage, it is an integral part of the
bound as the algorithm requires numerical inversions and
integrations.  This paper shows that exact sampling of the first
passage can be done at complexity $O(1+|\ln \delta|)$ without
numerical inversions or integrations.

In exact sampling, it is generally accepted that numerical inversion or
integration cannot be done exactly in a finite amount of time, while
closed form functions are idealized to allow exact evaluation in a
fixed amount of time.  By convention, power function, exponential
function, logarithm, and trigonometric functions are closed form
functions.  In this paper, the Gamma function $\Gamma(x)$ is also
regarded as closed form.  However, in reality, no real-valued
functions can be evaluated exactly by a digital computer, and
non-exact operations such as numerical inversion or integration can be
highly precise and fast, while exact methods such as the series method
\cite {devroye:86:sv-ny} can have large numerical errors (\cite
{hormann:04:sv-b}, p.~11).  Thus, the notion of exact sampling 
is not entirely aligned with numerical precision.  In view of
this, to make sense of the exact sampling that we will search for,
we conceptualize it as encoding a distribution in terms of a
relatively small (random) number of \iid $\Dunif(0,1)$ variables and
calls to evaluate closed form functions, where $\Dunif(a,b)$ denotes
the uniform distribution on $(a,b)$.  This interpretation turns the
issue of exact sampling into one of representation, hence bypassing
the practical issue of numerical precision.  It underlies almost every
acceptance-rejection (AR) algorithm in practice.

Let $S=(S_t)_{t\ge0}$ be a stable subordinator with $\mean(e^{-rS_t})
= \exp(-t r^\alpha)$, $r\ge0$.  The first passage of $S$ across a
barrier is characterized by the time $\tau$ when it occurs as well as
the pre-passage ``undershoot'' $S_{\tau-}$ and the jump $\Delta S_\tau
= S_\tau - S_{\tau-}$ at the moment of passage.  By definition, the
barrier is a function $b(t)$ and $\tau = \inf\{t>0: S_t > b(t)\}$.
For a non-increasing regular barrier, $\tau$ is easy to sample
\cite{chi:16:spa}, so it is enough to consider how to sample
$(S_{\tau-}, \Delta S_\tau)$  conditioning on $\tau$.  The first
passage occurs either by creeping, i.e., continuously moving across
the barrier, or by jumping across the barrier.  These two cases are
equivalent to $\Delta S_\tau=0$ and $\Delta S_\tau>0$, respectively.
The conditional probability of creeping given $\tau$ is found in
\cite{gonzalez:23:arxiva}, so the problem is reduced to the sampling
of $(S_{\tau-}, \Delta S_\tau)$ in the event of passage by jumping,
which can be further reduced to that of $S_{\tau-}$ alone; see
Algorithm SFP-Alg in \cite {gonzalez:23:arxiva} .  Thus, the
complexity of the sampling of the first passage is essentially
determined by how to sample $S_{\tau-}$ conditioning on $\tau$ and
$\Delta S_\tau$ being positive.

As will be seen shortly, the distribution of $S_{\tau-}$ conditioning
on $\tau=t$ and $\Delta S_\tau$ being positive can be reduced to one
parameterized by $\alpha$ and
\begin{align} \label{e:z-def}
  z = \delta [\alpha/b(t)]^{\alpha /\delta} t^{1/\delta}. 
\end{align}
Furthermore, $z$ has a representation independent of the barrier.  The
parameters play an important role in the paper.  In \cref{s:analysis},
the conditional sampling of $S_{\tau-}$ is formulated as the sampling
from a bivariate probability density function (\pdf) which is
parameterized by $\alpha$ and $z$.  The \pdf is the normalized
version of a certain function $\chi\alz$.  In \cref{s:general}, two
simple algorithms are given, one having complexity $O(1 + z^{-1} (1+
\delta/z)^2)$ and the other $O(\delta^{-2} / C_z)$, where $C_z>0$ is a
decreasing function of $z$.  Given any $\alpha_0\in (0,1)$ and
$z_0>0$, a combination of the algorithms can sample from the
normalized $\chi\alz$ with complexity uniformly bounded for all
$(\alpha,z)$ with $\alpha\le \alpha_0$ or $z\ge z_0$.  However, for
$(\alpha,z)$ close to $(1, 0)$, the combination has complexity
$O(\delta^{-2})$.  To reduce the complexity, most of the paper is
dedicated to the sampling for such $(\alpha,z)$.  In \cref{s:ar}, for
several elementary distributions, sampling routines are obtained with
complexity uniformly bounded for their parameters.  Using these
routines as building blocks, in \cref
{s:large-alpha-small-z,s:subroutine},  an algorithm is devised with
complexity $O(|\ln\delta|)$ for all $(\alpha,z)$ close to $(1,0)$.  By
combining the algorithm with those in \cref{s:general}, the normalized
$\chi\alz$ can be sampled with complexity $O(1 + |\ln\delta|)$ for all
$z$.  In \cref{s:experiments}, experiments that compared the
algorithms and incorporated them to sample the first passage
efficiently are reported.  All the proofs are collected in
\cref{s:appendix}.

In the rest of the section, we fix some notation and recall the AR
method.  The Gamma distribution with shape parameter $a>0$ and rate
$b>0$ has support $[0,\infty)$ with \pdf
\begin{align} \label{e:fp-gamma}
  \gamma_{a,b}(y) = \lfrac{b^a x^{a-1} e^{-bx}}{\Gamma(a)}, \quad
  y>0.
\end{align}
If $a=1$, the distribution is also known as the exponential
distribution with rate $b$, denoted $\Dexp(b)$.  Denote $\sign(x) =
\cf{x>0} - \cf{x<0}$, $x\wedge y = \min(x,y)$, $x\vee
y = \max(x,y)$, and $x_\pm = \max(0, \pm x)$.  By convention, if $x\ge
y$, $[x,y)=(x,y]=(x,y)=\emptyset$.  By sampling from a discrete set
with certain weights assigned to its elements, it means sampling from
the probability mass function obtained by normalizing the weights.

Let $f\ge0$ be a function on a Euclidean space with a finite positive integral.  To sample from the normalized $f$, i.e., the \pdf
$f/\int f$, the AR method uses an envelope (function) $g\ge f$ with
$\int g<\infty$ and proceeds as in \Cref{a:ar-basic}.
\begin{algorithm}[ht]
  \caption{Basic AR} \label{a:ar-basic}
  \begin{algorithmic}[1]
    \Repeat
    \State Sample $X\sim$ normalized $g$ and $U\sim\Dunif(0,1)$
    (independently)
    \Until $Ug(X)\le f(X)$. 
    \label{a:ar-general}
    \State \Return $X$. 
  \end{algorithmic}
\end{algorithm}
In the algorithm, when the condition on \cref{a:ar-general} is
satisfied, $X$ is said to be accepted.  The probability of acceptance
in each iteration is
\[
  p_{\rm accept}=\int \frac{f(x)}{g(x)} \frac{g(x)}{\int g}\,\dd x =
  \frac{\int f}{\int g}.
\]
The complexity of the algorithm is usually measured by $\mean(S)$ with
$S$ its total number of operations.  Let $m_i$ be the number of
operations in the $i$th iteration and $\xi_i$ the indicator that the
algorithm halts at the $i$th iteration.  Then $S=\sumoi k m_k
(1-\xi_1)\cdots(1-\xi_{k-1})$.  Since the $(m_i, \xi_i)$'s are \iid with
$\mean(\xi_1) = p_{\rm accept}$, $\mean(S) = \mean(m_1)/p_{\rm accept}$.

Let $f = \sum_{i\in I} f_i$ and $g = \sum_{i\in I} g_i$, where $I$ is
a finite set.  Suppose $f_i \le g_i$ and $\int g_i$ is known for all
$i$.  Then sometimes it is simpler to use \Cref{a:AR} to sample from
the normalized $f$.  Note that the algorithm is still AR.  If 
the above basic AR is used instead, then the condition on
\cref{a:ar-general2} is changed to $U g(X)\le f(X)$.  In either case,
the expected number of iterations is the same.
\begin{algorithm}[ht]
  \caption{AR with piecewise envelopes}\label{a:AR}
  \begin{algorithmic}[1]
    \Repeat
    \State Sample $Y\in I$ with weights $\int g_i$, $i\in I$.
    \State Sample $X\sim$ normalized $g_Y$ and
    $U\sim\Dunif(0,1)$.
    \Until $Ug_Y (X)\le f_Y(X)$.  \label{a:ar-general2} 
    \State \Return $X$.
  \end{algorithmic}
\end{algorithm}

\section{The sampling problem and related
  functions}  \label{s:analysis}
For $t>0$, the \pdf of $S_t$ is $p_t(x) = t^{-1/\alpha}
p(t^{-1/\alpha} x)$, where letting
\begin{align}\label{e:h-alpha}
  \hAT =\sin(\delta\theta) [\sin(\alpha \theta)]^{\alpha/\delta}
  (\sin\theta)^{-1/\delta},
\end{align}
$p$ is a \pdf supported on $[0,\infty)$ such that for $x>0$
\[
  p(x) = \frac{\alpha}{\delta\pi}\int^\pi_0
  \hAT x^{-1/\delta} \exp\{-\hAT x^{-\alpha/\delta}\}
  \,\dd\theta;
\]
see \cite {sato:99:cup}, p.~87 and \cite{zolotarev:66:stmsp}.
Let $b\in C^1((0,\infty))$ be non-increasing with $0<b(0+)<\infty$.
Let $\tau$ be the first passage time of $S$ across $b$.  It is known
that (cf.~\cite{chi:16:spa})
\begin{align} \label{e:fpt-S}
  \tau\sim B^{-1}(\varsigma), \quad\text{where~} \varsigma\sim p
  \text{~and~} B(t) = t^{-1/\alpha} b(t).
\end{align}
Note that $B$ is strictly decreasing and $B(\tau)>0$.  For $t>0$ with
$B(t)>0$, from \cite {gonzalez:23:arxiva}, $\pr\{\Delta S_\tau=0 \gv
\tau=t\} = -b'(t)/[-b'(t) + \alpha^{-1} t^{-1} b(t)]$.  On the other 
hand, from \cite{chi:16:spa},
\[
  \frac{\pr\{S_{\tau-}\in \dd u, \Delta S_\tau \in \dd v\gv \tau=t,
    \Delta S_\tau>0\}}{\dd u\,\dd v} 
  =
  \const\times p_t(u) \cf{0<b(t) - u < v} v^{-1-\alpha}.
\]
It follows that conditioning on $\tau=t$ and $S_{\tau-}=u\in (0,
b(t))$, $\Delta S_\tau\sim [b(t) - u]U^{-1/\alpha}$ with
$U\sim\Dunif(0,1)$.  Thus, the core issue is how to sample $S_{\tau-}$
conditioning on $\tau=t$ and $\Delta S_\tau$ being positive.  From the
above display,
\begin{align}\nonumber
  \frac{\pr\{S_{\tau-}\in \dd u\gv \tau=t, \Delta S_\tau>0\}}{\dd u}
  =
  \const\times [&b(t) - u]^{-\alpha} p(t^{-1/\alpha} u)
  \\\label{e:undershoot}
  =\const \times [b(t) - u]^{-\alpha} p(s u/b(t))
  &\quad\text{with~} s = B(t), \ u\in (0, b(t)).
\end{align}
Then conditioning on $\tau=t$ and $\Delta S_\tau$ being positive,
\[
  S_{\tau-}\sim b(t) X \quad\text{with $X$ having \pdf~}
  \const\times\cf{0<x<1} (1-x)^{-\alpha} p(sx).
\]
From \eqref{e:h-alpha}, $X$ can be coupled with a random variable
$\Theta\in (0,\pi)$ such that the \pdf of $(X, \Theta)$ at
$(x,\theta)\in (0,1)\times (0,\pi)$ is
\[
  \const\times (1-x)^{-\alpha} x^{-1/\delta}
  \hAT\exp\{-\hAT s^{-\alpha/\delta} x^{-\alpha/\delta}\}.
\]

Let $Y = X^{-\alpha/\delta}-1$.  Then the \pdf of $(Y,\Theta)$
at $(y,\theta)\in  (0,\infty) \times (0,\pi)$ is
\[
  \const\times [1-(y+1)^{-\delta/\alpha}]^{-\alpha} \hAT 
  \exp\{-\hAT s^{-\alpha/\delta} (y+1)\}.
\]
It is easy to check that as $\theta\to0+$, $\hAT\to\hAT[0]
:= \delta\alpha^{\alpha/\delta}$.  Let
\begin{align} \label{e:HAT-z}
  \HAT = \frac{\hAT}{\hAT[0]}, \quad
  z = \hAT[0] s^{-\alpha/\delta}.
\end{align}
From \eqref{e:fpt-S} and \eqref{e:undershoot}, the $z$ in
\eqref{e:HAT-z} is the same as the one in \eqref{e:z-def}.  Then the
\pdf of $(Y,\Theta)$ is the normalized 
\begin{align} \label{e:fpus-joint}
  \chi\alz(y,\theta)
  =
  [1-(y+1)^{-\delta/\alpha}]^{-\alpha} \HAT\exp\{-z\HAT (y+1)\}.
\end{align}

Clearly, if $(Y,\Theta)$ can be sampled, then $S_{\tau-}$ can be
sampled conditioning on $\tau=t$ and $\Delta S_{\tau-}$ being
positive.  To sample from the normalized $\chi\alz(y,\theta)$, we need
some properties of $\HAT$.  From \eqref{e:HAT-z},
\begin{align} \label{e:HAT}
  \HAT = \frac{\sinc(\delta\theta)}{\sinc(\theta)}
  \Sbr{\frac{\sinc(\alpha\theta)}{\sinc(\theta)}}^{\alpha/\delta},
\end{align}
where $\sinc(0)=0$ and $\sinc(x) = \sin(x)/x$ for $x\ne0$.  The
function $\HAT$ yields a simple representation of the $z$ in
\eqref{e:HAT-z}.  To see this, from \eqref{e:fpt-S},
$\varsigma=B(\tau)\sim p$.  Then from \cite{chambers:76:jasa}, 
\begin{align*}
  \varsigma
  \sim\frac{\sin(\alpha \Theta)}{[\sin(\Theta)]^{1/\alpha}}
  \Sbr{\frac{\sin(\delta\Theta)}\xi}^{\delta/\alpha}
  =\hAT[0]^{\delta/\alpha}
  \frac{\sinc(\alpha \Theta)}{[\sinc(\Theta)]^{1/\alpha}}
  \Sbr{\frac{\sinc(\delta\Theta)}\xi}^{\delta/\alpha}
\end{align*}
giving
\begin{align} \label{e:H-z-Exp}
  \hAT[0] \varsigma^{-\alpha/\delta}
  \sim \frac{[\sinc(\Theta)]^{1/\delta} \xi}{
    [\sinc(\alpha\Theta)]^{\alpha/\delta} \sinc(\delta\Theta)
  } = \frac\xi{\HAT[\Theta][\alpha]},
\end{align}
where $\Theta\sim\Dunif(0,\pi)$ and $\xi\sim\Dexp(1)$ are independent.
Then by \eqref{e:HAT-z}, $z$ is a sampled value of
$\xi/\HAT[\Theta] [\alpha]$, the distribution of which does not depend
on the barrier $b$.  Given $z$, the value of $\tau$ is $t = B^{-1}(s)$
with $s=[\hAT[0]/z]^{\delta/\alpha} = \alpha
(\delta/z)^{\delta/\alpha}$.  As a result, Algorithm SFP-Alg in
\cite{gonzalez:23:arxiva} can be reformulated into \Cref{a:fp}.   The
central role of the sampling of the normalized $\chi\alz$ is evident
in \Cref{a:fp}.  Indeed, it is the only non-elementary step in the
algorithm.  The sampling is the focus of the rest of the paper. 
\begin{algorithm}[t]
  \caption{Sampling of $(\tau, S_{\tau-}, \Delta S_\tau)$} \label{a:fp}
  \begin{algorithmic}[1]
    \State Sample $\Theta\sim\Dunif(0,\pi)$, $\xi\sim\Dexp(1)$,
    $U, V$ \iid$\sim \Dunif(0,1)$.
    \State $z\gets \xi/\HAT[\Theta]$, $s\gets \alpha
    (\delta/z)^{\delta/\alpha}$, $t\gets B^{-1}(s)$, $y\gets0$.
    \If{$U\ge-b'(t)/[-b'(t) + \alpha^{-1} t^{-1} b(t)]$}
    \State Sample $(y,\theta)\sim$ normalized $\chi_{\alpha,z}$.
    \EndIf 
    \State $x\gets (1+y)^{-\delta/\alpha}$.
    \State\Return $(t, b(t) x, b(t)(1-x) V^{-1/\alpha})$.
  \end{algorithmic}
\end{algorithm}

Since $\HAT\to1$ as $\alpha\to0$, define $\HAT[\theta][0] = 1$.  On
the other hand, as $\alpha\to1$,
\begin{align}\nonumber
  \HAT
  &\sim  \frac{1+o(1)}{\sinc(\theta)}
  \Sbr{\frac{\sinc(\alpha\theta)}{\sinc(\theta)}}^{\alpha/\delta}
  =\frac{1+o(1)}{\sinc(\theta)}
  \Sbr{1-\frac{\sinc(\theta)-\sinc(\alpha\theta)}
    {\sinc(\theta)}}^{\alpha/\delta} 
  \\\label{e:HAT1}
  &\to
  \nth{\sinc(\theta)}
  \exp\Sbr{-\frac{\theta\sinc'(\theta)}{\sinc(\theta)}}
  =\frac{\exp(1-\cos\theta/\sinc(\theta))}{\sinc(\theta)}
  := \HAT[\theta][1].
\end{align}
It is easy to see that given $\theta$, $\alpha^{-1}\ln\HAT$ as a
function of $\alpha$ is symmetric about $1/2$.  

\begin{prop} \label{p:HATH1} 
  For $\theta\in (0,\pi)$ and $\theta\in (0,1)$,
  \phref{p:p:HATH1}{Proof}
  \begin{align} \label{e:lHAT}
    \frac{\ln\HAT}\alpha
    =
    2\sum^\infty_{n=1}
    \frac{\zeta(2n)}{2n \pi^{2n}}
    \sum^{2n-1}_{k=0} (\alpha^k + \delta^k)\theta^{2n}
    \le \ln \HAT[\theta][1]
  \end{align}
  and $\HAT[\theta][1] = \exp\{\theta^2/2 + V_1(\theta)\}$ with
  $V_1(\theta) = \sum^\infty_{n=2}(2+1/n) \zeta(2n)
  (\theta/\pi)^{2n}$, where $\zeta(s)$ is the Riemann $\zeta$-function.
  Moreover, given $0<\theta_1 < \theta_2<\pi$,
  $\HAT[\theta_2]/\HAT[\theta_1]$ as a function of $\alpha$ is
  increasing on $[1/2,1]$.
\end{prop}
Recall that $\zeta(0)=-1/2$, $\zeta(2) = \pi^2/6$,
and $\zeta(2n)>0$ for all $n\ge1$ (\cite{NIST:10}, 25.6.1).  
\begin{prop} \label{p:H-sand}
  Define \phref{p:p:H-sand}{Proof}
  \[
    \XiAT = \frac{\sinc(\delta\theta)}{\sinc(\pi-\theta)} 
    \Sbr{\frac{\sinc(\pi-\alpha\theta)}
      {\sinc(\pi-\theta)}}^{\alpha/\delta}.
  \]
  Then
  \begin{align} \label{e:HAT-factor}
    \HAT = \XiAT \frac\theta{\pi-\theta}
    \Sbr{1 + \frac{\delta\pi}{\alpha(\pi-\theta)}}^{\alpha / \delta}.
  \end{align}
  Furthermore, for $\alpha\in (1/2,1)$, both $\XiAT$ and 
  \begin{align} \label{e:KAT}
    \XIAT:= \HAT
    \Sbr{1+\frac{\delta\pi}{\alpha(\pi-\theta)}}^{-1/\delta}
    = 
    \frac{\XiAT \theta}{\pi/\alpha-\theta}
  \end{align}
  are increasing on $[(\nth3+\nth{3\alpha})\pi, \pi)$.  Finally,
  \[
    [\ln\XiAT]' \le \theta \Rvl{\Sbr{\ln\nth{\sinc(t)}}''}_{t=\pi -
      \alpha\theta} =
    \theta\Sbr{\nth{\sin^2(\alpha\theta)} - \nth{(\pi
        - \alpha\theta)^2}}
  \]
  and
  \[
    [\ln\HAT]' = [\ln\XiAT]' + \frac{\pi^2}{\theta(\pi-\theta) (\pi -
      \alpha\theta)}.
  \]
\end{prop}

\section{Two simple algorithms}

\label{s:general}
This section supplies two simple algorithms to sample from the
normalized $\chi\alz(y,\theta)$.  Both algorithms apply to all
$(\alpha,z)\in (0,1)\times (0,\infty)$.  Henceforth, denote
\begin{align} \label{e:calpha}
  c_\alpha = (\lfrac\alpha\delta)^\alpha, \quad
  \tau=\tau(z, \theta) = z\HAT,
\end{align}
so that by \eqref{e:fpus-joint}, $\chi\alz(y,\theta) = z^{-1} [1-(y+1)^{-
  \delta/\alpha}]^{-\alpha} \tau e^{-\tau (y+1)}$.

\subsection{The first algorithm} \label{ss:general1} 
As $f(y):=
(y+1)^{-\delta/\alpha}$ is convex in $y>0$, $f'(0)\le [f(y)-f(0)]/y\le
f'(y)$, giving  $(\delta/\alpha) y(y+1)^{-1/\alpha}\le
1-(y+1)^{-\delta/\alpha} \le (\delta/\alpha)y$.  Define
\[
  A\alz(y,\theta) = c_\alpha z^{-1}(1+y) y^{-\alpha} 
  \tau e^{-\tau(y+1)}.
\]
Then
\begin{align} \label{e:general-sand}
  \nth{y+1}\le \frac{\chi\alz(y,\theta)}{A\alz(y,\theta)}
  = \nth{y+1} \Sbr{\frac{(\delta/\alpha) y}{1 -
      (y+1)^{-\delta/\alpha}}}^\alpha
  \le1.
\end{align}
It follows that if the normalized $A\alz(y, \theta)$ can be sampled,
then the normalized $\chi\alz(y,\theta)$ can be sampled by AR with
$A\alz(y,\theta)$ as the envelope.  Define
\[
  w\alz(\theta) = \lfrac\delta {(\tau + \delta)}, \quad
  A^*\alz(\theta)
  =\tau^\alpha e^{-\tau} \Grp{1+\lfrac\delta\tau}
\]
and recall that $\gamma_{a,b}(y)$ denotes a Gamma \pdf as in
\eqref{e:fp-gamma}.  Then
\begin{align}\nonumber
  A\alz(y,\theta)
  &\propto
  \tau e^{-\tau}
  [\Gamma(\delta)\tau^{-\delta}
  \gamma_{\delta,\tau}(y) +
  \Gamma(1+\delta)\tau^{-1-\delta}\gamma_{1+\delta, \tau}(y)]  
  \\\label{e:general-A}
  &\propto A^*\alz(\theta)
  [(1-w\alz(\theta))\gamma_{\delta,\tau}(y) + 
  w\alz(\theta)\gamma_{1+\delta,\tau}(y)].
\end{align}
As a result, under the normalized $A\alz(y,\theta)$, the marginal \pdf
of $\theta$ is the normalized $A^*\alz(\theta)$, and conditioning on
$\theta$, the \pdf of $y$ is a mixture of Gamma \pdf's
\begin{align} \label{e:A-gamma-mix}
  \phi_\theta(y)=(1-w\alz(\theta)) \gamma_{\delta,\tau}(y) +
  w\alz(\theta) \gamma_{1+\delta,\tau}(y).
\end{align}

To sample from the normalized $A^*\alz(\theta)$, from \eqref{e:lHAT},
$\HAT$ is strictly increasing on $[0,\pi)$ with $\HAT[0]=1$.  Then by
$\tau = z\HAT$, $A^*\alz (\theta) \le  (1+\delta/z) z^{-\delta} \tau
e^{-\tau}$.   From \eqref{e:lHAT} again, $\tau >z e^{\alpha
  \theta^2/2} > z + z\alpha \theta^2/2$.  On the other hand, $t
e^{-t}$ is increasing on $[0,1]$ and decreasing on $[1,\infty)$, so
for $t_2>t_1>0$, $t_2 e^{-t_2} \le(t_1\vee1) e^{-(t_1\vee1)}\le
(t_1\vee 1) e^{-t_1}$.  Therefore, letting
\[
  m_s(\theta) = e^{-s\theta^2/2} \cf{0<\theta<\pi}, \quad
  r\alz =(1+\delta/z)z^\alpha [(1+\alpha\pi^2/2) \vee(1/z)],
\]
one has $A^*\alz(\theta) \le (1+\delta/z) z^{-\delta}
[(z+z\alpha\theta^2/2) \vee1] e^{-(z+z\alpha\theta^2/2)} \le r\alz 
e^{-z} m_{\alpha z}(\theta)$.  From the bound, the normalized
$A^*\alz(\theta)$ can be sampled by AR with $r\alz e^{-z} m_{\alpha
  z}(\theta)$ as the envelope.  Once $\theta$ is sampled, $y$ is
sampled from $\phi_\theta$.  The resulting $(y,\theta)$ is a sample
point from the normalized $A\alz(y,\theta)$.  The steps to do this are
described  in \Cref{a:ar-A}.  Note that \cref{a:ar-A5} of the
algorithm follows from \eqref{e:general-sand}.

\begin{algorithm}[t]
  \caption{Sampling from the normalized $\chi\alz$ using envelope
    $A\alz$.}
  \label{a:ar-A} 
  \begin{algorithmic}[1]
    \Repeat \label{a:ar-A1} \Comment{outer loop}
    \Repeat \Comment{inner loop}
    \State Sample $\theta\sim$ normalized $m_{\alpha z}$ and
    $U_1\sim \Dunif(0,1)$.
    \label{a:ar-A2}
    \Until $U_1 r\alz e^{-(1+\alpha \theta^2/2) z} \le
    A^*\alz(\theta)$.
    \State Sample $y\sim\phi_\theta$ and $U_2\sim \Dunif(0,1)$,
    $R\gets \frac{(\delta/\alpha) y}{1-(y+1)^{-\delta/\alpha}}$.
     \label{a:ar-A4}
    \Until $U_2(y+1)\le R^\alpha$.
    \Comment{By \eqref{e:general-sand}}
    \label{a:ar-A5}
    \State \Return $(y,\theta)$.
  \end{algorithmic}
\end{algorithm}

\Cref{a:ar-A} has two non-elementary operations, one is the sampling
from the normalized $m_{\alpha z}$ on \cref{a:ar-A2}, the other
the sampling from a mixture of Gamma \pdf's on \cref{a:ar-A4}.
The normalized $m_s$ can be sampled with a uniformly bounded
complexity by using AR with $\cf{0<\theta<\pi}$ as the envelope if
$0<s\le1$, and via $|Z_s|/\sqrt s$ if $s>1$, where $Z_s$ follows the
normal distribution $N(0,1)$ restricted to $[-\sqrt s\pi, \sqrt
s\pi]$.  On the other hand, the complexity to sample from a Gamma \pdf
is uniformly bounded (\cite{devroye:86:sv-ny}, sections IX3.3 and
IX3.6), so the complexity to sample from $\phi_\theta$ is uniformly
bounded for $\theta$.  Therefore, to analyze the complexity of
\Cref{a:ar-A}, it suffices to consider the expected total number of
iterations.

For the inner loop, the probability of acceptance in each iteration is 
\begin{align} \label{e:accept-prob-A*}
  \frac{\int^\pi_0 A^*\alz(\theta)\,\dd\theta}{
    r\alz\int^\pi_0 e^{-(1+\alpha\theta^2/2)z}
    \,\dd\theta
  }\ge
  \frac{z^\alpha\int^\pi_0 e^{-z\HAT}\,\dd\theta}{
    r\alz
    \int^\pi_0 e^{-(1+\alpha\theta^2/2)z}\,\dd\theta
  }
\end{align}
as $A^*\alz(\theta)\ge \tau^\alpha e^{-\tau} \ge z^\alpha e^{-z\HAT}$.
From \Cref{p:HATH1}, $\HAT\le e^{\alpha(\theta^2/2 +
V_1(\theta))}$.  Then by $e^x \le 1 + x + x^2 e^x$ for $x>0$,
\begin{align} \nonumber
  \HAT
  &
  \le 1 + \alpha\theta^2/2 + \alpha V_1(\theta)
  + (\alpha\theta^2/2 + \alpha V_1(\theta))^2
  e^{\alpha\theta^2/2 + \alpha V_1(\theta)}
  \\
  & \label{e:fpus-H}
  \le 1 + \alpha\theta^2/2 + \alpha\theta^2 V(\theta),
\end{align}
where $V(\theta) = V_1(\theta)/\theta^2 + \theta^2(1/2 +
V_1(\theta)/\theta^2)^2 e^{\theta^2/2 + V_1(\theta)}$.  Then, letting
$s = \alpha z$,
\[
  \frac{\int^\pi_0 e^{-z\HAT}\,\dd\theta}{
    \int^\pi_0 e^{-(1+\alpha\theta^2/2)z}\,\dd\theta
  }
  \ge
  \frac{\int^\pi_0 e^{-(1 + \alpha\theta^2/2 + \alpha \theta^2
      V(\theta))z}\,\dd\theta}{
    \int^\pi_0 e^{-(1+\alpha\theta^2/2)z}\,\dd\theta
  }
  \ge
  \frac{\int^\pi_0 e^{-s(\theta^2/2)(1+ 2\theta^2
      V(\theta))}\,\dd\theta}{
    \int^\pi_0 e^{-s\theta^2/2}\,\dd\theta
  }.
\]
The \rhs tends to 1 as $s\dto0$.  On the other hand, since the \rhs
equals
\[
  \frac{\int^{\sqrt s\pi}_0 e^{-(\theta^2/2)(1+ 2s^{-1}\theta^2
      V(\theta/\sqrt s))}\,\dd\theta}{
    \int^{\sqrt s\pi}_0 e^{-\theta^2/2}\,\dd\theta
  },
\]
by dominated convergence, it tends to 1 as $s\toi$.  Then the \rhs has
a positive lower bound for all $s>0$, so from \eqref
{e:accept-prob-A*}, in each iteration of the inner loop, the
probability of acceptance is $O(z^\alpha/r\alz)$.  As a result,
within each iteration of the outer loop, the expected number of
iterations of the inner loop is $O(z^{-\alpha} r\alz)$.  Next, given
$\theta$, from the lower bound in \eqref{e:general-sand} and $\HAT>1$,
the probability of acceptance at
\cref{a:ar-A5} is
\begin{align*}
  p_\text{accept}(\theta)
  &:=\frac{
    \intzi \chi\alz(y,\theta)\,\dd y
  }{
    \intzi A\alz(y,\theta)\,\dd y
  }
  \ge
  \frac{
    \intzi (y+1)^{-1} A\alz(y)\,\dd y
  }{
    \intzi A\alz(y)\,\dd y
  }
  \\
  &=
  \frac{\intzi y^{-\alpha} e^{-z\HAT y}\,\dd y
  }{\intzi (y^{1-\alpha} + y^{-\alpha}) e^{-z\HAT y}\,\dd y}
  =
  \frac{z\HAT}{z\HAT + \delta} \ge \frac{z}{z + \delta}.
\end{align*}
Thus $p_{\rm a} (\theta)$ is uniformly lower bounded by $\lfrac
z{(z+\delta)}$.  As a result, the expected total number of iterations
is $O(z^{-\alpha} r\alz(1+\delta/z)) = O(1+z^{-1}(1 +\delta/z)^2)$.
In particular, given $z_0>0$, the expected total number is uniformly
bounded for all $\alpha\in(0,1)$ and $z\ge z_0$.

\subsection{The second algorithm}\label{ss:general2}
\begin{prop} \label{p:fpus-squeeze}
  Define $c_{1,\alpha} = \lfrac{c_\alpha}{(c_\alpha + 1)}$ and
  $c_{2,\alpha} = 1\vee(\alpha/\delta)$.  Then for $y>0$,
  \begin{align} \label{e:fpus-squeeze}
    c_{1,\alpha}(y^{-\alpha} + 1) \le
    [1-(y+1)^{-\delta/\alpha}]^{-\alpha}
    \le c_{2,\alpha}(y^{-\alpha}+1)
  \end{align}
  and $\inf_{\alpha\in (0,1)} c_{1,\alpha}>0$.
  \phref{p:p:fpus-squeeze}{Proof}
\end{prop}

As before, denote $\tau =z\HAT$.  Define
\[
  B\alz(y,\theta) = z^{-1} c_{2,\alpha}(y^{-\alpha}+1) \tau
  e^{-\tau(y+1)}.
\]
Then from \eqref{e:fpus-joint} and \Cref{p:fpus-squeeze},
\begin{align} \label{e:chi-B}
  \frac{c_{1,\alpha}}{c_{2,\alpha}}
  \le \frac{\chi\alz(y,\theta)}{B\alz(y,\theta)}
  =
  \frac{[1-(y+1)^{-\delta/\alpha}]^{-\alpha}}
  {c_{2,\alpha}(y^{-\alpha}+1)}
  \le1.
\end{align}
It follows that if the normalized $B\alz(y, \theta)$ can be sampled,
then the normalized $\chi\alz(y,\theta)$ can be sampled by AR with $B\alz(y,\theta)$ as the envelope.  Define
\[
  w^*\alz(\theta) = \nth{\Gamma(\delta) \tau^\alpha + 1},
  \quad
  B^*\alz(\theta)
  =\{\Gamma(\delta)\tau^\alpha+1\} e^{-\tau}.
\]
Then
\begin{align} \nonumber
  B\alz(y,\theta)
  &\propto
  \tau e^{-\tau}[\Gamma(\delta)\tau^{-\delta}
    \gamma_{\delta,\tau}(y) + \tau^{-1} \gamma_{1,\tau}(y)]
  \\\label{e:B-factor}
  &\propto
  B^*\alz(\theta) \{[1-w^*\alz(\theta)]
  \gamma_{\delta,\tau}(y) + w^*\alz(\theta)\gamma_{1,\tau}(y)\}.
\end{align}
Consequently, under the normalized $B\alz(y, \theta)$, the marginal
\pdf of $\theta$ is the normalized $B^*\alz(\theta)$, and
conditioning on $\theta$, $y$ follows a mixture of Gamma \pdf's
\begin{align} \label{e:B-gamma-mix}
  \phi^*_\theta(y)
  =[1-w^*\alz(\theta)] \gamma_{\delta,\tau}(y) + w^*\alz(\theta)
  \gamma_{1,\tau}(y).
\end{align}
Since for all $\theta\in [0,\pi)$, $B^*\alz(\theta) \le \sup_{t>0}
[\Gamma(\delta) t^\alpha+1] e^{-t} \le \Gamma(\delta) +1$, 
the normalized $B^*\alz(\theta)$ can be sampled with $[\Gamma(\delta)
+1]\cf{0\le \theta<\pi}$ as the envelope.  Once $\theta$ is
sampled, $y$ is sampled from $\phi^*_\theta$.  The resulting
$(y,\theta)$ is a sample from the normalized $B\alz(y,\theta)$.

\begin{algorithm}
  \caption{Sampling from the normalized $\chi\alz$ using envelope
    $B\alz$.}
  \label{a:ar-B}
  \begin{algorithmic}[1]
    \Repeat \label{a:ar-B1}  \Comment{outer loop}      
    \Repeat \label{a:ar-B2}  \Comment{inner loop}
    \State Sample $\theta\sim\Dunif(0,\pi)$ and $U_1
    \sim\Dunif(0,1)$.
    \Until $U_1 [\Gamma(\delta)+1] \le B^*\alz(\theta)$.
    \label{a:ar-B3}
    \State \label{a:ar-B4}
    Sample $y\sim \phi^*_\theta$ and $U_2 \sim\Dunif(0,1)$,
    $R \gets \frac{(\delta/\alpha) y}{1-(y+1)^{-\delta/\alpha}}$.
    \Until $U_2 (c_{2,\alpha}/c_\alpha)(1+y^\alpha) \le R^\alpha$.
    \Comment{By \eqref{e:chi-B}.}

    \State \Return $(y,\theta)$.
  \end{algorithmic}
\end{algorithm}

Combining the sampling of the normalized $B\alz$ with \eqref{e:chi-B},
the normalized $\chi\alz$ can be sampled as in \Cref{a:ar-B}.
The algorithm has one non-elementary operation, which is the sampling
from $\phi^*_\theta$ on \cref{a:ar-B4}.  As noted in last subsection,
the complexity of the operation is uniformly bounded for $\theta$.
Thus, to analyze the complexity of \Cref{a:ar-B}, it suffices to
consider the expected total number of iterations.  By \Cref{p:HATH1},
in each iteration of the inner loop, the probability of acceptance is
\[
  \frac{\int^\pi_0 B^*\alz(\theta)\,\dd \theta}
  {\Gamma(\delta)+1}
  \ge
  \frac{\int^\pi_0 e^{-z\HAT}\,\dd \theta}{\Gamma(\delta)+1} 
  \ge\frac{C_z}{\Gamma(\delta)+1} \quad\text{with~}
  C_z = \int^\pi_0 e^{-z \HAT[\theta][1]}\,\dd\theta.
\]
On the other hand, in each iteration of the outer loop, given any
$\theta\in (0,\pi)$ generated at the end of the inner loop, by
\eqref{e:chi-B} and \eqref{e:B-factor}, the probability of
acceptance
\[
  \frac{\intzi \chi\alz(y,\theta)\,\dd y}
  {\intzi B\alz(y,\theta)\,\dd y}
  \ge
  \frac{c_{1,\alpha}}{c_{2,\alpha}},
\]
regardless of the values of $\theta$ and $z$.  Thus the expected total
number of iterations is $O(\frac{c_{2,\alpha}}{c_{1,\alpha}}
(\Gamma(\delta) + 1)/C_z) = O(\delta^{-2}/C_z)$.  Note that $C_z$ is
increasing in $z$.  Then given $z_0>0$, by using \Cref{a:ar-A,a:ar-B}
for $z> z_0$ and $0<z\le z_0$ respectively, the normalized
$\chi\alz(y,\theta)$ can be sampled with complexity $O(\delta^{-2})$
for all $z>0$.

\section{Preparatory AR related results} \label{s:ar}
In the rest of the paper, we develop one more algorithm, which has a
uniform complexity bound $O(|\ln \delta|)$ for all $(\alpha, z)$
close to $(1,0)$.  This section provides sampling routines for two
elementary distributions what will be heavily used in the algorithm.

Recall that if $\nu$ is a measure on $\Reals^n$ and $\phi: \Reals^n
\to\Reals^m$ is measurable, then the pushforward of $\nu$ under $\phi$
is a measure $\tilde\nu$ on $\Reals^m$ such that $\tilde\nu(A) = \nu(\phi^{-1}(A))$ for every measurable set $A$.  Denote $\tilde\nu
=\phi\#\nu$.  When $\nu$ is the law of a random variable $X$,
$\phi\#\nu$ is the law of $\phi(X)$.  If $f\ge0$ is a function on
$\Reals^n$ and $g\ge0$ a function on $\Reals^m$, such that $g(y)\,\dd
y$ is the pushforward of $f(x)\,\dd x$ under $x\to y = \phi(x)$, then
$g$ will be referred to as the pushforward of $f$ under $\phi$,
denoted $g=\phi\# f$.  If $n=m=1$ and $\phi$ has a piecewise
differentiable inverse $\psi$, then $\phi\#f(y) = f(\psi(y))
|\psi'(y)|$ and $\int \phi\#f = \int f$.

\subsection{A variation for log-concave functions} 
\label{ss:log-cv} 
Let $a(x)$ be a log-concave function of $x\ge0$, i.e., $\ln a(x)$ is
concave.  Suppose $a(0)= \sup_{x\ge0} a<\infty$.   In general, if
$\int a$ is unknown, then $f = a/\int a$ can be sampled using the
algorithm in \cite{devroye:12:spl}.  However, if it is known that
$\int a$ is no greater than a constant $A$, then there is an
alternative more amenable to analysis.  Since $f$ is a log-concave
\pdf, from p.~290 of \cite{devroye:86:sv-ny}, $f(x) \le f(0)(1\wedge
e^{1-f(0) x}) =2\Vrx_{f(0)}(x)$, where given $t>0$,
\begin{align} \label{e:g-t}
  \Vrx_t(x) = (t/2)(1\wedge e^{1-tx})\cf{x\ge0}
\end{align}
is the \pdf of $t^{-1}[\eta U + (1-\eta)(1+\xi)]$, where
$\eta\in\{0,1\}$ with $\pr\{\eta=0\}=1/2$, $U\sim \Dunif(0,1)$, and
$\xi\sim\Dexp(1)$, and the three random variables are independent.  A
more compact representation is $t^{-1}[\cf{U\le1} U + \cf{U>1} \ln\frac
e{2-U}]$ with $U\sim\Dunif(0,2)$.  Then
\begin{align} \label{e:log-concave}
  a(x) \le a(0)(1\wedge e^{1-a(0) x/\int a}) \le
  a(0)(1\wedge e^{1-a(0) x/A})=2 A\Vrx_s(x) \quad\text{with~}
  s = a(0)/A.
\end{align}
If $f$ is sampled using AR with $2 A\Vrx_s$ as the envelope, then the
probability of acceptance in each iteration is $\nth{2A} \int a$.  It
is useful to note that by \eqref{e:g-t}, $\Vrx_0(x)\equiv0$.

Given two constants $b>a\ge0$, let
\begin{align} \label{e:varphi-def}
  \varphi_{a,b,c}(x)=\varphi_c(x) \cf{a<x\le b} \quad
  \text{with~}   \varphi_c(x) = x^{c-1} e^x\cf{x>0}.
\end{align}
Note that the exponential term in $\varphi_{a,b,c}$ is $e^x$ instead
of $e^{-x}$.  The function appears in several \pdf's later.  To sample
from those \pdf's, we need an envelope $\varphi^*_{a,b,c}$ for
$\varphi_{a,b,c}$ that meets two criteria: 1) it has an explicit
integral whose ratio to $\int \varphi_{a,b,c}$ is uniformly bounded
for $0\le a<b<\infty$ and $c\in (0,1)$, and 2) the normalized
$\varphi^*_{a,b,c}$ is easy to sample.  To get $\varphi^*_{a,b,c}$, we
will divide $\varphi_{a,b,c}$ into two parts, such that one under a
power transform becomes log-concave, while the other is restricted to
a uniformly bounded interval.  For each part, we construct an
envelope.  Put together, the two envelopes form $\varphi^*_{a,b,c}$.

From the power series expansion of $e^x$, 
\begin{align} \label{e:varphi}
  \int \varphi_{a,b,c}=\int^b_a \varphi_c = 
  \sumzi k \int^b_a \frac{x^{k+c-1}}{k!}\,\dd x
  = \sumzi k u_k, \quad\text{where~} u_k = \frac{b^{k+c} -
    a^{k+c}}{(k+c)k!}.
\end{align}
Since $c\in(0,1)$, for $k\ge1$,
\begin{align} \label{e:varphi2}
  \frac{u_k}{\lfrac{(b^{k+c}-a^{k+c})}{(k+1)!}}
  = \frac{k+1}{k+c}
  \in (1,2).
\end{align}
Then
\begin{gather}\label{e:varphi2b}
  \frac{b^c - a^c}c+\frac {D(a,b,c)}2
  \le \int\varphi_{a,b,c}\le \frac{b^c - a^c}c+D(a,b,c):=M(a,b,c),
\end{gather}
where
\[
  D(a,b,c) = 2\sumoi k \frac{b^{k+c}-a^{k+c}}{(k+1)!}
  = 2[b^{c-1} (e^b - 1 - b) -a^{c-1} (e^a - 1 -a)].
\]
It is plain that $\int\varphi_{a,b,c}>M(a,b,c)/2$.  By
\eqref{e:varphi}, $\varphi_{a,b,c}\equiv0$ if $a\ge b$.  For
consistency, define $M(a,b,c)=0$ if $a\ge b$.

Given $p\in (0,1)$, the pushforward of $\varphi_c(x)$ under 
$x\to t = x^{1/p}$ is $\pi_p(t) = p t^{pc-1} \exp(t^p)$.  Since
$(\ln\pi_p)''(t) = t^{-2}[1-pc - (1-p) p t^p]$, $\pi_p(t)$ is
log-concave on $[d(p)^{1/p}, \infty)$, where
\[
  d(p)= \frac{1-c}{1-p} + \nth p.
\]
To make use of log-concavity as much as possible, one may minimize
$d(p)$.  It is seen that
\begin{align} \label{e:xpc}
  \inf_{p\in(0,1)} d(p) = (1+\sqrt{1-c})^2:=x_c, \quad
  \mathop{\arg\inf}_{p\in (0,1)} d(p) =
  \nth{1+\sqrt{1-c}} := p_c.
\end{align}
Let
\[
  A = b\wedge x_c, \quad B = a\vee x_c.
\]
Then $(a, A]$ and $(B, b]$ partition $(a,b]$, so $\varphi_{a, b, c} =
\varphi_{a, A,c} + \varphi_{B, b, c}$.  Note that
\[
  \varphi_{a,A,c}\not\equiv0\Iff x_c>a, \quad 
  \varphi_{B,b,c}\not\equiv0\Iff x_c<b.
\]
First, consider $\varphi_{B,b, c}$ provided $x_c < b$.  In this case,
$\tilde\varphi_c:= \pi_{p_c}$ is increasing and log-concave on
$I:=(B^{1/p_c}, b^{1/p_c}]$.  From \eqref{e:varphi2b}, the ratio of
$\int_I\tilde\varphi_c =  \int \varphi_{B,b,c}$ to $M(B, b, c)$ is in
$[1/2,1]$.  Let $s = \tilde\varphi_c(b^{1/p_c})/M(B, b, c)$.  Since
$\tilde\varphi_c(b^{1/p_c}-t)$ is log-concave and decreasing on
$[0, b^{1/p_c}-B^{1/p_c})$, then from \eqref{e:log-concave}, for all
$t\in\Reals$,
\begin{align} \label{e:varphi3}
  \tilde\varphi_c(t)\cf{B^{1/p_c}< t \le b^{1/p_c}}
  \le
  2 M(B, b,c) \Vrx_s(b^{1/p_c}-t).
\end{align}
As a result, the pushforward of $2  M(B, b,c) \Vrx_s(b^{1/p_c} - t)$
under the signed power transform $t\to x = \sign(t) |t|^{p_c}$ is an
envelope for $\varphi_{B, b, c}(x) = \varphi_c(x)\cf{B<x\le b}$.
Clearly, the integral of the envelope is $2 M(B,b,c)$.  Next,
consider $\varphi_{a, A,c}$ provided $x_c>a$.  For this
function, we shall construct an envelope with integral $2
M(a,A,c)$.  Given integer $n\ge0$,
\[
  r_n(x) := \frac{e^x -\sum^n_{k=0} x^k/k!}{x^{n+1}/(n+1)!}
  =
  \sumzi k \frac{(n+1)! x^k}{(n+1+k)!}
\]
is increasing on $(0,\infty)$.  As $A\le x_c\le4$, it follows that for
$x\in (a, A]\subset (0,4]$, 
\[
  \varphi_c(x) \le 
  \sum^n_{k=0} \frac{x^{k+c-1}}{k!} + \frac{r_n(4)}{(n+1)!} x^{n+c} :=
  f_n(x).
\]
We need the integral of $f_n$ on any interval $(s,t]
\subset (0,4)$ to be at most $2M(s,t,c)$.  From \eqref
{e:varphi}--\eqref{e:varphi2}, this holds if $r_n(4)\le
\lfrac{2(n+1)}{(n+2)}$.  It is not hard to show, either analytically
or numerically, that the latter is true for $n\ge7$.  Let  $f^*(x)$
be the normalized $f_7(x)\cf{a<x\le A}$.  Then by $\int^A_a f_7 \le 2
M(a, A, c)$, 
\begin{align} \label{e:varphi4}
  \varphi_{a, A, c}(x) \le 2M(a, A, c)f^*(x).
\end{align}
Furthermore, $f^*(x)$ is a finite mixture of simple \pdf's.  To see
this, let
\begin{align} \label{e:varphi4-u}
  u_k = \frac{A^{k+c} - a^{k+c}}{(k+c)k!},
  \quad
  \beta_k(x) = \frac{(k+c)x^k\cf{a<x\le A}}{A^{k+c} - a^{k+c}}.
\end{align}
It is seen that $x^{c-1} \beta_k(x)$ is the \pdf of $A U^{1/(k+c)}$
with $U\sim \Dunif((a/A)^{k+c},1)$ and
\begin{align} \label{e:varphi4-w}
  f^*(x) = \frac{x^{c-1} \sum^8_{i=0}w_i \beta_i(x)}{\sum^8_{i=0}
    w_i}, \text{~where~}
  w_k = u_k \text{~for~}k\le7, \ w_8 = r_7(4) u_8.
\end{align}

Combining \eqref{e:varphi3} and \eqref{e:varphi4}, define
\begin{align} \nonumber
  \varphi^*_{a,b,c}(x)
  &=
  2M(a, A,c) f^*(x)\cf{a\le x\le A}
  \\\label{e:varphi5}
  &\qquad+\frac{2M(B, b,c)}{p_c} |x|^{1/p_c-1}
  \Vrx_s(b^{1/p_c} - \sign(x) |x|^{1/p_c}),
\end{align}
where $s$ is the same as in \eqref{e:varphi3} and can be written as 
$\lfrac{p_c b^{c-1/p_c} e^b}{M(B, b,c)}$.  Then
\begin{align} \label{e:varphi*}
  \varphi_{a,b,c}(x) \le \varphi^*_{a,b,c}(x), \quad
  \int\varphi^*_{a,b,c} = 2 M(a,b,c) \le 4\int\varphi_{a,b,c}.
\end{align}
Furthermore,  $\varphi^*_{a,b,c}(x)$ can be sampled by
\Cref{a:varphi}.
\begin{algorithm}[t]
  \caption{Sampling from the normalized $\varphi^*_{a,b,c}$, $c\in
    (0,1)$, $0\le a < b<\infty$.}\label{a:varphi}
  \begin{algorithmic}[1]
    \State $A\gets b\wedge x_c$, $B\gets a\vee x_c$.
    \State Sample $i\in\{0,1\}$ with weights $M(a, A, c)$ and $M(B,
    b,c)$.
    \If{$i=0$}
    \State Sample $k\in\{0, \ldots, 8\}$ with weights $w_k$ and $X\sim
    x^{c-1} \beta_k(x)$.
    \Else
    \State Sample $Y\sim \Vrx_s$, where $s=p_c b^{c-1/p_c}
    e^b/M(B, b,c)$ and $\Vrx_s$ is defined in  \eqref{e:g-t}.
    \State $D=b^{1/p_c}-Y$,  $X \gets \text{sign}(D)|D|^{p_c}$.
    \EndIf
    \State \Return $X$.
  \end{algorithmic}
\end{algorithm}

\subsection{Envelope for a domain-restricted Gamma density} \label{ss:icgamma}
Given $b>a\ge0$, let
\begin{align} \label{e:gabc}
  g_{a,b,c}(x) = x^{c-1} e^{-x} \cf{a<x\le b}.
\end{align}
The normalized $g_{a,b,c}$ is a Gamma density restricted to $(a,b]$.
Similar to last subsection, we need an envelope $g^*_{a,b,c}$ for
$g_{a,b,c}$, such that the normalized $g^*_{a,b,c}$ is easy to 
sample and $\int g^*_{a,b,c}/\int g_{a,b,c}$ is uniformly bounded for
$c\in (0,2)$ and $0\le a < b<\infty$.  To this end, we rely on
log-concavity.  Note that to sample from an (unrestricted) Gamma
density, the most efficient AR algorithms do not rely on log-concavity
(\cite{devroye:86:sv-ny}, section XI.3.4).  However, the envelopes in
those algorithms become inefficient for sampling from a Gamma density
restricted to a finite interval outside its bulk.

\begin{lemma} \label{l:icgamma}
  Let $0\le a<b<\infty$ and $c\in (0,2)$.  Then
  \phref{p:l:icgamma}{Proof}
  \begin{gather} \label{e:icgamma}
    e^{-1} G(a,b,c) \le \int g_{a,b,c} \le G(a,b,c),
  \end{gather}
  where, if $b\in (a,a+2]$, then
  \[
    G(a,b,c) =
    \begin{cases}
      e^{-a} (b^c - a^c)/c &\text{if~} b\in (a,a+1]
      \\
      G(a,a+1,c) + G(a+1,b,c) &\text{if~} b\in (a+1,a+2],
    \end{cases}
  \]
  and if $b>a+2$, then
  \[
    G(a,b,c)
    =
    \begin{cases}
      G(a,a+2,c) + 2[(a+2)^{c-1} e^{-a-2} - b^{c-1} e^{-b}]
      &\text{if~} c\in [1,2),
      \\
      G(a,a+2,c) + (a+2)^{c-1} (e^{-a-2} - e^{-b})
      &\text{if~} c\in (0,1).
    \end{cases}
  \]
\end{lemma}

By \eqref{e:gabc}, $g_{a,b,c}\equiv0$ if $a\ge b$.  For consistency,
define $G(a,b,c)=0$ if $a\ge b$.  First let $c\in [1,2)$.  
Let
\[
  A=b\wedge (c-1), \quad B=a\vee (c-1).
\]
Then $(a, A]$ and $(B, b]$ partition $(a,b]$, so $g_{a,b,c} = g_{a,A,
  c} + g_{B,b,c}$.  Provided $A>a$, $g_{a,A,c}$ is log-concave and
increasing on $(a, A]$; provided $B<b$, $g_{B, b,c}$ is log-concave
and decreasing on $(B, b]$.  Define
\begin{gather}\label{e:icgamma1-env}
  g^*_{a,b,c}(x)
  =
  2 G(a, A,c) \Vrx_s(A -x) + 2 G(B,b,c) \Vrx_t(x-B)\cf{x>B}
  \\\nonumber
  \text{with}\quad
  s
  =
  \begin{cases}
    \frac{A^{c-1} e^{-A}}{G(a,A,c)}
    &\text{if~} A>a, \\
    0 & \text{else,}
  \end{cases}
  \quad
  t =
  \begin{cases} \frac{B^{c-1} e^{-B}}{G(B, b,c)}
    & \text{if~} B<b, \\
    0 & \text{else.}
  \end{cases}
\end{gather}
From \eqref{e:log-concave} and \Cref{l:icgamma},
\begin{align} \label{e:icgamma-env}
  g_{a,b,c}(x) \le
  g^*_{a,b,c}(x), \quad
  \int g^*_{a,b,c} = 2G(a,b,c) \le 2e \int
  g_{a,b,c}.
\end{align}

Now let $c\in (0,1)$.  Then $g_{a,b,c}(x)$ is not log-concave,
however, under $x\to t = x^c$, its pushforward $\tilde g_{a,b,c}(t) =
(1/c)\exp\{-t^{1/c}\} \cf{a^c < t \le b^c}$ is and also
decreases on $(a^c, b^c]$.  From \eqref{e:log-concave} and
\Cref{l:icgamma}, if $\tilde g^*_{a,b,c}(t) = 2 G(a,b,c) \Vrx_s(t -
a^c)$ with $s= e^{-a}/[c G(a,b,c)]$, then $\tilde g_{a,b,c}(t) \le
\tilde g^*_{a,b,c}(t)$ and $\int \tilde g^*_{a,b,c} \le 2 e \int
\tilde g_{a,b,c}$.  Thus, \eqref{e:icgamma-env} still holds by letting
\begin{align} \label{e:icgamma2-env}
  g^*_{a,b,c}(x) = 2 c G(a,b,c) x^{c-1} \Vrx_s(x^c - a^c),
  \quad\text{where~} s = \frac{e^{-a}}{c G(a,b,c)}.
\end{align}

Combining the above results, the normalized $g^*_{a,b,c}$ can be
sampled by \Cref {a:incomplete-gamma}.

\begin{algorithm}[t]
  \caption{Sampling from the normalized $g^*_{a,b,c}(x)$,
    $0<c<2$, $0\le a<b<\infty$.} \label{a:incomplete-gamma}
  \begin{algorithmic}[1]
    \If{$1\le c<2$}
    \State $A\gets b\wedge (c-1)$, $B\gets a\vee(c-1)$.
    \State Sample $i\in\{0,1\}$ with weights $G(a,A,c)$ and
    $G(B,b,c)$.
    \If{$i=0$}
    \State $s\gets \lfrac{A^{c-1} e^{-A}}{G(a, A, c)}$, sample $Z\sim
    \Vrx_s$, $X\gets A -Z$.
    \Else
    \State $t\gets \lfrac{B^{c-1} e^{-B}}{G(B,b,c)}$,  sample $Z\sim
    \Vrx_t$, $X\gets B+Z$.
    \EndIf
    \Else
    \State $s\gets e^{-a}/[c G(a,b,c)]$.
    \State Sample $Z\sim \Vrx_s$, $T\gets
    a^c + Z$, $X\gets T^{1/c}$. 
    \EndIf
    \State \Return $X$.
  \end{algorithmic}
\end{algorithm}

\section{Main routine of the third algorithm}
\label{s:large-alpha-small-z}
This section supplies the main routine of the third algorithm to
sample from the normalized $\chi\alz(y,\theta)$.  It actually works
for all $\alpha\in (0,1)$ and $z>0$.  The next section develops the
subroutine, which is the most technical part of the paper.  In both
the main routine and the subroutine, the expected number of iterations
is bounded.  However, the subroutine requires an $O(|\ln\delta|)$
overhead as well as $O(|\ln\delta|)$ number of operations within each
iteration.

Henceforth, denote
\begin{align} \label{e:ell}
  \ell(y) = \ln(1+y).
\end{align}
Since $\ell^{-1}(v) = e^v-1$, for any positive function $f(y)$,
under the mapping $y\to v=\ell(y)$,
\[
  \ell\#f(v) = f(e^v-1) e^v = f(y) (y+1)\quad\text{with~}
  y = e^v-1.
\]

\begin{prop} \label{p:chi-G}
  Put  \phref{p:p:chi-G}{Proof}
  \begin{gather} \label{e:m-r}
    m_t(y)
    = \frac{\cf{0<y<1/t}}{[\ell(y)]^\alpha}, \quad
    r_{t,1}(y)=\frac{\cf{y\ge1/t} e^{-ty}}{[\ell(1/t)]^\alpha}, \quad
    r_{t,2}(y)=\frac{\cf{y>0} e^{-ty}}{c_\alpha}.
  \end{gather}
  Define
  \begin{gather} \label{e:F-G}
    F_\alpha(y,t) = m_t(y) + r_{t,1}(y) + r_{t,2}(y),\quad
    G\alz(y,\theta)
    = c_\alpha e^\alpha\HAT e^{-z\HAT} F_\alpha(y,z\HAT).
  \end{gather}
  Then
  \begin{gather} \label{e:chi-G}
    \chi\alz(y,\theta)\le G\alz(y,\theta), \quad
    \inf_{z>0} \frac{\int \chi\alz}{\int G\alz}
    \ge \nth{2 e^{\alpha+1}} \intzi (y+1)^{-\alpha} e^{-y}\,\dd y.
  \end{gather}
\end{prop}

From \eqref{e:chi-G}, $G\alz$ is an envelope of $\chi\alz$.  However,
it is not easy to work with.  Note that given
$t>0$, from \eqref {e:varphi-def}, $\ell\# m_t(v) = v^{-\alpha} e^v
\cf{0<v<\ell(1/t)} = \varphi_{0,\ell(1/t), \delta}(v)$.  Denote by
$\ell\otimes\Id$ the mapping $(y,\theta)\to (v,\theta) = (\ell(y),
\theta)$ and note $(\ell\otimes\Id) f(v,\theta) = f(e^v-1, \theta)
e^v$ for any $f(y,\theta)\ge 0$.  As in \eqref{e:calpha}, let $\tau =
z\HAT$.  Then from \eqref{e:F-G},
\[
  (\ell\otimes\Id)\#G\alz(v,\theta)
  =
  \frac{c_\alpha e^\alpha}z \tau e^{-\tau}
  [\varphi_{0,\ell(1/\tau), \delta}(v) + \ell\#r_{\tau,1}(v) +
  \ell\#r_{\tau,2}(v)].
\]
Recall  $\varphi^*_{0,\ell(1/\tau), \delta}$ defined in
\eqref{e:varphi5}.  Let
\begin{align}\label{e:Palz}
  P\alz(v,\theta)=
  \frac{c_\alpha e^\alpha}z \tau e^{-\tau}
  [\varphi^*_{0,\ell(1/\tau), \delta}(v) + \ell\#r_{\tau,1}(v) +
  \ell\#r_{\tau, 2}(v)].
\end{align}
Then from \eqref{e:varphi*},
\[
  (\ell\otimes\Id)\#G\alz(v,\theta) \le P\alz(v,\theta), \quad
  \frac{\int(\ell\otimes\Id)\#G\alz}{\int P\alz}\ge\nth4.
\]

Combining \eqref{e:chi-G} with the above display,
\begin{align} \label{e:chi-P}
  \begin{split}
    (\ell\otimes\Id)\#\chi\alz(v,\theta)
    \le P\alz(v,\theta), \hspace{1.3cm}
    \\
    \inf_{z>0} \frac{\int (\ell\otimes\Id)\#\chi\alz}
    {\int P\alz}
    \ge \nth{8 e^{\alpha+1}} \intzi (y+1)^{-\alpha} e^{-y}\,\dd y.
  \end{split}
\end{align}
Thus, to sample from the normalized $\chi\alz$, one can sample
$(v,\theta)$ from the normalized $(\ell\otimes\Id)\#\chi\alz$ with
$P\alz$ as the envelope and return $(\ell\otimes \Id)^{-1}(v,\theta) = 
(e^v-1, \theta)$.  The expected number of iterations of the sampling
is bounded for $\alpha\in (0,1)$ and $z>0$.

The issue now is how to sample from the normalized $P\alz$.  From
\eqref{e:Palz}, if $(v,\theta)$ has the normalized $P\alz$ as the
joint \pdf, then conditioning on $\theta$, $v$ has the
normalized $\varphi^*_{0, \ell(1/\tau),\delta}(v) +
\ell\#r_{\tau,1}(v) + \ell\#r_{\tau, 2}(v)$ as its \pdf.  From
\eqref{e:varphi*},
\begin{align} \nonumber
  \int \varphi^*_{0,\ell(1/\tau),\delta}
  = 2 M(0,\ell(1/\tau), \delta)
  &=
  \frac{2[\ell(1/\tau)]^\delta}\delta + 
  \frac{4[e^{\ell(1/\tau)} - 1 - \ell(1/\tau)]}{
    [\ell(1/\tau)]^\alpha
  }
  \\\label{e:marginal0}
  &=
  \frac4{\tau [\ell(1/\tau)]^\alpha} + (2/\delta-4)
  [\ell(1/\tau)]^\delta := I_0(\theta)
\end{align}
and the normalized $\varphi^*_{0,\ell(1/\tau),\delta}(v)$ can be
sampled by \Cref{a:varphi}.  On the other hand,
\begin{align} \label{e:marginal2}
  \int r_{\tau,1}
  = \nth{e\tau [\ell(1/\tau)]^\alpha} := I_1(\theta), \quad
  \int r_{\tau,2}
  = (c_\alpha \tau)^{-1} := I_2(\theta),
\end{align}
and the normalized $r_{\tau,1}$ and $r_{\tau,2}$ are the \pdf's of
$(\xi+1)/\tau$ and $\xi/\tau$, respectively, where $\xi\sim
\Dexp(1)$.  As a result, conditioning on $\theta$, $v$ can be sampled
efficiently.  Next, the marginal \pdf of $\theta$ is the normalized
$\int P\alz(v,\theta)\,\dd v$.  From
\eqref{e:Palz}--\eqref{e:marginal2},
\begin{align*}
  \int P\alz(v,\theta)\,\dd v
  &= \frac{c_\alpha e^\alpha} z \tau e^{-\tau}
  [I_0(\theta) + I_1(\theta) + I_2(\theta)]
  \\
  &=\frac{e^\alpha} z
  e^{-\tau}
  \Cbr{2c_\alpha(1/\delta-2)\tau [\ell(1/\tau)]^\delta +
    \frac{c_\alpha(4+1/e)}{[\ell(1/\tau)]^\alpha} + 1
  }.
\end{align*}
Define 
\begin{gather} \label{e:kappa-psi-alpha}
  \begin{split}
    \kappa_{\alpha,1} = 2c_\alpha(1/\delta-2), \quad
    \kappa_{\alpha,2} = c_\alpha (4+1/e),
    \\
    \psi_\alpha(t)=
    \kappa_{\alpha,1}t [\ell(1/t)]^\delta
    +
    \kappa_{\alpha,2}[\ell(1/t)]^{-\alpha} +1.
  \end{split}
\end{gather}
Then $\int P\alz(v,\theta)\,\dd v = z^{-1}e^\alpha Q\alz(\theta)$,
where 
\begin{align} \label{e:Dpsi2}
  Q\alz(\theta)=\psi_\alpha(z\HAT) e^{-z\HAT}.
\end{align}
The sampling from the normalized $Q\alz$ is the focus of
\cref{s:subroutine}.  For now, assume that a subroutine to sample
from the normalized $Q\alz$ is available.  Then the normalized
$P\alz$, and hence $\chi\alz$ can be sampled as in \Cref{a:ell-P}.

\begin{algorithm}[t]
  \caption{Sampling from the normalized $\chi\alz$ using mapping
    $\ell$ and envelope $P\alz$.} \label{a:ell-P}
  \begin{algorithmic}[1]
    \Repeat
    \State Sample $\theta\sim$ normalized $Q\alz(\theta)$ and
    $U\sim \Dunif(0,1)$, $\tau\gets z\HAT$.
    \label{a:ell-P1}
    \State Sample $i\in\{0,1,2\}$ with weights $I_i(\theta)$ defined
    in \eqref{e:marginal0}--\eqref{e:marginal2}.
    \If{$i=0$}
    \State Sample $v\sim$ normalized $\varphi^*_{0, \ell(1/\tau),
      \delta}(v)$ by \Cref{a:varphi}.  \label{a:ell-P5}
    \Else
    \State Sample $\xi\sim\Dexp(1)$, $v\gets
    \ell((\xi+2-i)/\tau)$.
    \EndIf
    \Until $v>0$ and $U\le\chi\alz(e^v-1,\theta)e^v/P\alz(v,
    \theta)$.  \label{a:ell-P9}
    \Comment{By \eqref{e:chi-P}.}
    \State \Return $(e^v-1,\theta)$. 
  \end{algorithmic}
\end{algorithm}

\def\nal{N}
\section{Subroutine of the third algorithm}
\label{s:subroutine}
This section deals with the sampling from the normalized 
$Q\alz(\theta)$ in \eqref{e:Dpsi2}.  The final result is
\Cref{a:qalz}.  Henceforth, let $0<z\le1$.

The first step is to find a suitable proxy for $\HAT$.  Fix parameters
\begin{align} \label{e:pars}
  \textstyle
  \Delta\in (0,1), \quad \alpha_0 \in (\nth2,1),
  \quad
  \theta_0 \in
  (\frac\pi3 + \frac\pi{3\alpha_0}, \pi).
\end{align}
Henceforth, in the expression $a=O(b)$, the implicit constant depends
on the fixed parameters $\Delta$, $\alpha_0$, and $\theta_0$.  As
before, $\delta = 1-\alpha$.

First consider $\HAT$ on $[0, \theta_0)$.  Fix $0=\seqop[0]t < m <  t_{m+1}=\theta_0$ such that
\begin{align} \label{e:m-al}
  m = O(1), \quad \HAT[t_i][1]\le(1+\Delta) \HAT[t_{i-1}][1],
  \quad i=1, \ldots, m+1.
\end{align}
This can be done as follows.  By \eqref{e:HAT1}, $\HAT[\theta][1]$ is
log-convex and $(\ln H_1)'(\theta)  = f(\theta)/\theta$, where
$f(\theta)=1+ [\sinc(\theta)]^{-2} - \lfrac{2\cos\theta}
{\sinc(\theta)}$.  Let $T_0=\theta_0$, $T_{i+1} = T_i[1 -
\lfrac{\ln (1+\Delta)}{f(T_i)}]_+$, $i\ge0$, and $m=\min\{i\ge0:
T_{i+1}=0\}$.   For $i\le m$,
\[
  0\le \ln \frac{\HAT[T_i][1]}{\HAT[T_{i+1}][1]}
  \le (\ln H_1)'(T_i) (T_i - T_{i+1}) \le \frac{f(T_i)}{T_i} \frac{T_i
    \ln(1+\Delta)}{f(T_i)} = \ln (1+\Delta).
\]
Then $t_i = T_{m+1-i}$, $0\le i\le m+1$, satisfy \eqref{e:m-al}.

From \Cref{p:HATH1}, $\HAT[t_i] < (1+\Delta) \HAT[t_{i-1}]$ for
$\alpha\in [\nth2,1)$.  Define 
\begin{align} \label{e:JAT-lower}
  \JAT = \sum^m_{i=0} \cf{\theta\in [t_i, t_{i+1})} \HAT[t_i], \quad
  \theta\in [0,\theta_0).
\end{align}
Since $\HAT$ is an increasing function, then for $\theta\in
[0,\theta_0)$, 
\begin{align} \label{e:JH}
  \JAT \le \HAT \le (1+\Delta) \JAT.
\end{align}
Since $\JAT$ is a piecewise constant function on $[0,\theta_0)$, it is
easy to handle.  The next step is to extend $\JAT$ to $[\theta_0,
\pi)$, so that \eqref{e:JH} holds on the entire $[0,\pi)$.  The
construction of $\JAT$ on $[\theta_0, \pi)$ relies on the following
properties of $\XIAT$ in \eqref{e:KAT}.

\begin{prop} \label{p:nodes}
  Fix $\alpha_0\le\alpha<1$.  Let $C = -\pi[\ln\sinc(x)]''_{x=\pi
    - \alpha_0\theta_0}$.  Define
  \phref{p:p:nodes}{Proof}
  \[
    f_\alpha(\theta)
    = \sup\Cbr{\theta'\le\pi: \frac{\theta'/(\pi -
        \alpha\theta')} {\theta/(\pi - \alpha\theta)} \le 1+\Delta/2,
      \quad 
      \theta'\le \theta +\nth C\ln\frac{1+\Delta}{1+\Delta/2}
    }.
  \]
  Define $\theta_i = f_\alpha(\theta_{i-1})$ for $i\ge1$.  Let $\nal
  =\#\{i\ge0: \theta_i\in (0, \pi)\}$.  Then $\nal =
  O(\ln(1/\delta))$, $\theta_0 < \seqop\theta < \nal <
  \theta_{\nal+1}:=\pi$, and
  \begin{align} \label{e:K-nodes}
    1\le \frac{\XIAT[\theta_{i+1}]}{\XIAT[\theta_i]} \le1+\Delta,
    \quad 0\le i\le\nal.
  \end{align}
\end{prop}

Now for $\theta\in [\theta_0, \pi)$, define
\begin{align} \label{e:JAT-upper}
  \JAT
  =
  \Sbr{\sum^\nal_{i=0} \cf{\theta\in [\theta_i, \theta_{i+1})}
    \XIAT[\theta_i]
  }\Sbr{1+\frac{\delta\pi}{\alpha(\pi-\theta)}}^{1/\delta}.
\end{align}
Then $\JAT$ is increasing and right continuous on $[0,\pi)$.  From
\Cref{p:H-sand} and  \eqref{e:K-nodes}, \eqref{e:JH} also holds for
$\theta\in(\theta_0, \pi)$.  Then \eqref{e:JH} holds on $[0,\pi)$.

It is easy to see that $\JAT\ge1$, $\JAT[t_i]=\HAT[t_i]$ for $0\le
i\le m$, and $\JAT[\theta_i] = \HAT[\theta_i]$ for $0\le i\le \nal$.  
Let
\begin{align} \label{e:theta-*-z}
  \vartheta_z = \inf\{\theta: \JAT\ge1/z\}.
\end{align}
Since $0<z\le1$, the following are all the possibilities:
\begin{align*}
  1\le 1/z\le\HAT[\theta_0]
  \implies&
  \vartheta_z = t_i\ \text{with~} i = \min\{k\le m+1:
  \HAT[t_k]\ge1/z\} \\
  1/z>\HAT[\theta_0]
  \implies&
  \vartheta_z = 
    \Grp{
      \pi-\frac{\lfrac{\delta\pi}\alpha}{
      [1/(z\XIAT[\theta_i])]^\delta -1}
  }\wedge \theta_{i+1}\ \text{with~}
  \\&i = \max\{k:
  \HAT[\theta_k]\le 1/z\}
\end{align*}

We apply \Cref{a:AR} to sample from the normalized $Q\alz$.  To do
this, we partition $[0, \pi)$ into several intervals, which include
$[\vartheta_z, \pi)$ and multiple intervals contained in $[0,
\vartheta_z)$.  For each of the intervals $S$, we construct an
envelope $Q^*_S$ for $Q\alz(\theta) \cf{\theta\in S}$.  The
construction on $[\vartheta_z,\pi)$ is relatively simple so is
considered first.   \phref{p:p:N-J-sup}{Proof}
\begin{prop} \label{p:N-J-sup}
  Denote $\kappa_{\alpha,3}=\kappa_{\alpha,1} + (\ln 2)^{-\alpha} 
  \kappa_{\alpha,2}+1$ and $\kappa_{\alpha,4}=\kappa_{\alpha,1}
  (\ln2)^\delta + \kappa_{\alpha,2}$.  Let
  \[
    \omega_z = \frac{2\kappa_{\alpha,3}}{\tau_z\varrho_z}, \quad
    s_z = \varrho_z \tau^{\alpha+1}_z e^{-\tau_z}.
  \]
  where $\tau_z = z\HAT[\vartheta_z]$ and $\varrho_z = (\ln
  H_\alpha)'(\vartheta_z)$.  Then for $\theta\in [\vartheta_z, \pi)$,
  \begin{align} \label{e:z-star}
    Q\alz(\theta) \le 
    \omega_z\Vrx_{s_z}(\theta - \vartheta_z).
  \end{align}
  Furthermore,
  \begin{align}\label{e:K-sup}
    \omega_z
    \le
    \frac{2\kappa_{\alpha,3}/\kappa_{\alpha,4}}
    {(z+1) e^{-z} - 2/e} \int^{\vartheta_z}_0 Q\alz. 
  \end{align}
  Finally, for all $\alpha\ge\alpha_0$,
  \[
    \frac{\kappa_{\alpha,3}}{\kappa_{\alpha,4}} \le
    (\ln 2)^{\alpha-1} + \frac{2+e^{-1}}{\ln
      2}\frac{1-\alpha}{2\alpha-1} +
    \Grp{\frac{1-\alpha}\alpha}^\alpha
    \frac{1-\alpha}{2\alpha-1},
  \]
  and the \rhs is a decreasing function of $\alpha\in[\alpha_0, 1)$.
\end{prop}

From \Cref{p:N-J-sup}, $\omega_z\Vrx_{s_z}(\theta - \vartheta_z)$ will
be used as the envelope for $Q\alz$ on $[\vartheta_z, \pi)$.  To
construct the envelops on $[0,\vartheta_z)$, the following result is the
starting point.
\begin{prop} \label{p:N-J-sandwich} 
  For $\theta\in [0, \vartheta_z)$,  \phref{p:p:N-J-sandwich}{Proof} 
  \begin{align}\label{e:psi-phi} 
    e^{-1-\Delta} \le
    \frac{Q\alz(\theta)}{\psi_\alpha(z\JAT)}
    \le 1+\Delta.
  \end{align}
\end{prop}

By \Cref{p:N-J-sandwich}, the envelopes on $[0,\vartheta_z)$ will be
made from $\psi_\alpha(z\JAT)$.  On $[0,\vartheta_z)\cap
[0,\theta_0) = [0, \vartheta_z\wedge\theta_0)$, as $\JAT$ is piecewise
constant, from \eqref{e:psi-phi}, each constant piece of $(1+\Delta)
\psi_\alpha(z\JAT)$ is used as an envelope.  The construction on
$[0,\vartheta_z) \setminus [0, \theta_0) = [\theta_0, \vartheta_z)$ is
more involved.  Clearly, one only has to consider it when $\theta_0<  \vartheta_z$, or equivalently $z<1/\HAT[\theta_0]$.
In this case, 
\begin{align} \label{e:I_n}
  I_n = I_{n,z} :=[\theta_n \wedge \vartheta_z, \theta_{n+1}\wedge
  \vartheta_z), \quad 0\le n\le\nal,
\end{align}
form a partition of $[\theta_0, \vartheta_z)$, and for each $n$, $I_n
\ne \emptyset\Iff\theta_n<\vartheta_z$.  Define
\begin{align} \label{e:lambda_z}
  \laz(\theta)=
  \ell\Grp{\nth{z\JAT}}.
\end{align}
Then $\laz$ is strictly decreasing on $[\theta_0, \vartheta_z)$ and
smooth on each $I_n\ne\emptyset$, and $\laz([\theta_0, \vartheta_z))$ is
partitioned by $\laz(I_n)=(c_n, d_n]$, $I_n\ne\emptyset$, where $c_n =
c_{n,z}$ and $d_n = d_{n,z}$ are defined by
\begin{align} \label{e:lambda-I}
  \begin{split}
    c_n 
    =\laz((\theta_{n+1}\wedge\vartheta_z)-))
    &=
    \begin{cases}
      \displaystyle
      \ell\Grp{
        {\nth{z\XIAT[\theta_n]}
          \Sbr{1+\frac{\delta\pi}{\alpha(\pi-\theta_{n+1})}}^{-1/\delta}
        }
      }
      &\text{if~} \theta_{n+1} \le\vartheta_z \\
      \ln2 & \text{else,}
    \end{cases}
    \\
    d_n= \laz(\theta_n\wedge \vartheta_z)
    &=
    \ell\Grp{
      {\nth{z\XIAT[\theta_n]}
        \Sbr{1+\frac{\delta\pi}{\alpha(\pi-\theta_n)}}^{-1/\delta}
      }
    }.
  \end{split}
\end{align}
\begin{lemma}  \label{l:N-J-mid}
  Fix $\alpha\in[\alpha_0, 1)$ and $z<1/\HAT[\theta_0]$.
  Given $0\le n\le \nal$ with $I_n\ne\emptyset$, let
  \phref{p:l:N-J-mid}{Proof}
  \begin{align} \label{e:mid-range-consts}
    a_n = a_{n,z}
    = z\JAT[(\theta_{n+1}\wedge\vartheta_z)-] = (e^{c_n}-1)^{-1},
    \quad
    b_n = \frac{\delta\pi}{\pi - \alpha\theta_n}.
  \end{align}
  Define function $P_n(t) = P_{n,z}(t)$ on $\laz(I_n)$ as well as
  constants $\pi_n = \pi_{n,z}$ and $v_n = v_{n,z}$ as follows.  If
  $\theta_n\le (1 -\delta/\alpha)\pi$, then
  \begin{gather} \label{e:mid-range-P1}
    \begin{split}
      \pi_n = 
      \frac{(1+a_n)^2\alpha} {\pi[z\XIAT[\theta_n]]^\delta}
      &(\pi - \theta_n)^2, \quad
      v_n = (1-\Delta)/[(1+\Delta)(1+a_n)^{\alpha+3}], \\
      P_n(t)
      &=
      \pi_n[(1+a_n)\kappa_{\alpha,1} t^\delta e^{-t}  +
      \kappa_{\alpha,2} t^{-\alpha} + 1] e^{-\delta t},
    \end{split}
  \end{gather}
  where $\kappa_{\alpha,1}$ and $\kappa_{\alpha,2}$ are given in
  \eqref {e:kappa-psi-alpha}.  If $\theta_n>(1 - \delta/\alpha)\pi$,
  then
  \begin{align} \label{e:mid-range-P2}
    \begin{split}
      \pi_n = \frac{(1+a_n)^2\delta^2\pi}{\alpha b^2_n}
      &[z\XIAT[\theta_n]]^\delta, \quad
      v_n = b^2_n/[2^{1+\delta}(1+a_n)^2],
      \\
      P_n(t)
      &=
      \pi_n[\kappa_{\alpha,1} t^\delta e^{-t} +\kappa_{\alpha,2}
      t^{-\alpha} + 1] e^{\delta t}.
    \end{split}
  \end{align}
  Then
  \begin{align} \label{e:Q-P}
    e^{-1-\Delta} v_n
    \le \frac{Q\alz(\theta)}{\laz^{-1}\# P_n(\theta)} \le 1+\Delta,
    \quad\theta\in I_n.
  \end{align}
\end{lemma}

\begin{remark}
  From \eqref{e:mid-range-consts}, $a_n\le1$, and if $\theta_n >
  (1-\delta/\alpha)\pi$, then $1/2<b_n<1$.  As a result,
  \[
    v_n\ge \frac{1-\Delta}{2^{\alpha+3}(1+\Delta)} \wedge 2^{-5-\delta}.
  \]
\end{remark}

Although $\laz^{-1}\# P_n$ in \Cref{l:N-J-mid} is an envelope of
$Q\alz$ on $I_n\ne\emptyset$, it is not easy to directly work with.
Note that $P_n$ consists of functions of the form
$\varphi_{a,b,c}$ in \eqref{e:varphi-def} or $g_{a,b,c}$ in
\eqref{e:gabc} as well as exponential functions restricted to a
bounded interval.  From \cref{s:ar}, a suitable envelope can be
obtained on the basis of $\laz^{-1}\# P_n$.  \phref{p:p:N-J-mid2}{Proof}
\begin{prop}  \label{p:N-J-mid2}
  Fix $\alpha\in [\alpha_0, 1)$.  Given $z>0$, let $I_n = I_{n,z}$ as
  in \eqref{e:I_n}.  For $0\le n\le\nal$, define $P^*_n(\theta) = 
  P^*_{n,z}(\theta)$ on $I_n$ as follows.  If $I_n=\emptyset$, then
  $P^*_n(\theta)\equiv 0$, otherwise, 
  \[
    P^*_n(t)
    =
    \pi_n [P^*_{n,1}(t) + P^*_{n,2}(t)+P^*_{n,3}(t)],
  \]
  where, if $\theta_n\le (1 -\delta/\alpha)\pi$, then   
  \begin{align*}
    P^*_{n,1}(t)&=
    (1+\delta)^{-\delta}(1+a_n)\kappa_{\alpha,1}
    g^*_{(1+\delta)c_n, (1+\delta) d_n, 1+\delta}((1+\delta) t),
    \\
    P^*_{n,2}(t)&=
    \kappa_{\alpha,2} \delta^\alpha g^*_{\delta c_n, \delta d_n, 
      \delta}(\delta t),
    \quad
    P^*_{n,3}(t)= e^{-\delta t}\cf{c_n < t\le d_n},
  \end{align*}
  and if $\theta_n>(1 - \delta/\alpha)\pi$, then
  \begin{align*}
    P^*_{n,1}(t)&=
    \alpha^{-\delta}\kappa_{\alpha,1}
    g^*_{\alpha c_n, \alpha d_n, 1+\delta}(\alpha t), 
    \\
    P^*_{n,2}(t)&=
    \kappa_{\alpha,2} \delta^\alpha \varphi^*_{\delta c_n,
      \delta d_n, \delta}(\delta t),
    \quad
    P^*_{n,3}(t)= e^{\delta t}\cf{c_n < t\le d_n}.
  \end{align*}
  Then
  \begin{gather}\label{e:mid-range-P-env}
    \begin{split}
      Q\alz(\theta)
      &\cf{\theta\in I_n} \le (1+\Delta)\laz^{-1}\#
      P^*_n(\theta) \quad\text{and}\\
      &\frac{e^{-2-\Delta} v_n}{2}
      \int \laz^{-1}\# P^*_n \le \int_{I_n} Q\alz.
    \end{split}
  \end{gather}
\end{prop}

From \Cref{p:N-J-sandwich,p:N-J-mid2,p:N-J-sup}, the normalized
$Q\alz$ can be sampled as follows.  First, partition $[0,\pi)$
into the following intervals arranged from left to right, 
\begin{gather*} 
  D_0 = D_{0,z} = [t_0\wedge\vartheta_z, t_1\wedge\vartheta_z), \ldots,
  D_m = D_{m,z} = [t_m\wedge \vartheta_z, t_{m+1}\wedge\vartheta_z),\\
  I_0 = I_{0,z}=
  [\theta_0\wedge\vartheta_z, \theta_1\wedge\vartheta_z),\ldots,
  I_\nal = I_{\nal,z} = [\theta_\nal\wedge\vartheta_z, \theta_{\nal+1}
  \wedge\vartheta_z),\
  E=E_z = [\vartheta_z,\pi).
\end{gather*}
Recall that $t_0=0$, $t_{m+1}=\theta_0$, and $\theta_{\nal + 1}=\pi$.
Put $\Cal I = \{E, \eno[0] D m, \eno[0] I \nal\}$.  
For each $S\in \Cal I$, use $Q^*_S(\theta)$ as the envelope of
$Q\alz(\theta)\cf{\theta \in S}$, where
\[
  Q^*_S(\theta)
  =
  \begin{cases}
    \omega_z \Vrx_{s_z}(\theta - \vartheta_z)
    &
    \text{if~} S = E = [\vartheta_z,\pi)
    \\
    (1+\Delta)\psi_\alpha(z\HAT[t_i]) \cf{\theta\in D_i},
    &
    \text{if~} S = D_i = [t_i\wedge\vartheta_z, t_{i+1}\wedge\vartheta_z)
    \\
    (1+\Delta)\laz^{-1}\#P^*_n(\theta),
    &
    \text{if~} S = I_n = [\theta_n\wedge\vartheta_z, \theta_{n+1}\wedge
    \vartheta_z).
  \end{cases}
\]
From \Cref{p:N-J-sandwich,p:N-J-mid2}, 
\[
  \Cbr{\frac{e^{-2-\Delta}}{2(1+\Delta)}\min_{n\le N} v_n}
  \int \sum_{S\ne E} Q^*_S\le \int^{\vartheta_z}_0 Q\alz .
\]
Then from \Cref{p:N-J-sup} and the above bounds,
\[
  \frac{\int Q\alz}{\int\sum_S Q^*_S}
  \ge
  \frac{\int^{\vartheta_z}_0 Q\alz}{
    \int \sum_{S\ne E} Q^*_S + \omega_z
  }
  \ge
  \Sbr{
    \frac{2(1+\Delta)e^{2+\Delta}}{\min_{n\le N}v_n}
    +\frac{2\kappa_{\alpha,3}/\kappa_{\alpha,4}}
    {(z+1) e^{-z} - 2/e}
  }^{-1}.
\]
Given $z_0\in(0,1)$, the \rhs is uniformly bounded below from 0 for
$\alpha\ge\alpha_0$ and $z\le z_0$.  Therefore, if
\Cref{a:AR} is applied to $Q\alz(\theta) =
\sum_{S\in \Cal I} Q\alz(\theta)\cf{\theta\in S}$ with $Q^*_S(\theta)$
as the envelope for $Q\alz(\theta)\cf{\theta\in S}$, then the expected
number of iterations is uniformly bounded.  Note that some $S\in\Cal
I$ may be empty and while $Q\alz(\theta)\cf{\theta\in S}$ have bounded
and disjoint supports, $Q^*_S$ may have unbounded and overlapping
supports.

To use \Cref{a:AR}, we need to know $\int Q^*_S$ and to sample from
the normalized $Q^*_S$ for each $S\in\Cal I$.  First, from
\Cref{p:N-J-sup},
\[
\int Q^*_E = \omega_z
\]
and the normalized $Q^*_E$ is the \pdf of $\vartheta_z + Z$ with $Z\sim
\Vrx_{s_z}$.  Second, from \eqref {e:JAT-lower}, for each $0\le i\le
m$,
\[
  \int Q^*_{D_i}  = (1+\Delta)\psi_\alpha(z\HAT[t_i])
  (t_{i+1}\wedge\vartheta_z - t_i\wedge\vartheta_z),
\]
and if $D_i\ne\emptyset$, the normalized $Q^*_{D_i}$ is the \pdf
of $\Dunif(D_i)$.  Note that $D_i\ne\emptyset\Iff
t_i < \vartheta_z$, in which case $D_i = [t_i, t_{i+1}\wedge
\vartheta_z)$.  Third, for each $0\le n\le\nal$, if $I_n=\emptyset$,
then $Q^*_{I_n}\equiv0$, and if $I_n\ne\emptyset$,
then by \Cref{p:N-J-mid2}, 
\[
  \int Q^*_{I_n} = (1+\Delta) \int P^*_n =
  (1+\Delta)\pi_n \int (P^*_{n,1} +
  P^*_{n,2} + P^*_{n,3}).
\]
From \eqref{e:varphi*} and \eqref{e:icgamma-env}, if $\theta_n \le(1 -
\delta/\alpha)\pi$,
\begin{align*}
  \int P^*_{n,1}
  &=
  2(1+\delta)^{-1-\delta}(1+a_n)\kappa_{\alpha,1}
  G((1+\delta)c_n, (1+\delta) d_n, 1+\delta), \\
  \int P^*_{n,2}
  &=
  2\kappa_{\alpha,2} \delta^{-\delta} G(\delta c_n, \delta d_n, 
  \delta), \quad
  \int P^*_{n,3}= \delta^{-1}(e^{-\delta c_n} -
  e^{-\delta d_n})
\end{align*}
and if $\theta_n>(1-\delta/\alpha)\pi$, then
\begin{align*}
  \int P^*_{n,1}
  &=
  2\alpha^{-1-\delta}\kappa_{\alpha,1}
  G(\alpha c_n, \alpha d_n, 1+\delta)\\
  \int P^*_{n,2}
  &=
  2\kappa_{\alpha,2} \delta^{-\delta} M(\delta c_n, \delta d_n, 
  \delta), \quad
  \int P^*_{n,3}= \delta^{-1}(e^{\delta d_n} - e^{\delta c_n}).
\end{align*}

Note that $\int Q^*_{I_n}>0\Iff I_n\ne\emptyset\Iff \theta_n <
\vartheta_z$, in which case $I_n = [\theta_n, \theta_{n+1}\wedge
\vartheta_z)$.  For $I_n\ne\emptyset$, the normalized $Q^*_{I_n}$ is
sampled via $\theta = \laz^{-1}(t)$ with $t$ sampled from the
normalized $P^*_n$.  The latter is a mixture of the normalized
$P^*_{n,i}$.  When $\theta_n\le (1-\delta/\alpha)\pi$, the normalized
$P^*_{n,1}(t)$ is the pushforward of the normalized $g^*_{(1+
  \delta)c_n, (1+\delta)d_n,   1+\delta}(x)$ under the scaling $x\to
t=x/(1+\delta)$ and likewise for the normalized $P^*_{n,2}(t)$, while
the normalized $P^*_{n,3}(t) = e^{-\delta t}\cf{c_n< t\le d_n}$ is the
\pdf of $-\ln X/\delta$ with $X\sim \Dunif(e^{- \delta  d_n},
e^{-\delta c_n})$.  The case where $\theta_n>(1- \delta/\alpha)\pi$ is
similar, except that now the normalized $P^*_{n,2}(t)$ is the
pushforward of the normalized $\varphi^*_{\delta c_n, \delta
  d_n, \delta}(x)$ under the scaling $x\to t=x/\delta$, and the
normalized $P^*_{n,3}(t) = e^{\delta t} \cf{c_n < t\le d_n}$ is the
\pdf of $\ln X/\delta$ with $X\sim \Dunif(e^{\delta c_n}, e^{\delta
  d_n})$.  Using these facts, the normalized $P^*_n(t)$ can be sampled
by \Cref{a:Penv}.
\begin{algorithm}[t]
  \caption{Sampling from the normalized $P^*_n(t)$} \label{a:Penv}
  \begin{algorithmic}[1]
    \State Sample $i\in \{1,2,3\}$ with weights $\int^{d_n}_{c_n}
    P^*_{n,i}$, $i=1,2,3$.    
    \If{$\theta_n\le (1-\delta/\alpha)\pi$}
    \If{$i=1$}
    \State Sample $X\sim$ normalized $g^*_{(1+\delta)c_n,
      (1+\delta)d_n, 1+\delta}$ by \Cref{a:incomplete-gamma}.  $T\gets
    \lfrac X{(1+\delta)}$.
    \ElsIf{$i=2$}
    \State Sample $X\sim$ normalized $g^*_{\delta c_n, \delta d_n,
      \delta}$ by \Cref{a:incomplete-gamma}.  $T\gets
    \lfrac X\delta$.
    \Else
    \State Sample $X\sim\Dunif(e^{-\delta d_n}, e^{-\delta c_n})$.
    $T\gets -\lfrac{\ln X}\delta$.
    \EndIf
    \Else
    \If{$i=1$}
    \State Sample $X\sim$ normalized $g^*_{\alpha c_n, \alpha d_n,
      1+\delta}$ by \Cref{a:incomplete-gamma}.  $T\gets \lfrac
    X\alpha$.
    \ElsIf{$i=2$}
    \State Sample $X\sim$ normalized $\varphi^*_{\delta c_n, \delta
      d_n, \delta}$ by \Cref{a:varphi}.    $T\gets \lfrac X\delta$.
    \Else
    \State Sample $X\sim\Dunif(e^{\delta c_n}, e^{\delta d_n})$.
    $T\gets \lfrac{\ln X}\delta$.
    \EndIf
    \EndIf
    \State {\bf return} $T$.
  \end{algorithmic}
\end{algorithm}
Once $t$ is sampled, $\theta = \laz^{-1}(t)$ and $\laz^{-1}
\#P^*_n(\theta) = P^*_n(t) |\lambda'_z(\theta)|$ can be
evaluated using the following.
\begin{lemma} \label{l:lambda-inv}
  Let $t\in (c_n, d_n]$ and $\theta = \laz^{-1}(t)$.  Then
  \phref{p:l:lambda-inv}{Proof}
  \begin{align} \label{e:lambda-inv}
    \theta=\pi-\frac{\delta\pi/\alpha}
    {[z\XIAT[\theta_n](e^t-1)]^{-\delta} - 1},
    \quad
    \laz'(\theta)=
    - \frac{\pi(1-e^{-t})}{(\pi-\theta)(\pi - \alpha\theta)}.
  \end{align}
\end{lemma}

\begin{algorithm}
  \caption{Sampling from the normalized $Q\alz$, $\alpha\ge\alpha_0$,
    $z\le1$.} \label{a:qalz}
  \begin{algorithmic}[1]
    \Require $\Delta$, $\alpha_0$, and $\theta_0$ defined in
    \eqref{e:pars}. 
    \Repeat
    \State Sample $S\in\Cal I=\{E, \eno[0] Dm, \eno[0]I\nal\}$ with
    weights $\int Q^*_A$, $A\in \Cal I$.
    \State Sample $U\sim\Dunif(0,1)$,  $r\gets\infty$.
    \If{$S=E = [\vartheta_z, \pi)$}
    \State Sample $Z\sim \Vrx_{s_z}$, $\theta\gets \vartheta_z + Z$.
    \If{$\theta<\pi$}
    \State $r\gets U\omega_z \Vrx_{s_z}(Z)/Q\alz(\theta)$.
    \EndIf
    \ElsIf{$S=D_i = [t_i, t_{i+1}\wedge \vartheta_z)$}
    \State Sample $\theta\sim \Dunif(D_i)$,
    $r\gets U(1+\Delta)\psi_\alpha(z\HAT[t_i])/Q\alz(\theta)$.
    \ElsIf{$S=I_n = [\theta_n, \theta_{n+1}\wedge
      \vartheta_z)$}
    \State Sample $t\sim$ normalized $P^*_n(t)$ by \Cref{a:Penv}.
    \If{$t\in (c_n, d_n]$}
    \State $\theta\gets\laz^{-1}(t)$,  $r\gets U(1+\Delta)
    \laz^{-1}\#P^*_n(\theta) / Q\alz(\theta)$.
    \EndIf
    \EndIf
    \Until{$r\le1$}
    \State \Return $\theta$.
  \end{algorithmic}
\end{algorithm}

Putting the pieces together yields \Cref{a:qalz}.  As already seen,
its expected number of iterations is uniformly bounded for
$\alpha\ge\alpha_0$ and $z\le z_0$.  On the other hand, within each
iteration, it requires the evaluation of $\int Q^*_S$ 
for $S\in \Cal I$ and sampling from the intervals in $\Cal I = \{E,
\eno[0] D m, \eno[0] I \nal\}$.  From \eqref{e:m-al} and
\Cref{p:nodes}, the number of non-empty intervals in $\Cal I$ is
$O(\ln(1/\delta))$.   Thus the expected total number of operations of
the algorithm is $O(\ln(1/\delta))$.

\section{Experiments} \label{s:experiments}
To incorporate \Cref{a:ar-A,a:ar-B,a:ell-P} into \Cref{a:fp} to sample
the first passage, we conducted several numerical experiments.  Recall
that in \Cref{a:fp}, $z$ is random and can take any positive value,
and given $z$, $\chi\alz$ is sampled.  The main goal of the
experiments was to find which of \Cref{a:ar-A,a:ar-B,a:ell-P} was
relatively efficient at sampling from $\chi\alz$ for a given $z$.
\Cref{a:fp} was then tailored to these algorithms.  Specifically, for
different $z$, it determined which of the three algorithms to call to
sample from $\chi\alz$.  We then tested this version of \Cref{a:fp}.
The R code for the experiments is in the GitHub repository
\cite{chi:25:github}.

\subsection{Numerical issues}
To implement \Cref{a:ar-A,a:ar-B,a:ell-P}, several numerical issues
need to be addressed in particular for $\alpha$ close to 1.  In
\Cref {a:ar-A,a:ar-B}, a Gamma mixture has to be sampled.  From \eqref
{e:A-gamma-mix} and \eqref {e:B-gamma-mix}, the mixture has
$\gamma_{\delta, \tau}$ as a component.  When $\delta = 1-\alpha$ is
close to 0, it is numerically challenging to sample from
$\gamma_{\delta, \tau}$, as the sample value is often extremely close
to 0.  However, $\ln\# \gamma_{\delta, \tau}$, i.e., the pushforward
of $\gamma_{\delta, \tau}$ under the logarithmic transform, can be
sampled accurately and efficiently \cite {liu:17:compstat}.  Thus, in
the experiments, $\ln\#\gamma_{\delta, \tau}$ was sampled and if the
sample value was $\xi$, then $y=e^\xi$ was stored as the sample value
from $\gamma_{\delta,\tau}$.  Meanwhile, to avoid underflow, several
calculations involving $y$ was reformulated as ones involving $\xi$.

Now turn to \Cref{a:ell-P}.  First, on its \cref{a:ell-P5}, the
normalized $\varphi^*_{0,A, \delta}$ has to be sampled for some $A>0$.
From \eqref{e:varphi4-w} and \eqref{e:varphi5}, $\varphi^*_{0,A,
  \delta}$ is a mixture and has the \pdf of $A U^{1/\delta}$ as a
component, where $U\sim \Dunif(0,1)$.  For $\delta$ close to 0, the
sample value of $A U^{1/\delta}$ is often extremely close to 0.  On
the other hand, $\ln(A U^{1/\delta}) \sim \xi  + \ln A$, where
$\xi\sim \Dexp(1/\delta)$.  In the experiments, the shifted
exponential distribution was sampled and the sample value was
stored as the logarithm of a sample value of $A U^{1/\delta}$.
Second, on \cref{a:ell-P9}, the ratio of $(\ell\otimes \Id) \chi(v,\theta)=\chi\alz(e^v-1, \theta)e^v$ to $P\alz(v,\theta)$
needs to be evaluated, where from \cref{a:ell-P5}, $v$ is a sample
value from $\phi^*_{0, \ell(1/\tau), \delta}$.  As seen earlier, when
$\delta$ is close to 0, $v$ is often extremely close to 0, resulting
in extremely small $\chi\alz(e^v-1, \theta)e^v$ and $P\alz(v,\theta)$.
To evaluate their ratio, it follows from \eqref{e:fpus-joint} and
\eqref{e:Palz} that 
\begin{align} \label{e:chi-P-ratio}
  \frac{\chi\alz(e^v-1, \theta)e^v}{P\alz(v,\theta)}
  =
  e^{-\alpha}
  \Grp{\frac{\delta v/\alpha}{1 - e^{-\delta v/\alpha}}}^\alpha
  \nth{
    e^{\tau y - v} R(v) + (\cf{v\ge b} b^{-\alpha}  + 1/c_\alpha)v^\alpha
  },
\end{align}
where $y = e^v-1$, $b=\ell(1/\tau)$, and $R(v) = v^\alpha\varphi^*_{0,
  b,\delta}(v)$. 
Using 
\[
  \frac t{e^t-1} = \sumoi n B_n \frac{t^n}{n!}, \quad |t|<2\pi,
\]
where $B_n$ are the Bernoulli numbers (\cite{NIST:10}, 24.2.1),
$(\delta v/\alpha)/(1-e^{-\delta v/\alpha})$ can be evaluated
precisely for very small $v>0$.  On the other,  by
\eqref{e:varphi4-u}--\eqref{e:varphi5}, $R(v)$ can be written as
\[
  R(v)=M_0 [v^\alpha f^*(v)] \cf{0\le v<b\wedge x_\delta} + M_1
  v^{1/p_\delta-\delta} \Vrx_s(s^{1/p_\delta} - v^{1/p_\delta}),
\]
where $M_0$, $M_1$, $x_\delta$, $p_\delta$, and $s$ are explicit positive
constants, and $v^\alpha f^*(v)$ is a polynomial with a positive value
at $v=0$.  As a result, the evaluation using \eqref{e:chi-P-ratio} has
a good precision and was used in the experiments.  Finally, similar to
\Cref{a:ar-A,a:ar-B}, $\ln(e^v-1)$ as well as $(e^v-1, \theta)$ was
output on the last line of \Cref{a:ell-P}.  For $v$ close to 0,
$\ln(e^v-1)$ was evaluated via $\ln v - \ln\frac v{e^v-1}$.

\subsection{Comparison of algorithms}
We compared the running times
of \Cref{a:ar-A,a:ar-B,a:ell-P}
for different
values of $(\alpha,z)$.  For each pair of the algorithms, the basic
routine was as follows.
\begin{center}
  \begin{algorithmic}[1]
    \Require $(\alpha,z)$ and sample size $N$.
    \State Apply each algorithm to draw a sample of size $N$ from
    $\chi\alz$.  Record $T_1$ and $T_2$, the running times of the two
    algorithms, respectively.
    \State Run Kolmogorov--Smirnov (KS) test on the two samples
    and record its $p$-value.
  \end{algorithmic}
\end{center}

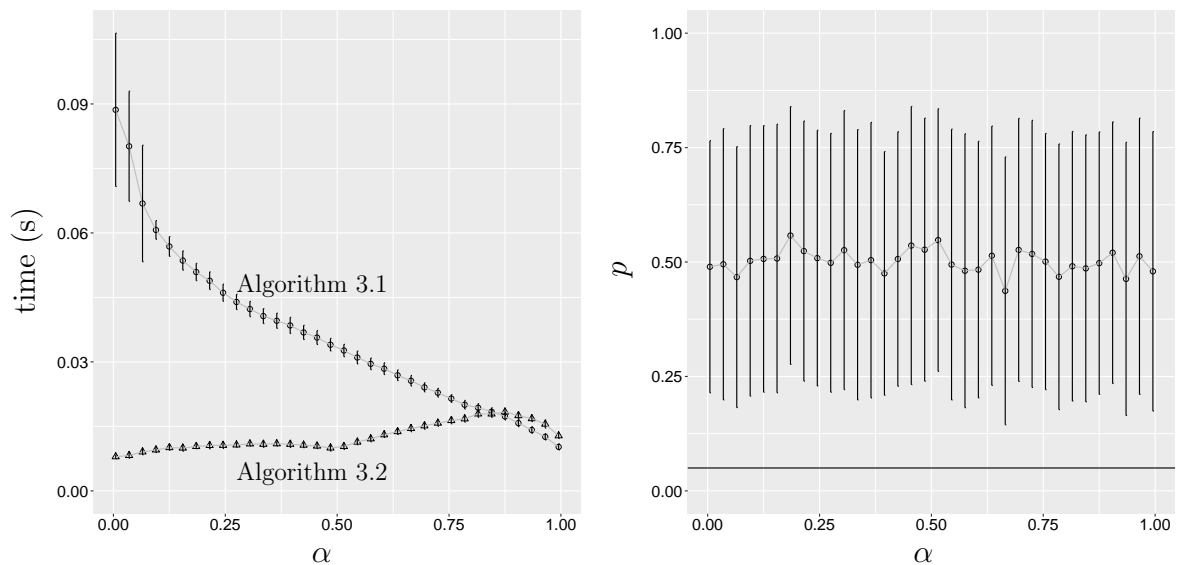
\begin{figure}[p]
  \caption{Comparison of \Cref{a:ar-A,a:ar-B} when $z=0.05 \times
    10^{0.6}$.  Left:
    for each $\alpha$, the mean and SD of running time (in s) over 200
    repetitions, with $10^3$ values sampled from $\chi\alz$ by each
    algorithm per repetition.  Right: the mean and SD of $p$-value of
    KS test on the     two samples.  The horizontal line at the bottom
    is $p=0.05$.
  }\label{f:A-B1}
  \begin{center}
    \setlength{\unitlength}{1mm}
    \begin{picture}(152,73)(2,2)
      \put(4,5){
        \includegraphics[width=2.8in] {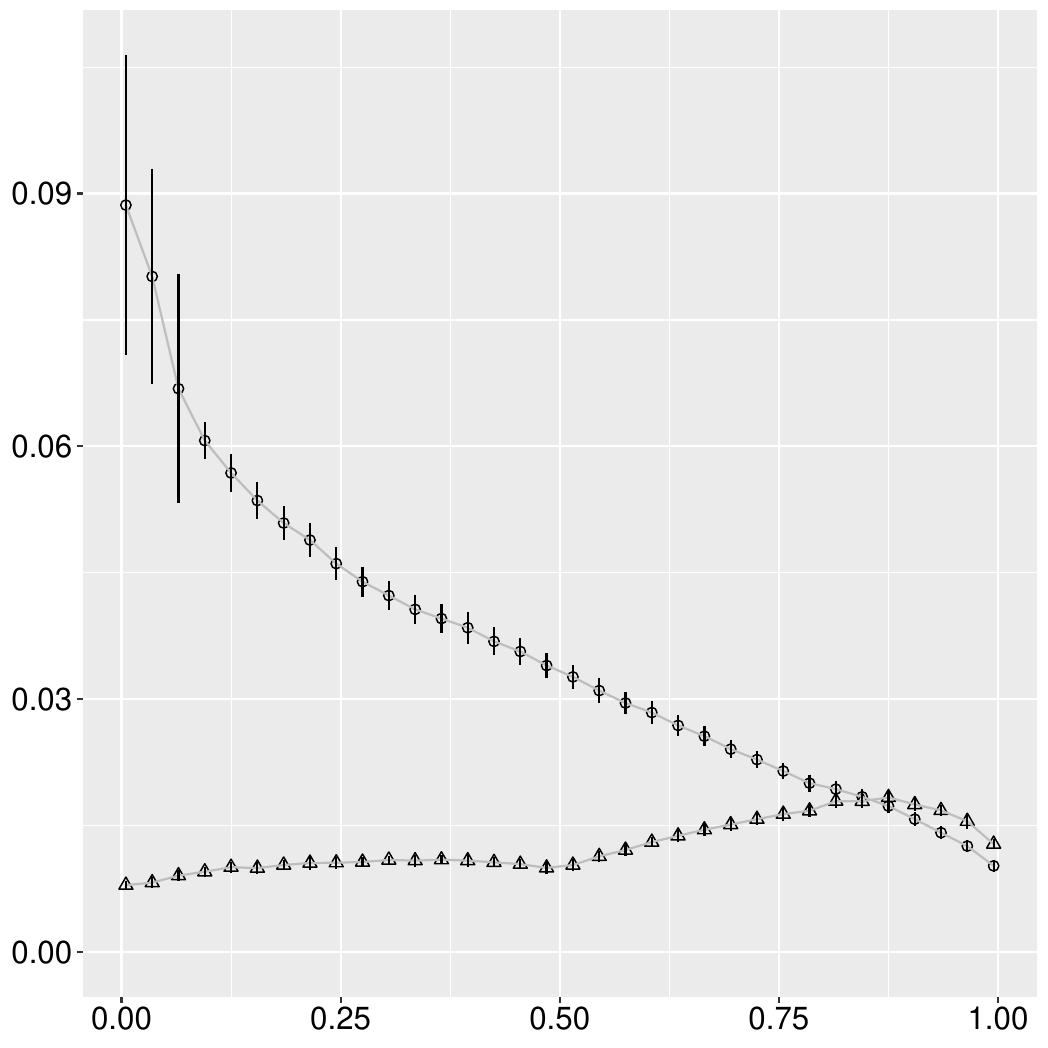}
      } \put(40,2){$\alpha$} \put(0,35){\rotatebox{90}{time (s)}}

      \put(83,5){
        \includegraphics[width=2.8in] {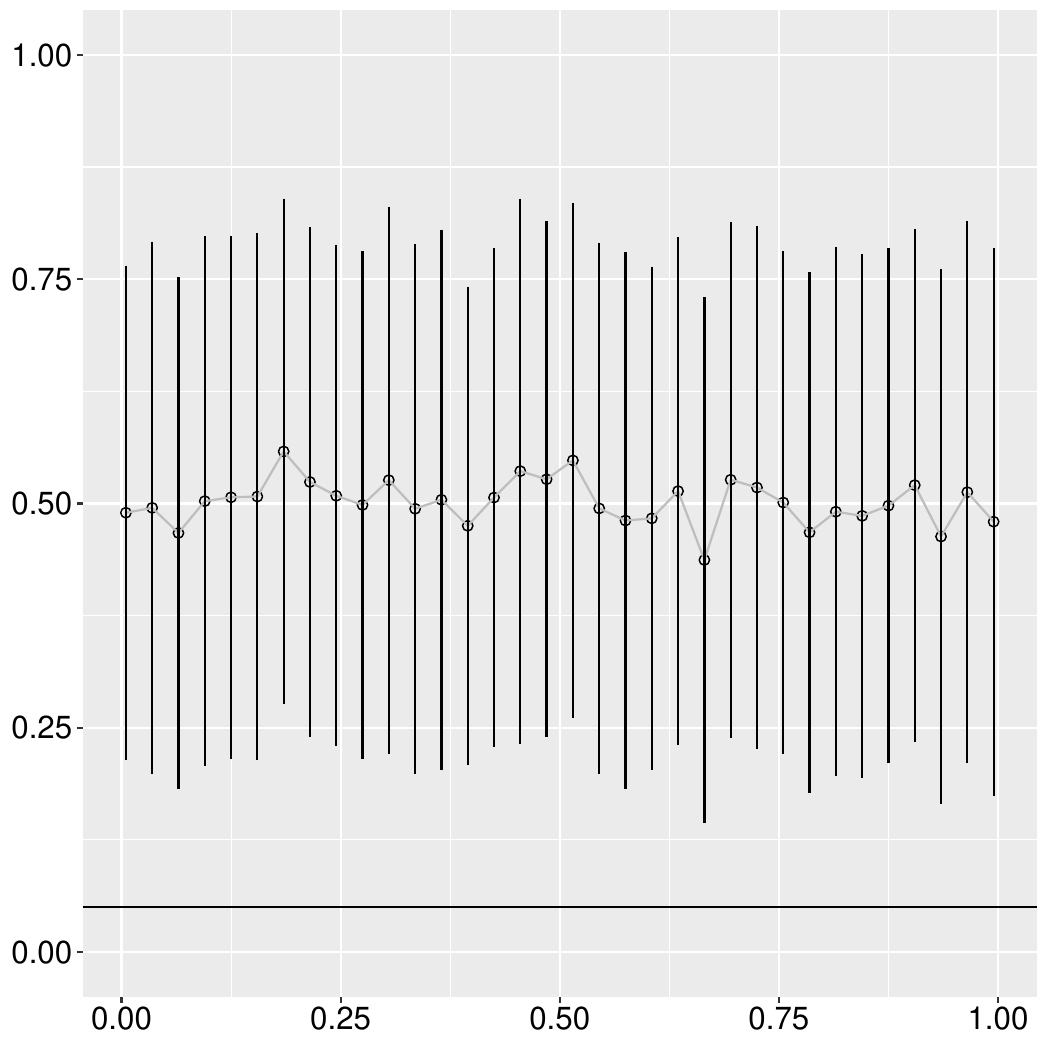}
      } \put(120,2){$\alpha$} \put(80,40){\rotatebox{90}{$p$}}
      \put(30, 38){\scalebox{.8}{\Cref{a:ar-A}}}
      \put(30, 13){\scalebox{.8}{\Cref{a:ar-B}}}
    \end{picture}
  \end{center}
\end{figure}

First, we compared \Cref{a:ar-A,a:ar-B} for all the pairs of $\alpha
= 0.005 + 0.03 i\in [0.005, 0.995]$ and $z = 0.05 \times 10^{0.15j}\in
[0.05, 4]$, where $i,j$ are integers.  For each $(\alpha,z)$, we set
$N=1000$ and repeated the basic routine 200 times.  Then we calculated 
the average and standard deviation (SD) of the running times and
$p$-value of KS test over the repetitions.  In \Cref{f:A-B1,f:A-B2},
the results are displayed as functions of $\alpha$ while $z$ is
fixed.  From \Cref{f:A-B1}, for relatively small $z$, unless $\alpha$
was close to 1, \Cref{a:ar-B} was faster than \Cref{a:ar-A}.  On the
other hand, from \Cref{f:A-B2}, for relatively large $z$,
\Cref{a:ar-A} was faster than \Cref{a:ar-B}.  These observations are
consistent with the complexity analysis in \cref{s:general}.  For all
$(\alpha, z)$ in the experiment, the $p$-values in the figures confirm
that the two algorithms sampled from the same distribution.

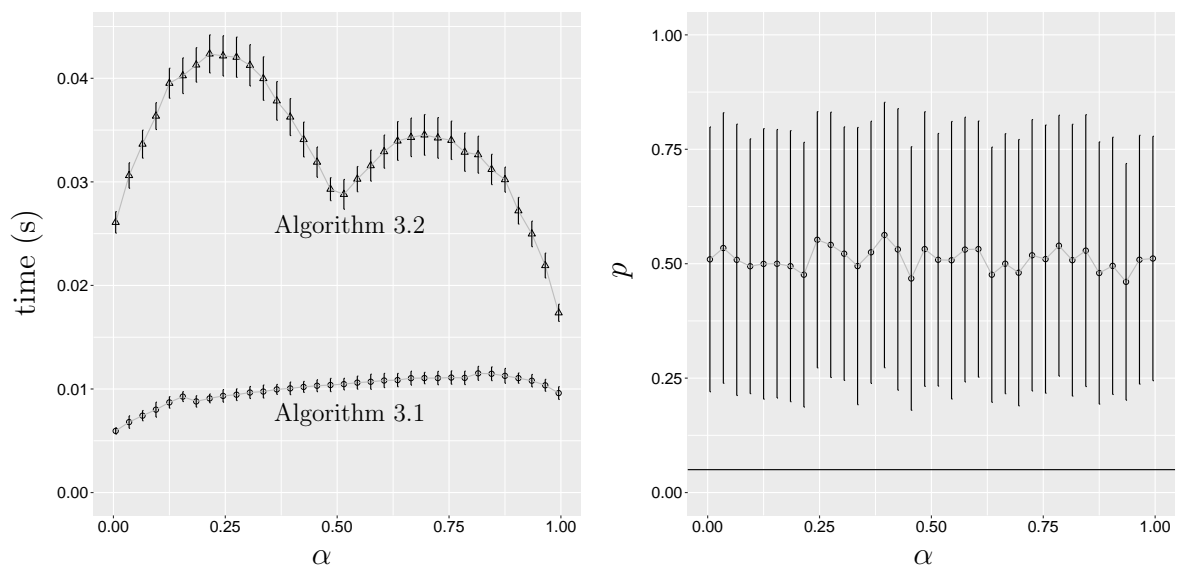
\begin{figure}[p]
  \caption{Comparison of \Cref{a:ar-A,a:ar-B} under the same setting
    as in \Cref{f:A-B1} except that $z=0.05\times 10^{1.5}$.}\label{f:A-B2}
  \begin{center}
    \setlength{\unitlength}{1mm}
    \begin{picture}(152,73)(2,2)
      \put(4,5){
        \includegraphics[width=2.8in] {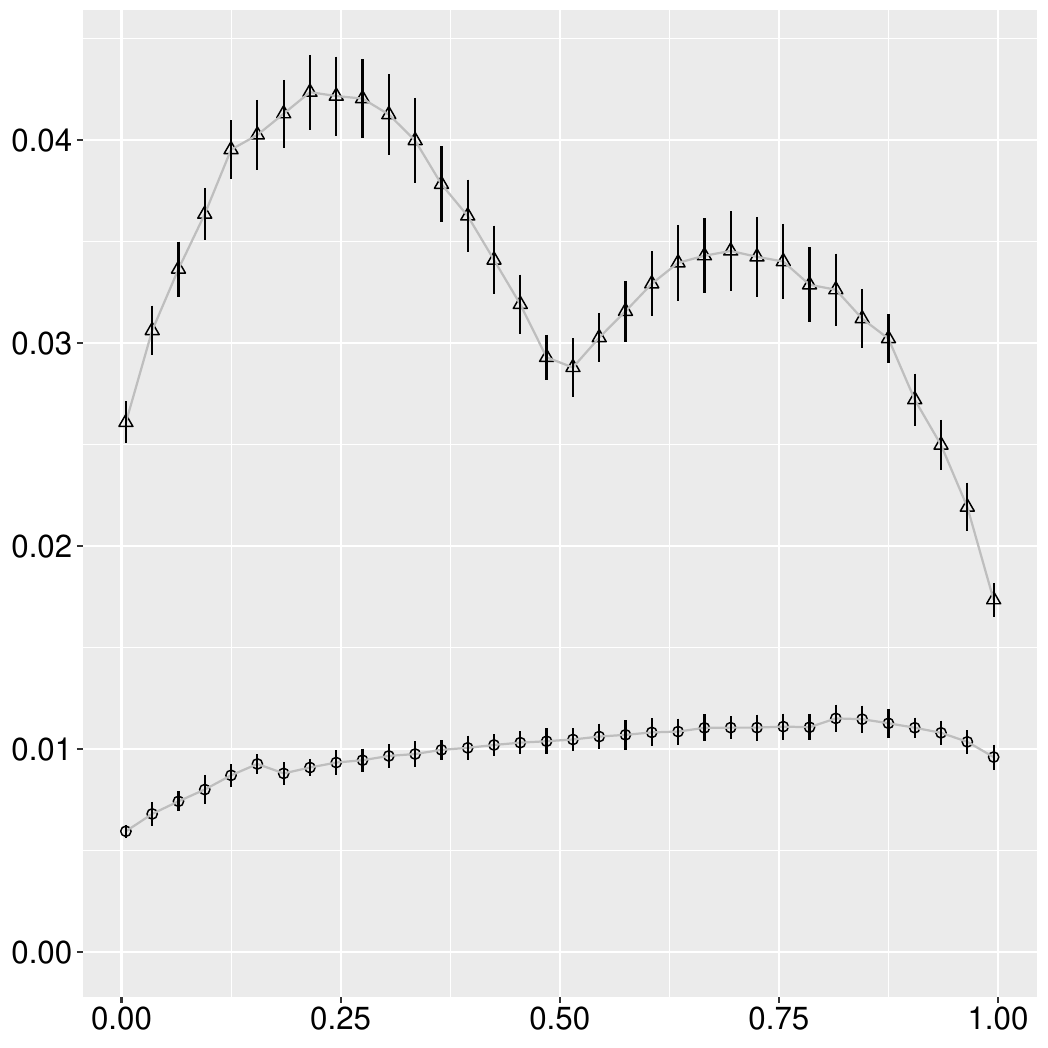}
      } \put(40,2){$\alpha$} \put(0,35){\rotatebox{90}{time (s)}}
      \put(83,5){
        \includegraphics[width=2.8in] {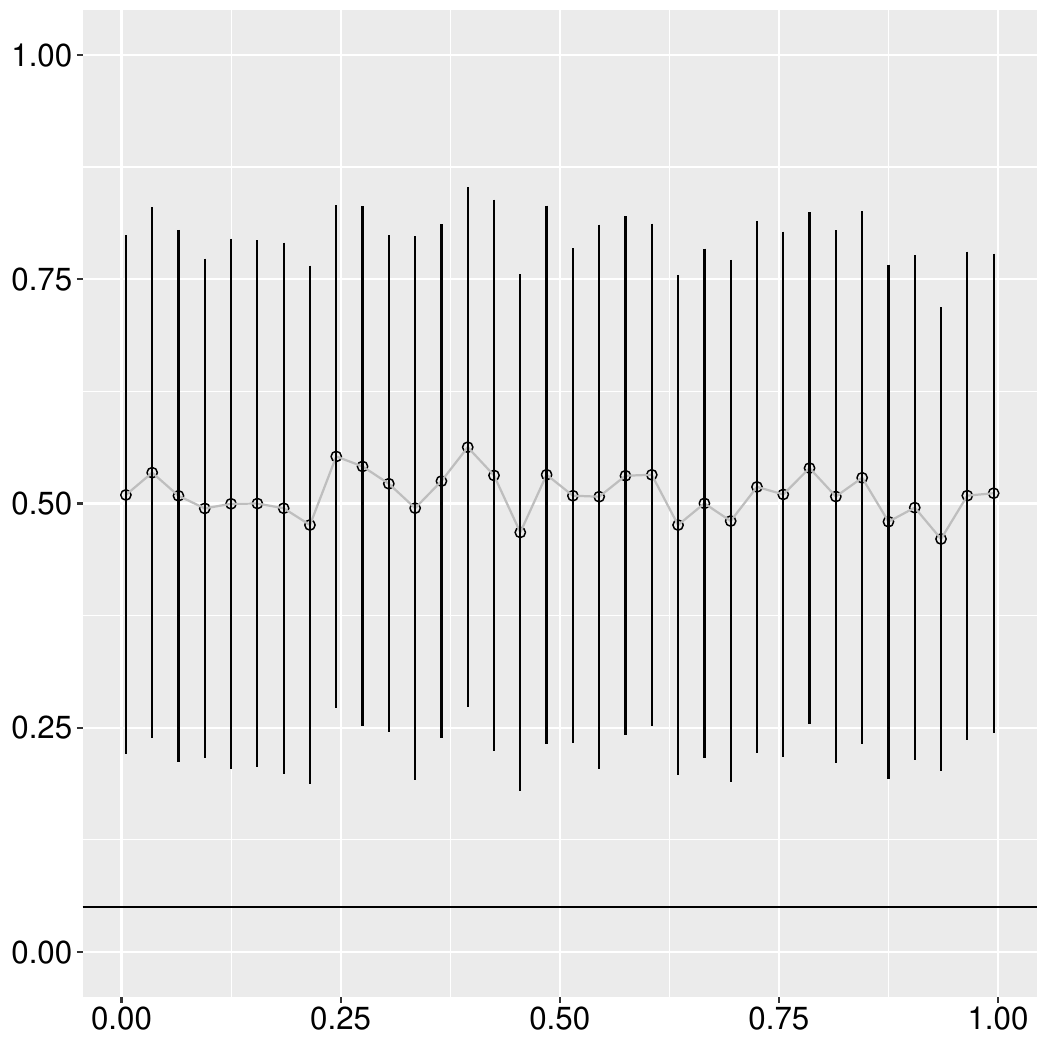}
      } \put(120,2){$\alpha$} \put(80,40){\rotatebox{90}{$p$}}
      \put(35, 21){\scalebox{.8}{\Cref{a:ar-A}}}
      \put(35, 46){\scalebox{.8}{\Cref{a:ar-B}}}
    \end{picture}
  \end{center}
\end{figure}

\begin{figure}[p]
  \caption{Comparison of running times of \Cref{a:ar-A,a:ar-B} to
    sample $2\times 10^5$ values in a single run, with the same values
    of $z$ as in \Cref{f:A-B1} (left) and \Cref{f:A-B2}
    (right).}
  \label{f:A-B3} 
  \begin{center}
    \setlength{\unitlength}{1mm}
    \begin{picture}(152,73)(2,2)
      \put(4,5){
        \includegraphics[width=2.8in] {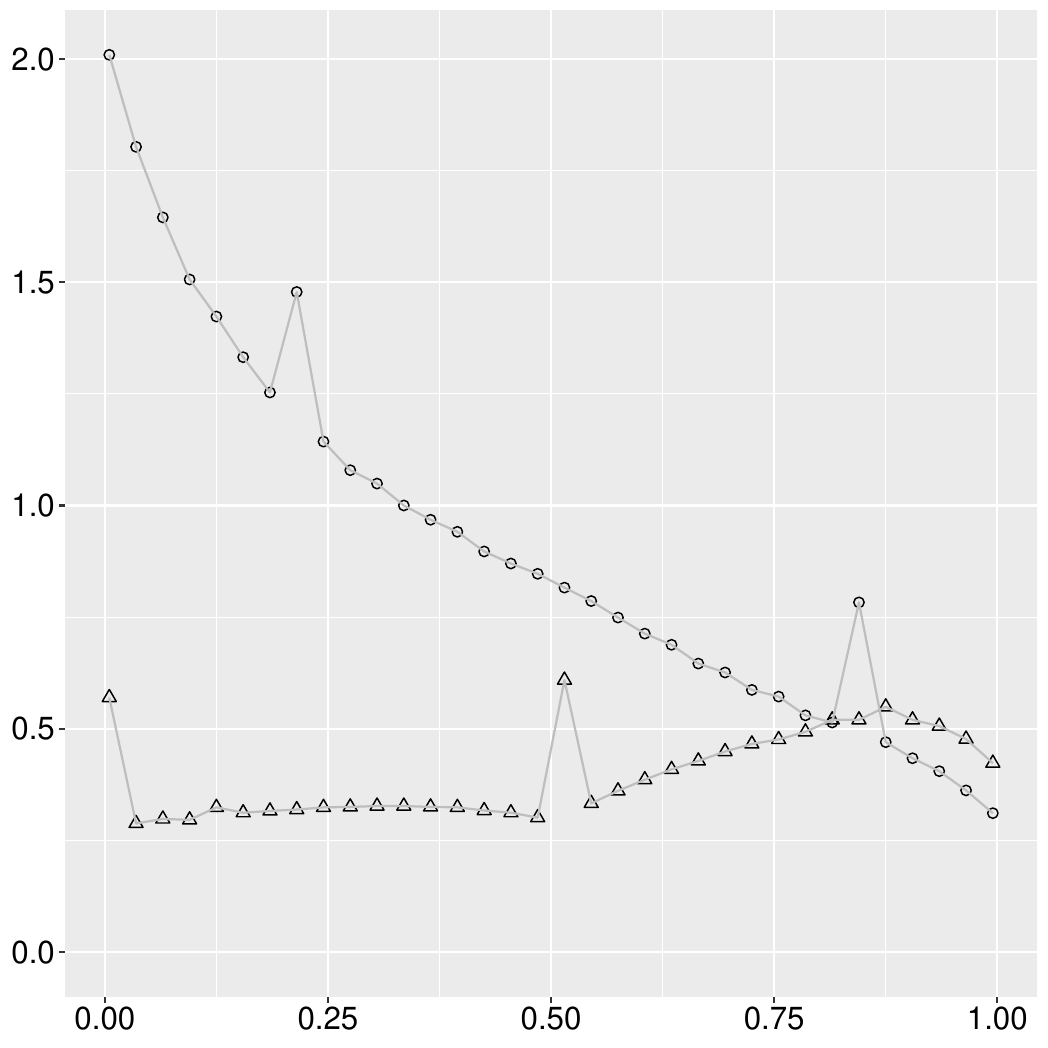}
      } \put(40,2){$\alpha$} \put(0,35){\rotatebox{90}{time (s)}}
      \put(83,5){
        \includegraphics[width=2.8in] {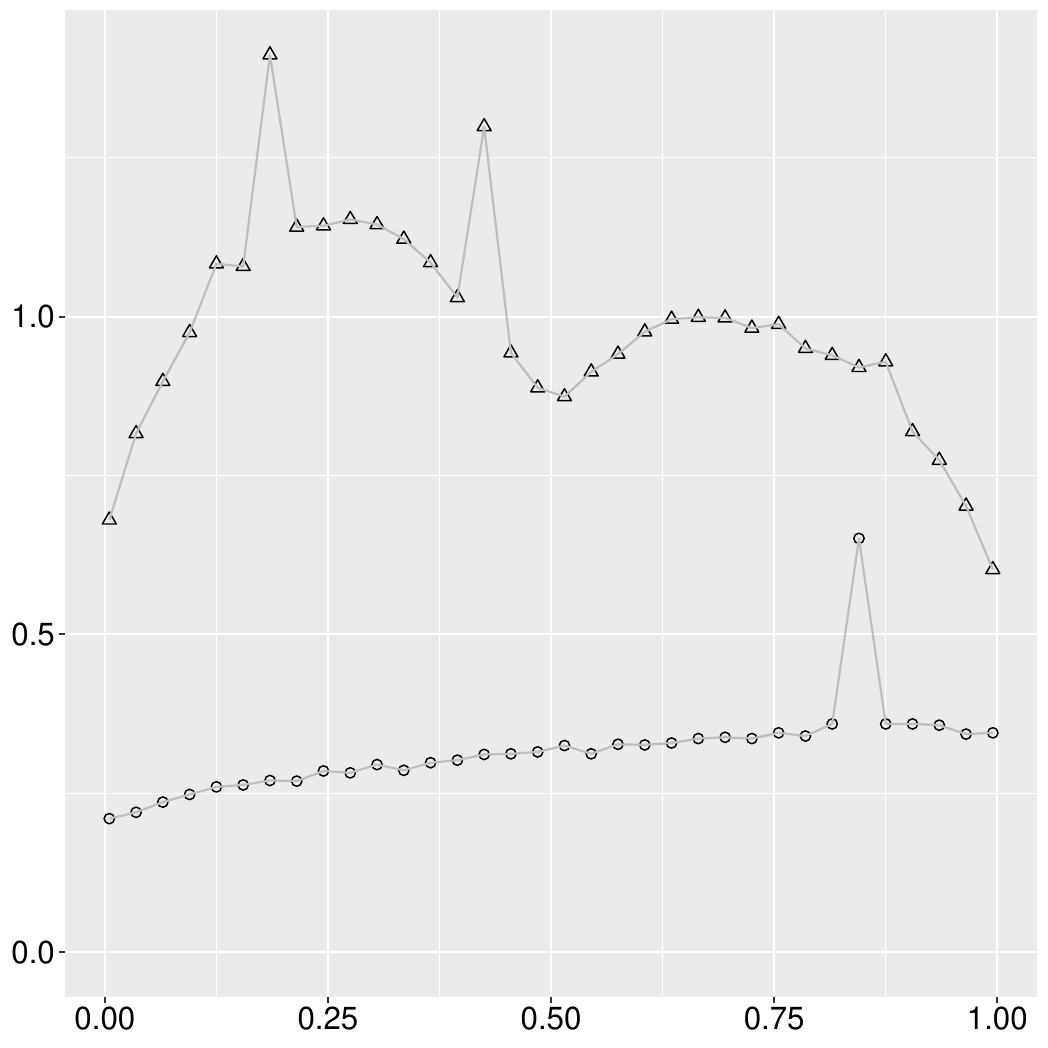}
      } \put(120,2){$\alpha$} \put(80,40){\rotatebox{90}{$p$}}
      \put(30, 45){\scalebox{.8}{\Cref{a:ar-A}}}
      \put(30, 17){\scalebox{.8}{\Cref{a:ar-B}}}
      \put(114, 20.5){\scalebox{.8}{\Cref{a:ar-A}}}
      \put(114, 45){\scalebox{.8}{\Cref{a:ar-B}}}
    \end{picture}
  \end{center}
\end{figure}

In the above experiment, for each pair $(\alpha,z)$, a total of
$2\times 10^5$ values were sampled by each algorithm over 200
repetitions.  From \Cref{f:A-B1,f:A-B2}, the total running time to
do this was relatively long.  However, the algorithms were implemented
in R using vectorized computation.  As \Cref{f:A-B3} shows, if the
algorithms were used to sample $2\times 10^5$ values in a single run,
the total running times were much shorter, while their ratio had the
same pattern as in \Cref{f:A-B1,f:A-B2}.  Next, the algorithms were
used to sample $2\times 10^5$ values from $\chi\alz$ in a single run
with $\alpha\in\{0.905, 0.935, 0.965, 0.995, 0.999\}$ and $z\in
\{10^{-4}, 10^{-3}, 10^{-2}, 2, 3, 4\}$.   Consistent with
\Cref{f:A-B1,f:A-B2}, \Cref{t:A-B4} shows that for very small $z$,
unless $\alpha$ was very close to 1, \Cref{a:ar-B} was much faster
than \Cref{a:ar-A}, while for large $z$, \Cref{a:ar-B} was
significantly faster.

\setlength{\tabcolsep}{4pt}
\begin{table}[t]
  \caption{Comparison of running times of \Cref{a:ar-A,a:ar-B} to
    sample $2\times 10^5$ values in a single run for $\alpha$ close to
    1 and $z$ close to 0 or at least 2.} \label{t:A-B4} 
  \begin{center}
    \small
    \begin{tabular*}{.48\textwidth}{@{\extracolsep{\fill}}|c|c|c|c|c|c|}\hline
      $\alpha$ & 0.905 & 0.935 & 0.965 & 0.995 & 0.999 \\\hline
      \multicolumn{6}{|c|}{$z=10^{-4}$} \\\hline
      $T_1$ &3206 & 2129 & 1140 & 162.7 & 34.89 \\\hline
      $T_2$ & 2.260 & 2.596 & 3.011 & 3.389 & 3.411 \\\hline
      \multicolumn{6}{|c|}{$z=10^{-3}$} \\\hline
      $T_1$ &180.0 & 124.8 & 68.17 & 10.89 & 3.622 \\\hline
      $T_2$ & 1.589 & 1.765 & 1.941 & 1.970 & 1.964 \\\hline
      \multicolumn{6}{|c|}{$z=10^{-2}$} \\\hline
      $T_1$ &9.743 & 7.048 & 4.298 & 1.412 & 1.054 \\\hline
      $T_2$ & 0.966 & 1.064 & 1.092 & 1.076 & 1.040 \\\hline
    \end{tabular*}\hspace{1ex}
    \small
    \begin{tabular*}{.48\textwidth}{@{\extracolsep{\fill}}|c|c|c|c|c|c|}\hline
      $\alpha$ & 0.905 & 0.935 & 0.965 & 0.995 & 0.999 \\\hline
      \multicolumn{6}{|c|}{$z=2$} \\\hline
      $T_1$ & 0.367 & 0.409 & 0.363 & 0.342 & 0.406 \\\hline
      $T_2$ & 1.195 & 1.054 & 1.036 & 0.858 & 0.794 \\\hline
      \multicolumn{6}{|c|}{$z=3$} \\\hline
      $T_1$ & 0.348 & 0.355 & 0.349 & 0.344 & 0.355 \\\hline
      $T_2$ & 2.531 & 2.310 & 2.120 & 1.773 & 1.718 \\\hline
      \multicolumn{6}{|c|}{$z=4$} \\\hline
      $T_1$ & 0.353 & 0.363 & 0.396 & 0.351 & 0.359 \\\hline
      $T_2$ & 6.150 & 5.548 & 4.856 & 4.062 & 3.949 \\\hline
    \end{tabular*}
  \end{center}
\end{table}

Next, we tested the algorithms when $(\alpha,z)$ was close to $(1,0)$,
For such $(\alpha,z)$, the above experiments demonstrated that
\Cref{a:ar-B} was faster than \Cref{a:ar-A}.  Therefore, in this
experiment, we only compared \Cref{a:ar-B} and \Cref{a:ell-P} by using
them to sample $10^5$ values from $\chi\alz$ in a single run for each
pair of $\alpha= 0.995+0.0005\times i\in [0.995, 0.9995]$ and
$z=10^{-5j} \in [10^{-60}, 10^{-5}]$, where $i,j$ are integers.  From
\Cref{f:B-ell1,f:B-ell2}, even for $z$ extremely small, \Cref{a:ar-B}
was faster than \Cref{a:ell-P}.  However, for $z$ ``ultra-small'',
e.g., $z\le 10^{-50}$, the latter became much faster.

To see why it is necessary to consider ultra-small $z$, recall that in
\Cref{a:fp}, given $\alpha$, $z$ is sampled as $\xi/\HAT[\Theta]$,
where $\xi\sim\Dexp(1)$ and $\Theta\sim \Dunif(0, \pi)$.  When
$\alpha$ is very close to 1, there is a significant chance to get
ultra-small $z$.  For example, when $\alpha = 1-10^{-3}$, by
simulation, the chance to get $z\le 10^{-50}$ is approximately
$2.2\times 10^{-2}$.  \Cref{f:B-ell1,f:B-ell2} indicate that for such
$z$, \Cref{a:ell-P} should be used.

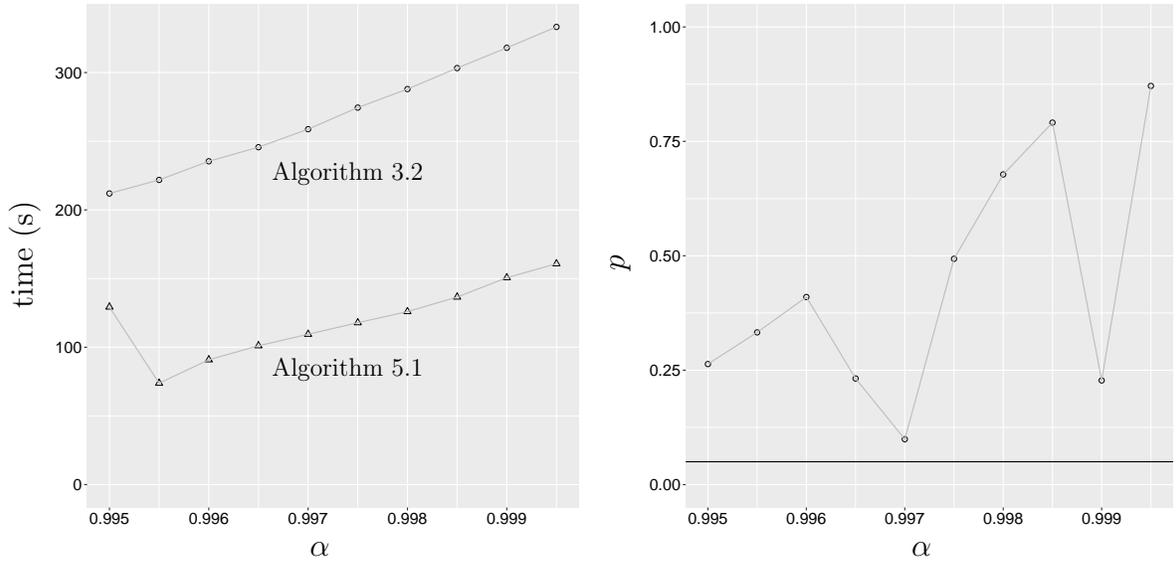
\begin{figure}[p]
  \caption{Comparison of \Cref{a:ar-B,a:ell-P} when $z=10^{-50}$.
    Left: running time to sample $10^5$ values in a single
    run.  Right: $p$-value of KS test on the two
    samples.}\label{f:B-ell1}
  \begin{center}
    \setlength{\unitlength}{1mm}
    \begin{picture}(152,73)(2,2)
      \put(4,5){
        \includegraphics[width=2.8in] {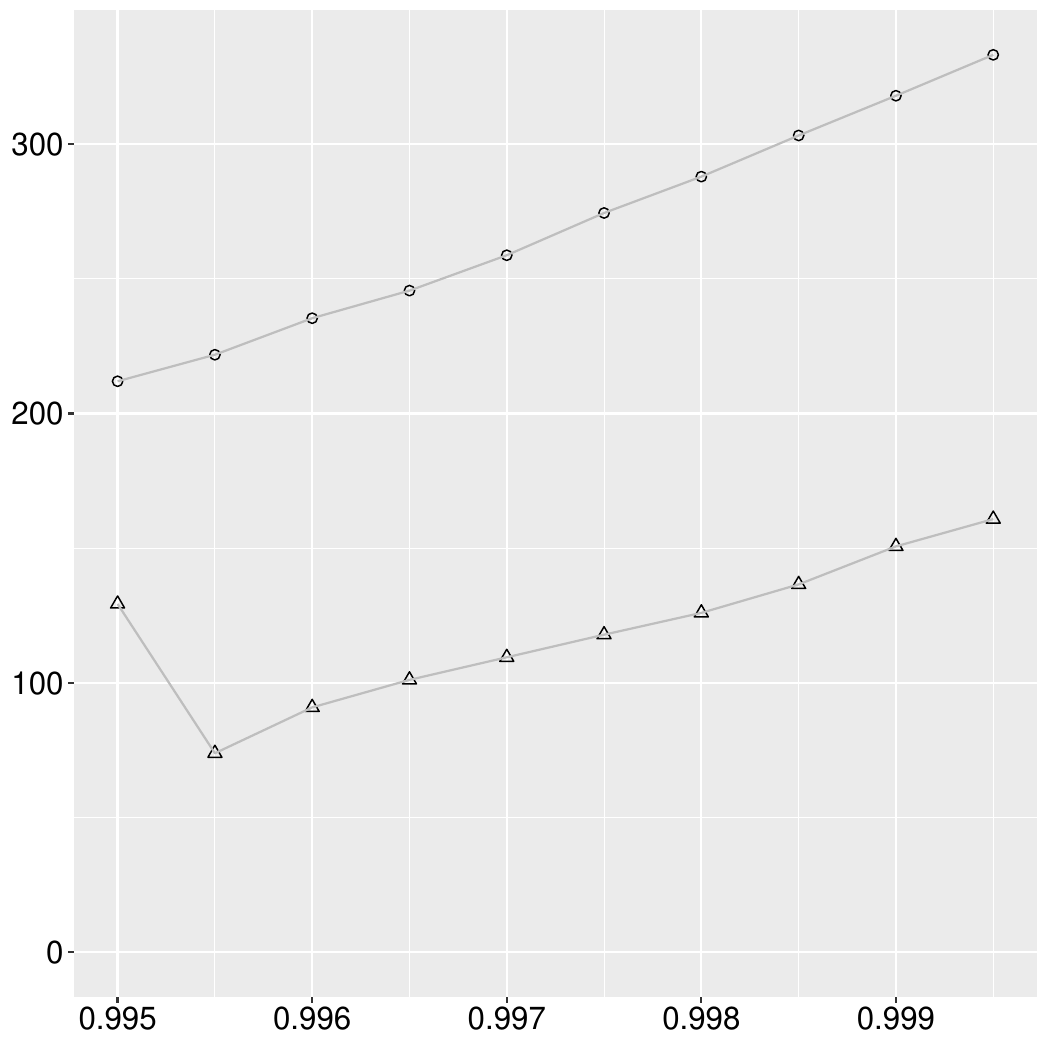}
      } \put(40,2){$\alpha$} \put(0,35){\rotatebox{90}{time (s)}}
      \put(83,5){
        \includegraphics[width=2.8in] {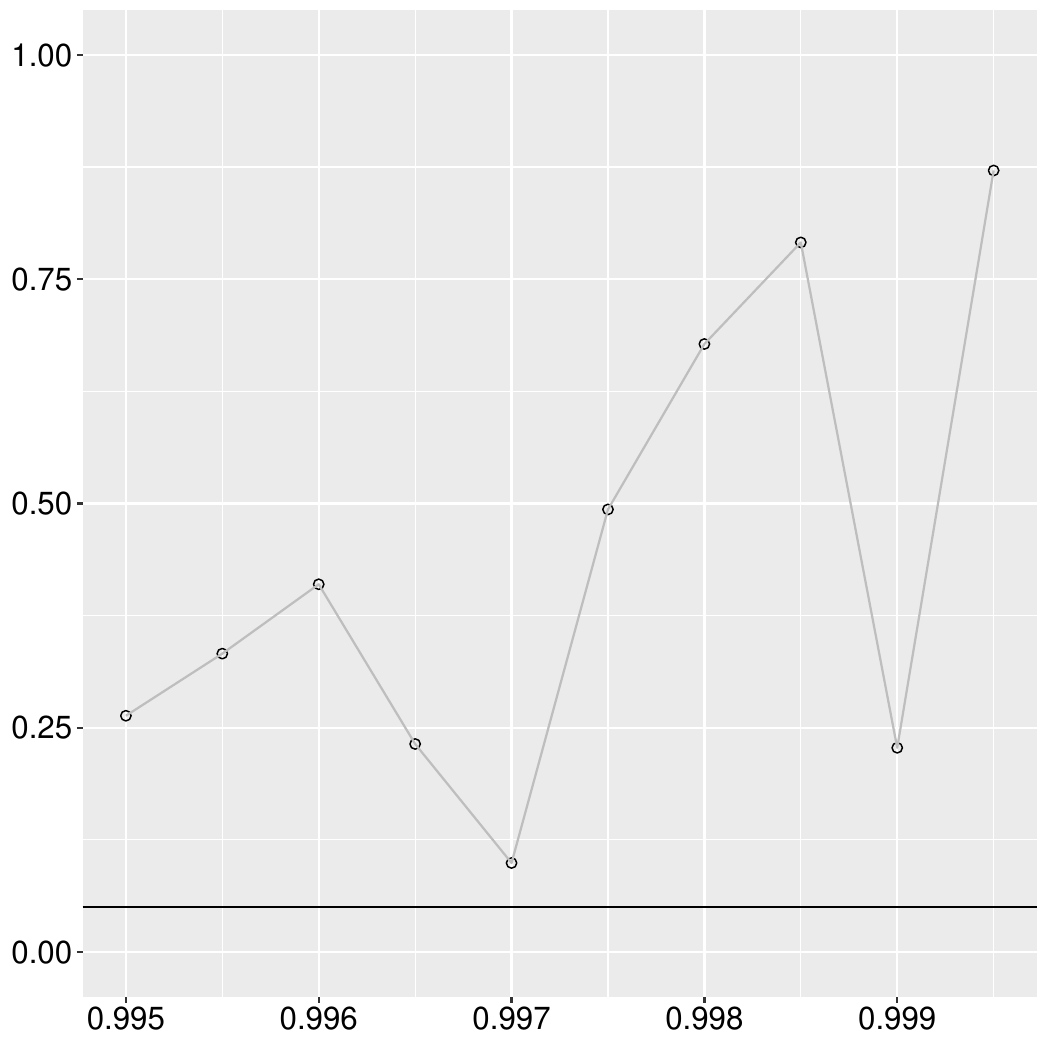}
      } \put(120,2){$\alpha$} \put(80,40){\rotatebox{90}{$p$}}
      \put(35, 26){\scalebox{.8}{\Cref{a:ell-P}}}
      \put(35, 52){\scalebox{.8}{\Cref{a:ar-B}}}
    \end{picture}
  \end{center}
\end{figure}

\begin{figure}[p]
  \caption{Comparison of \Cref{a:ar-B,a:ell-P} under the same setting
    as in \Cref{f:B-ell1} except that $z=10^{-25}$.} \label{f:B-ell2}
  \begin{center}
    \setlength{\unitlength}{1mm}
    \begin{picture}(152,73)(2,2)
      \put(4,5){
        \includegraphics[width=2.8in] {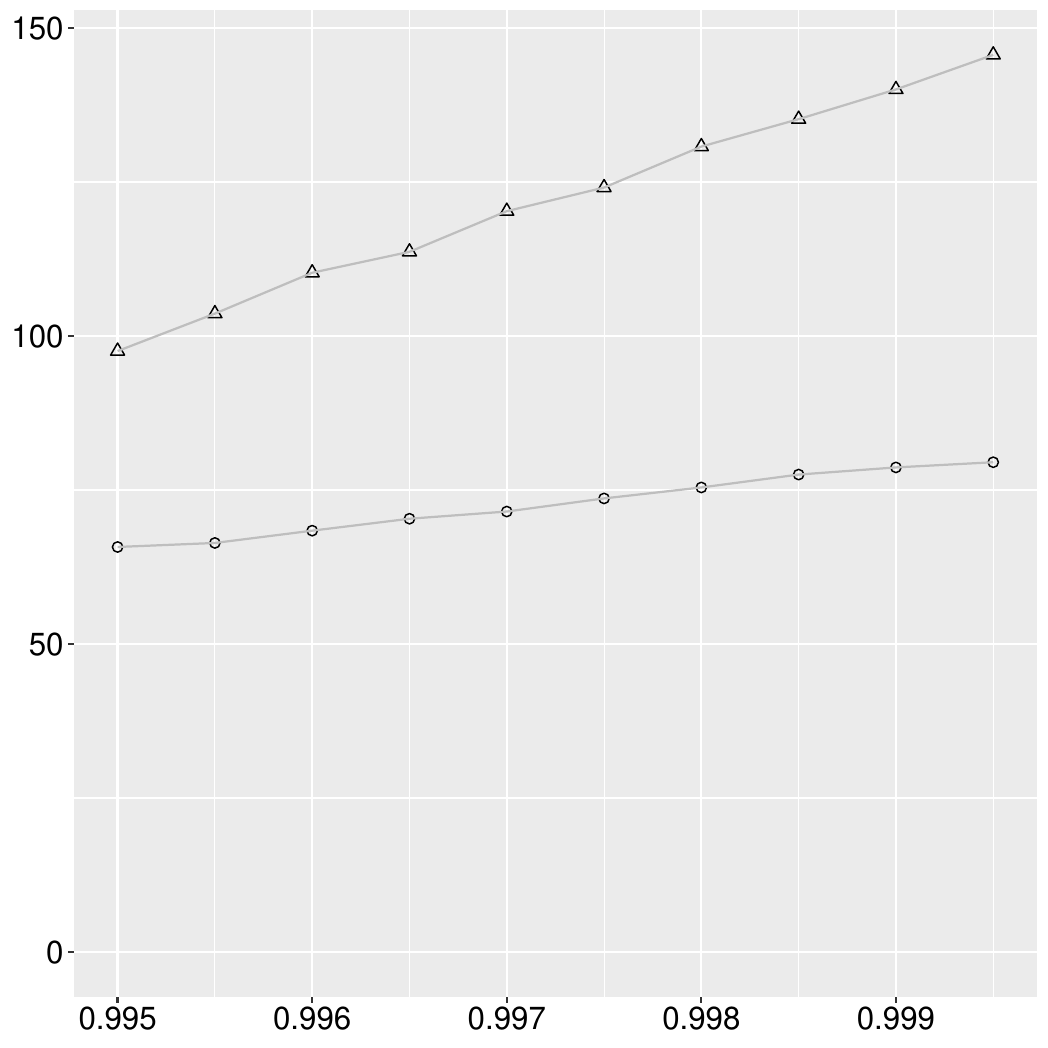}
      } \put(40,2){$\alpha$} \put(0,35){\rotatebox{90}{time (s)}}
      \put(83,5){
        \includegraphics[width=2.8in] {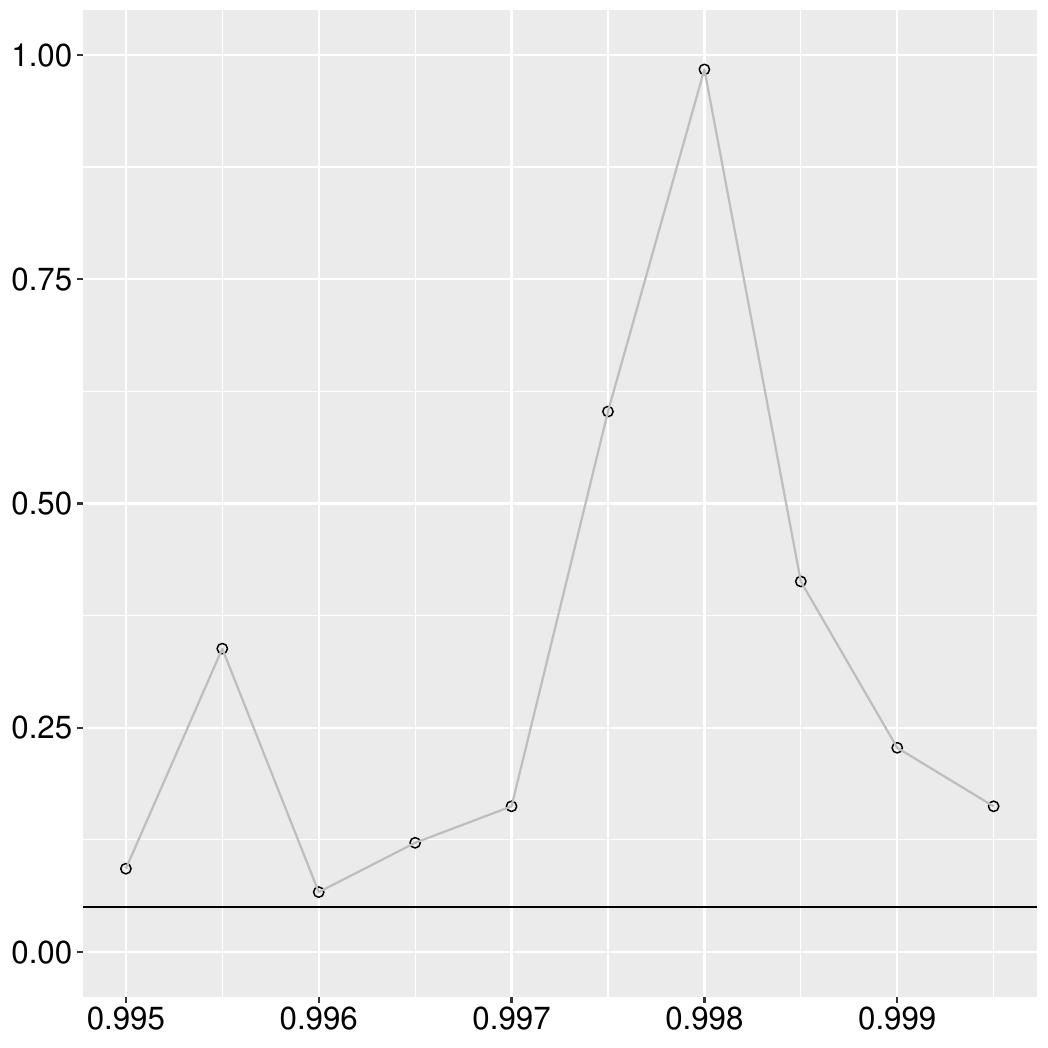}
      } \put(120,2){$\alpha$} \put(80,40){\rotatebox{90}{$p$}}
      \put(35, 36){\scalebox{.8}{\Cref{a:ar-B}}}
      \put(35, 56){\scalebox{.8}{\Cref{a:ell-P}}}
    \end{picture}
  \end{center}
\end{figure}

\begin{algorithm}[t]
  \caption{Modified \Cref{a:fp}} \label{a:fp-modified}
  \begin{algorithmic}[1]
    \State Sample $\Theta\sim\Dunif(0,\pi)$, $\xi\sim\Dexp(1)$,
    $U$ \iid$\sim \Dunif(0,1)$.
    \State $z\gets \xi/\HAT[\Theta]$, $s\gets \alpha
    (\delta/z)^{\delta/\alpha}$, $t\gets B^{-1}(s)$, $y\gets0$,
    $\theta\gets 0$.
    \If{$U\ge-b'(t)/[-b'(t) + \alpha^{-1} t^{-1} b(t)]$}
    \If{$z\ge1$}
    \State Use \Cref{a:ar-A} to sample $(y,\theta)\sim$ normalized
    $\chi_{\alpha,z}$.
    \ElsIf{$\alpha\le0.9$ or $z\ge (1-\alpha)\times 10^{-30}$}
    \State Use \Cref{a:ar-B} to sample $(y,\theta)\sim$ normalized
    $\chi_{\alpha,z}$.
    \Else
    \State Use \Cref{a:ell-P} to sample $(y,\theta)\sim$ normalized
    $\chi_{\alpha,z}$ with $\Delta=0.5$, $\alpha_0=2/3$,
    \State and $\theta_0=6\pi/7$ in the subroutine \Cref{a:qalz}.
    \EndIf
    \EndIf
    \State \Return $(z, t,y, \theta)$.
  \end{algorithmic}
\end{algorithm}

\subsection{Sampling of the first passage}
Based on the above comparisons, we modified \Cref{a:fp} so
that for different $(\alpha,z)$, a relatively efficient one among
\Cref{a:ar-A,a:ar-B,a:ell-P} was used to sample from the normalized
$\chi\alz(y,\theta)$.  The result is \Cref{a:fp-modified}.  Note that,
unlike \Cref{a:fp}, it does not return $x=(1+y)^{-\delta/\alpha}$ and
$(1-x)  V^{-1/\alpha}$, where $V\sim \Dunif(0,1)$.  These two are the
sample values of the undershoot and jump at the first passage.  For
$\alpha$ close to 1, as the sampled $y$ is often extremely close to 0,
to maintain numerical precision, appropriate transforms of the
undershoot and jump should be calculated instead.  This issue is not
directly related to the core issue of the paper, so for simplicity it
was not addressed in the experiments.

\Cref{a:fp-modified} was tested for two barrier functions: 1)
$b(t)\equiv10$, giving $b'(t)\equiv0$ and $B^{-1}(s) =
(10/s)^\alpha$, and 2) $b(t) = (100 - t^{1/\alpha})_+$, giving $b'(t)
= -\alpha^{-1} t^{\delta/\alpha} \cf{t<100^\alpha}$ and $B^{-1}(s) =
[100/(s+1)]^\alpha$.  For each barrier function, all the following
values of $\alpha$ were considered: $0.005 + 0.03 i\in [0.005,
0.995]$, $0.996 + 0.001j \in [0.996, 0.999]$, and $0.9999$.  For each
value, the algorithm was run 100 times with $10^4$ values sampled per
run.  Then the average and SD of the running time were calculated.
\Cref{f:fp} displays the results.  For $\alpha$ not too close to 1,
the average running time per run was well below 1 second with very
small SD.  On the other hand, for $\alpha\ge0.995$, the running time
was substantially longer.  The first half of \Cref{t:fp} shows the
actual values of the mean and SD for $\alpha\in\{0.996, .997, .998,
.999, .9999\}$.  To see how the algorithm scaled up, we then ran a
second experiment for the same set of $\alpha$.  In the experiment,
\Cref{a:fp-modified} was again run 100 times.  However, in each run,
$10^5$ values were sampled.  The second half of \Cref{t:fp} shows the
results.  The experiments indicate that by incorporating
\Cref{a:ar-A,a:ar-B,a:ell-P} appropriately, \Cref{a:fp-modified}
sampled the first passage efficiently.

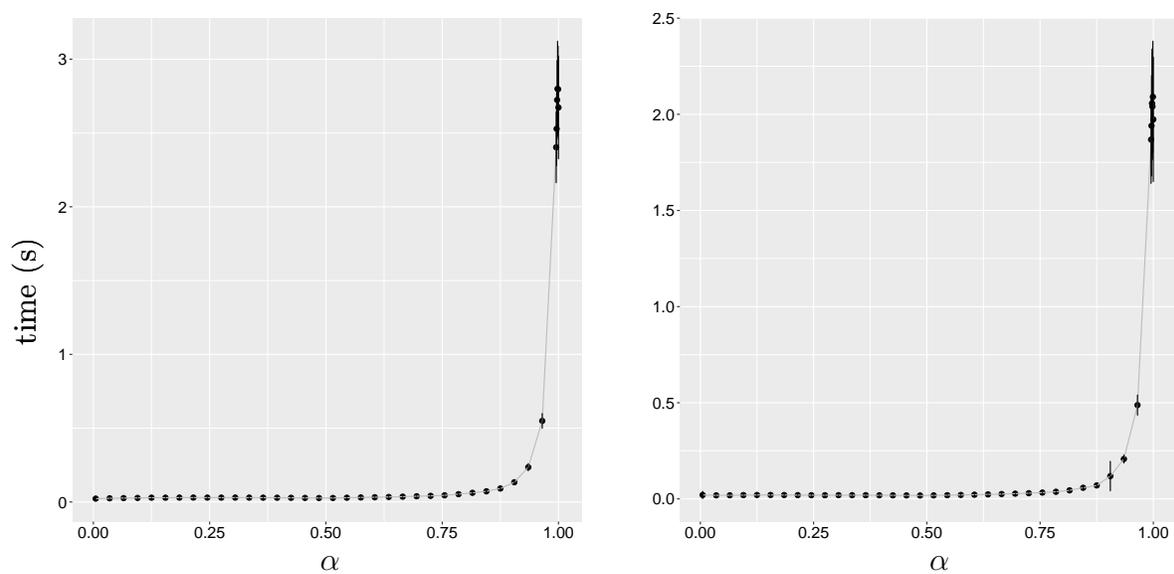
\begin{figure}[p]
  \caption{Mean and SD of running time (in s) of \Cref{a:fp-modified}
    over 100 repetitions, with $10^4$ values sampled per repetition.
    Left: $b(t)\equiv 10$.  Right: $b(t)
    = (100 - t^{1/\alpha})_+$} \label{f:fp}
  \begin{center}
    \setlength{\unitlength}{1mm}
    \begin{picture}(152,75)(2,2)
      \put(4,8){
        \includegraphics[height=2.8in] {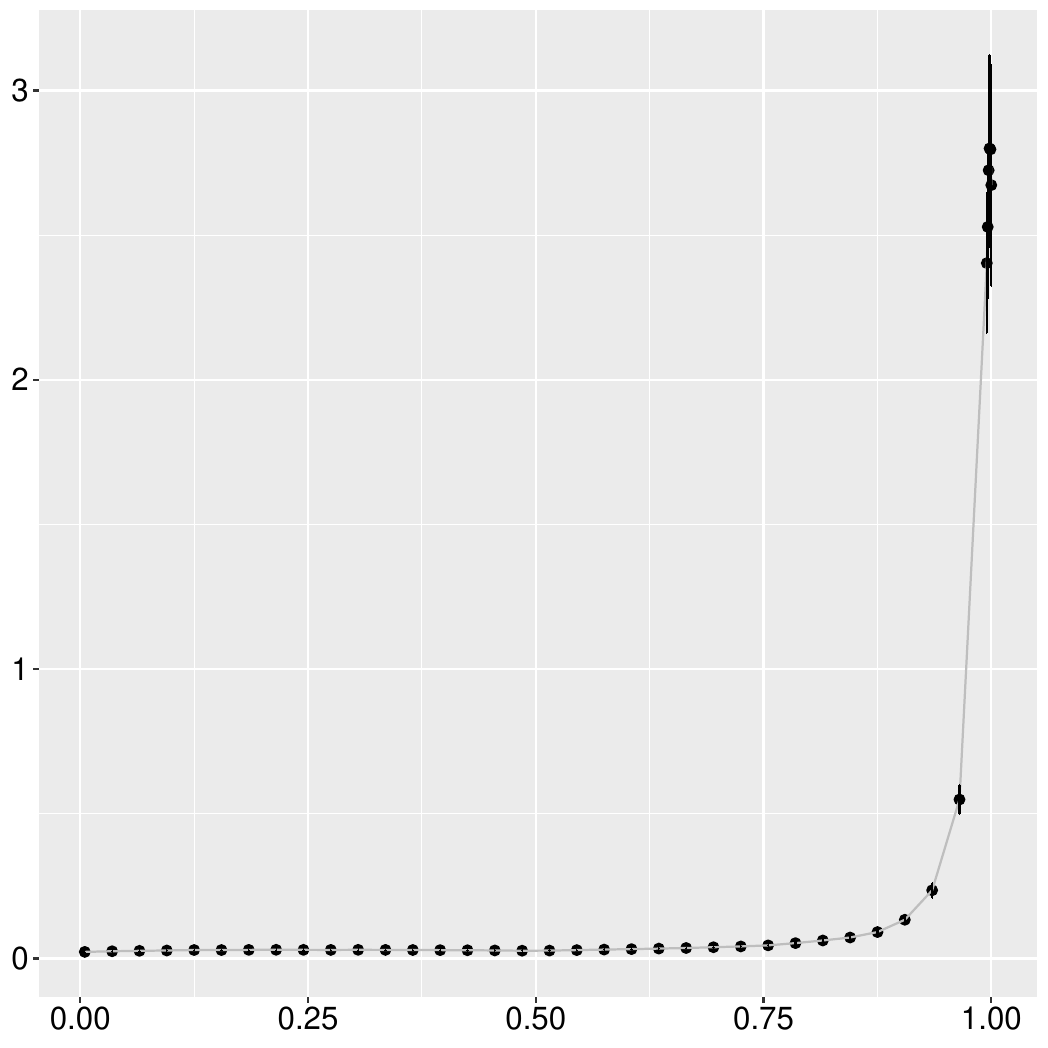}
      }
      \put(41,5){$\alpha$} \put(0,35){\rotatebox{90}{time (s)}}
      \put(83,8){
        \includegraphics[height=2.8in] {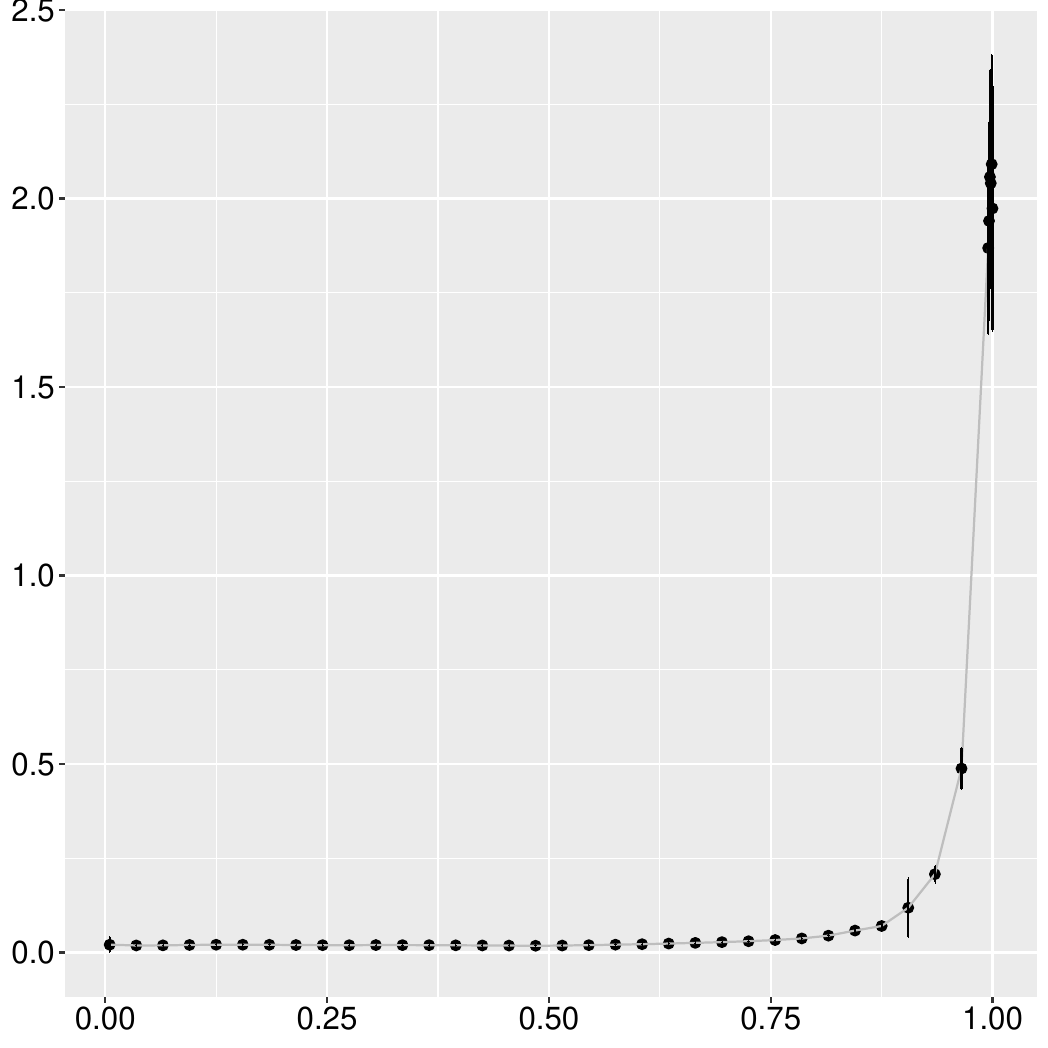}
      }
      \put(122,5){$\alpha$} \put(0,35){\rotatebox{90}{time (s)}}
    \end{picture}
  \end{center}
\end{figure}

\begin{table}[t]
  \caption{Mean and SD of running time (in s) of \Cref{a:fp-modified}
    over 100 repetitions, with $n=10^4$ (top) or
    $n=10^5$ (bottom) values sampled per repetition.  The results for
    $n=10^4$ are also plotted in \Cref{f:fp}.} \label{t:fp}
  \begin{center}
    \small
    \begin{tabular*}{.98\textwidth}{@{\extracolsep{\fill}}|c||c|c|c|c|c|c||c|c|c|c|c|c|}\hline
      & \multicolumn{6}{c||}{$b(t)\equiv10$} &
      \multicolumn{6}{c|}{$b(t)=(100-t^{1/\alpha})_+$}
      \\\hline
      $\alpha$ & .995 & .996 & .997 & .998 & .999 & .9999
      & .995 & .996 & .997 & .998 & .999 & .9999
      \\\hline
      \multicolumn{1}{|c}{} &\multicolumn{12}{c|}{$n=10^4$}
      \\\hline
      Mean & 2.403 & 2.528 & 2.724 & 2.800 & 2.796 & 2.673
      & 1.869 & 1.940 & 2.057 & 2.041 & 2.091 & 1.974
      \\\hline
      SD & 0.241 & 0.251 & 0.267 & 0.323 & 0.294 & 0.349
      & 0.228 & 0.262 & 0.283 & 0.277 & 0.290 & 0.324
      \\\hline
      &\multicolumn{12}{c|}{$n=10^5$}
      \\\hline
      Mean & 7.613 & 8.059 & 8.330 & 8.912 & 8.848 & 8.650
      & 5.792 & 6.006 & 6.269 & 6.308 & 6.310 & 5.869
      \\\hline
      SD & 0.312 & 0.347 &  0.315 & 0.403 & 0.356 &0.388
      & 0.324 & 0.321 & 0.332 & 0.339 & 0.338 & 0.371
      \\\hline
    \end{tabular*}
  \end{center}
\end{table}


\appendix
\section{Proofs} \label{s:appendix}
\begin{proof}[Proof of \Cref{p:HATH1}]\phlabel{p:p:HATH1}
  Recall
  that $\sinc(z)$ is an analytic function and for $|z|<\pi$,
  \begin{align} \label{e:sinc-zeta}
    \ln[\sinc(z)]
    = -\sumoi n \frac{\zeta(2n)}{n} (z/\pi)^{2n};
  \end{align}
  see \cite{NIST:10}, 25.8.8.  
  Then for $c\in (0,1)$ and $|x|<1$, 
  \[
    \frac{\dd}{\dd x}\Sbr{\ln\frac{\sinc(c\pi x)}{\sinc(\pi x)}}
    =
    \frac{\dd}{\dd x} \Sbr{\sumoi n \frac{\zeta(2n)}{n} (1-
      c^{2n}) x^{2n}}
    =
    2\sumoi n \zeta(2n) (1-c^{2n}) x^{2n-1}.
  \]
  Then from \eqref{e:HAT}, it follows that 
  \begin{align*}
    \frac{\dd}{\dd x}[\ln \HAT[\pi x]]
    &=
    2\sumoi n \zeta(2n) (1-\delta^{2n}) x^{2n-1}
    +
    \frac{2\alpha}\delta\sumoi n \zeta(2n) (1-\alpha^{2n}) x^{2n-1}
    \\
    &=
    2\alpha \sumoi n \zeta(2n)\sum^{2n-1}_{k=0} (\alpha^k + \delta^k)
    x^{2n-1}.
  \end{align*}
  From $\ln \HAT[0]=0$, the equality in \eqref{e:lHAT} follows.  Treat
  $\alpha^k + \delta^k$ as a function of $\alpha$.  If
  $k=0$, then $\alpha^k + \delta^k\equiv 2$ on $[\nth2,1]$.  On the
  other hand, if $k\ge1$, then $\alpha^k + \delta^k$ is increasing on
  $[\nth2, 1]$.  As a result, given $\theta$, $\alpha^{-1} \ln\HAT$ is
  less than the limit as $\alpha\to1$, which is $\ln\HAT[\theta][1]$.
  This also yields $\ln \HAT[\theta][1]= \sumoi n \zeta(2n)(2+1/n)
  (\theta/\pi)^{2n} = \theta^2/2 + V_1(\theta)$, as claimed.
  Similarly, given $0<\theta_1 <  \theta_2<\pi$, $\alpha^{-1} [\ln
  \HAT[\theta_2] - \ln \HAT[\theta_1]]$ as a function of $\alpha$ is
  increasing on $[1/2,1]$.  Then it is easy to see that
  $\HAT[\theta_2]/\HAT[\theta_1]$ is increasing on $[1/2,1]$. 
\end{proof} 

\begin{proof}[Proof of \Cref{p:H-sand}]\phlabel{p:p:H-sand}
  Denote $u=\pi-\theta$.  By $\pi - \alpha\theta=\delta\pi + \alpha u$, 
  \begin{align*}
    \hAT
    =
    \frac{\sin(\delta\theta)[\sin(\alpha\theta)]^{\alpha/\delta}}
    {[\sin(\theta)]^{1/\delta}}
    &=
    \frac{\sin(\delta\theta)}{\sin(u)}
    \Sbr{\frac{\sin(\delta\pi + \alpha u)}{\sin(u)}}^{\alpha/\delta}
    \\
    &\hspace{-2cm}=
    \XiAT
    \frac{\delta\theta}{u}
    \Grp{\frac{\delta\pi+\alpha u}u}^{\alpha/\delta}
    =    
    \delta \alpha^{\alpha/\delta} \XiAT 
    \frac\theta u
    \Grp{1+ \frac{\delta\pi}{\alpha u}}^{\alpha/\delta}.
  \end{align*}
  Since $\hAT[0]=\delta \alpha^{\alpha/\delta}$ and $\HAT =
  \hAT/\hAT[0]$, then \eqref{e:HAT-factor} follows.

  From \eqref{e:sinc-zeta} and $u=\pi-\theta$,
  \begin{align*}
    [\delta\ln\XiAT]'
    &=
    \sumoi n \frac{\zeta(2n)}{n} \pi^{-2n}
    \frac{\dd}{\dd\theta}[(\pi-\theta)^{2n} -
    \delta (\delta\theta)^{2n} - \alpha(\pi-\alpha\theta)^{2n}]
    \\
    &=
    2\sumoi n \zeta(2n) \pi^{-2n}[-(\pi-\theta)^{2n-1}-
    \delta^2(\delta\theta)^{2n-1} + \alpha^2(\pi-\alpha\theta)^{2n-1}].
  \end{align*}
  Let $\alpha\in (1/2,1)$.  Then $I:=[(\nth3 + \nth{3\alpha})\pi,
  \pi)\ne\emptyset$.  To show that $\XiAT$ is increasing on $I$, it 
  suffices to show that for each $n\ge1$, $\alpha^2(\pi -
  \alpha\theta)^n \ge (\pi-\theta)^n + \delta^2(\delta\theta)^n$ on
  $I$, or equivalently $\alpha^2 \ge [(\pi-\theta)/(\pi -
  \alpha\theta)]^n + \delta^2[\delta\theta/(\pi - \alpha\theta)]^n$.
  By $\delta=1-\alpha$, it is easy to check that the inequality holds
  for $n=1$.  Then by $\pi-\alpha\theta \ge \max(\pi-\theta, \delta
  \theta) \ge0$, the inequality holds for all $n\ge1$.  The claim on
  $\XIAT$ then easily follows.  Finally, from the above display
  \begin{align*}
    [\ln\XiAT]'
    &\le
    \frac2\delta \sumoi n \zeta(2n) \pi^{-2n}
    [(\pi-\alpha\theta)^{2n-1}-(\pi-\theta)^{2n-1}]
    \\
    &\le
    2\theta\sumoi n \zeta(2n) \pi^{-2n}
    (2 n-1)(\pi-\alpha\theta)^{2n-2}.
  \end{align*}
  By \eqref{e:sinc-zeta}, the \rhs is $\theta[\ln
  \nth{\sinc(t)}]''|_{t=\pi - \alpha\theta}$.  Since
  $[\nth{\sinc(t)}]'' = \nth{(\sin t)^2} - \nth{t^2}$, the inequality
  for $[\ln\XiAT]'$ follows.  The last equality follows by using
  \eqref{e:HAT-factor} and direct calculation.
\end{proof} 

\begin{proof}[Proof of \Cref{p:fpus-squeeze}]\phlabel{p:p:fpus-squeeze}
  Let $k_\alpha(y)=(y^{-\alpha}+1)[1- (y+1)^{-\delta/\alpha}]^\alpha$,
  $y>0$.  Then \eqref{e:fpus-squeeze} holds $\Iff c_{1,\alpha}\le \nth
  {\sup k_\alpha}$ and $c_{2,\alpha}\ge \nth{\inf k_\alpha}$.  Since
  on $(0,\infty)$, $f(y):= (y+1)^{-\delta/\alpha}$ is decreasing and
  convex and takes values in $(0,1)$,
  \begin{align*}
    \sup k_\alpha
    &\le \sup_{y>0} y^{-\alpha}[1- f(y)]^\alpha +
    \sup_{y>0} [1- f(y)]^\alpha
    \\
    &\le \sup_{y>0} \Abs{\frac{f(y)-1}y}^\alpha + 1
    \le |f'(0)|^\alpha + 1 = (\delta/\alpha)^\alpha + 1 =
    \nth{c_\alpha}+1.
  \end{align*}
  Then the first inequality in \eqref{e:fpus-squeeze} holds.  Since
  $c_\alpha$ is a continuous and nonnegative function of $\alpha\in
  (0,1]$ and tends to 1 as $\alpha\to0$, then $\inf_{\alpha\in (0,1)}
  c_{1,\alpha}>0$.  To show the second inequality in
  \eqref{e:fpus-squeeze}, if $\alpha\le1/2$, then $\delta/\alpha\ge1$,
  giving
  \[
    k_\alpha(y)
    \ge
    (y^{-\alpha}+1)[1- (y+1)^{-1}]^\alpha = t^\alpha + (1-t)^\alpha,
    \quad
    \text{with~} t = 1/(y+1).
  \]
  The infimum of the \rhs is attained at $t=0,1$, so $\inf k_\alpha\ge
  1$ and the second inequality in \eqref{e:fpus-squeeze} holds
  if $c_{2,\alpha} = 1$.  On the other hand, if
  $\alpha>1/2$, then by $|f'(y)|\le |f'(0)| =\delta/\alpha < 1$,
  $0<[1-f(y)]/y<1$.  Then 
  \[
    k_\alpha(y)
    =[1-f(y)]^\alpha + \Sbr{\frac{1-f(y)}y}^\alpha
    \ge 1-f(y) + \frac{1-f(y)}y
    =
    \frac{(y+1) [1-f(y)]}{y}.
  \]
  Since $(y+1)[1-f(y)] = (y+1) - (y+1)^{1-\delta/\alpha}$ is convex in
  $y\ge0$ and equal to 0 at $y=0$, the infimum of the \rhs is attained
  at $y=0$ with value $\delta/\alpha$.  Then $\inf k_\alpha\ge
  \delta/\alpha$ and the second inequality \eqref{e:fpus-squeeze}
  holds.
\end{proof}

\begin{proof}[Proof of \Cref{l:icgamma}] \phlabel{p:l:icgamma}
  If $b<a+1$, then
  \[
    e^{-b} \int^b_a x^{c-1}\,\dd x \le \int g_{a,b,c} =
    \int^b_a x^{c-1} e^{-x}\,\dd x \le e^{-a} \int^b_a x^{c-1}\,\dd x.
  \]
  Since $e^{-a} \int^b_a x^{c-1}\,\dd x = G(a,b,c)$ and $e^{-b} >
  e^{-a-1}$, then \eqref{e:icgamma} holds.  If $a+1\le b<a+2$, then
  from $g_{a,b,c} = g_{a,a+1,c} + g_{a+1,b,c}$ and the above proof,
  \eqref{e:icgamma} holds.
  
  Let $b>a+2$.  Put $d=a+2$.  First, suppose $c\in [1,2)$.  Then 
  \begin{align*}
    e^{-1} [G(a,a+1,c) &+ G(a+1,a+2,c)] + \int g_{d,b,c}
    \\
    &\le
    \int g_{a,b,c} \le G(a,a+1,c) + G(a+1,a+2,c) +\int g_{d,b,c}.
  \end{align*}
  From $\int g_{d,b,c} = \int^b_d [x^{c-1} - (c-1) x^{c-2}]
  e^{-x}\, \dd x + \int^b_d (c-1) x^{c-2} e^{-x}\,\dd x =
  d^{c-1} e^{-d} - b^{c-1} e^{-b} + \int^b_d (c-1) x^{c-2}
  e^{-x}\,\dd x$ and 
  \[
    \frac{\int^b_d [x^{c-1} - (c-1) x^{c-2}) e^{-x}\, \dd x}{
      \int^b_d (c-1) x^{c-2} e^{-x}\,\dd x}
    \ge
    \inf_{x\ge d} \frac{[x^{c-1} - (c-1) x^{c-2}] e^{-x}}{
      (c-1) x^{c-2} e^{-x}} \ge1,
  \]
  it follows that $d^{c-1} e^{-d} - b^{c-1} e^{-b} \le \int g_{d,b,c}
  \le 2[d^{c-1} e^{-d} - b^{c-1} e^{-b}]$.  Then \eqref{e:icgamma}
  holds.

  Now let $c\in (0,1)$.  As in the previous case, it suffices to show
  $e^{-1} d^{c-1} (e^{-d} - e^{-b})\le \int g_{d,b,c} \le d^{c-1}
  (e^{-d} - e^{-b})$.  Since $c<1$, $\int g_{d,b,c} \le \int^b_d
  d^{c-1}e^{-x}\,\dd x$, yielding the second inequality.  If $b\le
  d+1$, then by $1-c<1$, $\int g_{d,b,c} = d^{c-1}\int^b_d
  (\frac dx)^{1-c} e^{-x}\,\dd x \ge d^{c-1} \int^b_d(\frac{d}{d+1})
  e^{-x}\,\dd x \ge \frac23 d^{c-1} \int^b_d e^{-x}\,\dd x$.  Since
  $\frac23>1/e$, the first inequality follows.  If $b>d+1$, then by
  $g_{d,b,c} > g_{d,d+1, c}$, $\int g_{d,b,c}\ge (2/3) d^{c-1}
  \int^{d+1}_d e^{-x}\,\dd x \ge e^{-1}d^{c-1} e^{-d}\ge
  e^{-1}d^{c-1}(e^{-d}  - e^{-b})$.
\end{proof} 

\begin{proof}[Proof of \Cref{p:chi-G}]\phlabel{p:p:chi-G}
  Since $f(s)=1-e^{-s}$ is concave, for $s>0$, $f'(s) s \le f(s) \le
  f'(0) s$, implying $e^{-1}(s\wedge1) \le 1-e^{-s}\le s\wedge1$.
  Letting $s = (\delta/\alpha) \ell(y)$,
  \begin{gather*}
    e^{-1} \le \frac{1-(y+1)^{-\delta/\alpha}}{[(\delta/\alpha)
      \ell(y)]\wedge1}\le1\implies\\
    \nth{[(\delta/\alpha)\ell(y)]^\alpha \wedge1}
    \le \nth{[1-(y+1)^{-\delta/\alpha}]^\alpha}
    \le
    \frac{e^\alpha}{[(\delta/\alpha)\ell(y)]^\alpha \wedge1}.
  \end{gather*}
  Since $\nth2(1/a + 1/b) \le(a\wedge b)^{-1}\le 1/a + 1/b$ for
  $a,b>0$,
  \begin{align} \label{e:fpus-sandwich}
    \nth2\Cbr{\frac{c_\alpha}{[\ell(y)]^\alpha}+1}
    \le \nth{[1-(y+1)^{-\delta/\alpha}]^\alpha}
    \le
    e^\alpha\Cbr{\frac{c_\alpha}{[\ell(y)]^\alpha}+1}.
  \end{align}
  Put
  \[
    F^*_\alpha(y,t) = \frac{\cf{0<y<1/t} + \cf{y\ge1/t}
      e^{-ty}}{[\ell(y)]^\alpha} +  \frac{\cf{y>0} e^{-ty}}{c_\alpha}.
  \]
  It is easy to see that for any $t>0$ and $y>0$,
  \[
    \frac{c_\alpha}{[\ell(y)]^\alpha} +  1
    \le
    c_\alpha F^*_\alpha(y,t) e^{ty} \le e
    \Cbr{\frac{c_\alpha}{[\ell(y)]^\alpha} +  1}.
  \]
  Let $G^*\alz(y,\theta) =e^\alpha c_\alpha \HAT e^{-z\HAT}
  F^*_\alpha(y,z\HAT)$.  Then from and \eqref{e:fpus-sandwich} and 
  \eqref{e:fpus-joint},
  \begin{align} \label{e:chi-G*}
    \frac{\chi\alz(y,\theta)}{G^*\alz(y,\theta)}
    =
    \frac{[1-(y+1)^{-\delta/\alpha}]^{-\alpha}}{
       e^\alpha c_\alpha F^*_\alpha(y, z\HAT) e^{z\HAT y}
    } \in [e^{-\alpha-1}/2, 1].
  \end{align}

  Let $r^*_{t,1}(y) = \cf{y\ge1/t}[\ell(y)]^{-\alpha} e^{-ty}$.  By
  \eqref{e:m-r},
  \begin{align} \label{e:F*-F}
    F^*_\alpha(y,t) = m_t(y) + r^*_{t,1}(y) +
    r_{t,2}(y) \le m_t(y) + r_{t,1}(y) + r_{t,2}(y) = F_\alpha(y,t).
  \end{align}
  Then from \eqref {e:F-G} and \eqref{e:chi-G*}, the first
  inequality in \eqref {e:chi-G} follows.  Next, from
  \eqref{e:chi-G*},
  \begin{align*}
    \frac{\iint \chi\alz(y,\theta)\,\dd y\,\dd
      \theta}{\iint G\alz(y,\theta)\,\dd y\,\dd\theta}
    &=
    \frac{\iint \chi\alz(y,\theta)\,\dd y\,\dd
      \theta}{\iint G^*\alz(y,\theta)\,\dd y\,\dd\theta}
    \frac{\iint G^*\alz(y,\theta)\,\dd y\,\dd
      \theta}{\iint G\alz(y,\theta)\,\dd y\,\dd\theta}
    \\
    &
    \ge
    \nth{2e^{\alpha+1}}
    \frac{\int \HAT e^{-z\HAT} [\int F^*_\alpha(y,z\HAT)\,\dd y]
      \dd \theta}
    {\int \HAT e^{-z\HAT} [\int F_\alpha(y,z\HAT)\,\dd y]\dd\theta}
    \\
    &\ge
    \nth{2e^{\alpha+1}}
    \inf_{t>0}
    \frac{\int F^*_\alpha(y,t)\,\dd y}{\int F_\alpha(y,t)\,\dd y}.
  \end{align*}
  For each $t>0$, by \eqref{e:F*-F},
  \[
    \frac{\int F^*_\alpha(y,t)\,\dd y}{\int F_\alpha(y,t)\,\dd y}
    \ge
    \frac{\int r^*_{t,1}(y)\,\dd y}{\int r_{t,1}(y)\,\dd y}
    =
    \frac{\int^\infty_{1/t} e^{-ty}[\ell(y)]^{-\alpha}\,\dd y}
    {\int^\infty_{1/t} e^{-ty}[\ell(1/t)]^{-\alpha}\,\dd y}
    =
    \int^\infty_1 \Sbr{\frac{\ell(1/t)}{\ell(y/t)}}^\alpha
    e^{-y+1}\,\dd y.
  \]
  Since $\ell$ is concave with $\ell(0)=0$ and $y\ge1$, $y\ell(1/t)\ge
  \ell(y/t)$.  Therefore, the \rhs is at least $\int^\infty_1
  y^{-\alpha} e^{-y+1}\,\dd y$.  This combined with the previous
  display then finishes the proof.
\end{proof}

\begin{proof}[Proof of \Cref{p:nodes}]\phlabel{p:p:nodes}
  Since $\alpha\ge\alpha_0$, by \eqref{e:pars}, $\theta_0 > \pi/3 +
  \pi/(3\alpha)$.  From \Cref{p:H-sand}, $\XIAT$ is increasing on
  $[\theta_0, \pi)$, so every ratio in \eqref{e:K-nodes} is at least
  1.  Next, from \eqref{e:KAT}, 
  \[
    \frac{\XIAT[\theta']}{\XIAT}=
    \frac{\theta'/(\pi-\alpha\theta')}{\theta/(\pi-\alpha\theta)}
    \frac{\XiAT[\theta']}{\XiAT}, \quad
    \theta_0 \le\theta\le\theta'\le\pi.
  \]
  By \Cref{p:H-sand}, $0<\ln\frac{\XiAT[\theta']}{\XiAT} \le \theta^*
  [\ln\nth{\sinc(x)}]''|_{x = \pi - \alpha\theta^*} (\theta' - 
  \theta)$ for some $\theta^*\in (\theta, \theta')$.  From \eqref
  {e:sinc-zeta}, $[\ln\nth{\sinc(x)}]''$ is increasing on $(0,\pi)$.
  Let $C$ be as in \Cref{p:nodes}.  Then $\ln\frac{\XiAT[\theta']}
  {\XiAT}\le C(\theta' - \theta)$.  Then
  $\XIAT[f_\alpha(\theta)]/\XIAT \le (1+\Delta/2) \exp[C 
  (f(\theta) - \theta)] \le 1+\Delta$.  Furthermore, $\theta\le
  f_\alpha(\theta) \le \pi$ with either equality holding
  $\Iff\theta=\pi$.  Next show that the number of $\theta_i$ in $(0,
  \pi)$ is $O(\ln(1/\delta))$.  Put
  \[
    \rx = \nth C\ln\frac{1+\Delta}{1+\Delta/2}.
  \]
  Then
  \begin{align} \label{e:falpha}
    f_\alpha(\theta) = \min(\pi, g_\alpha(\theta), \theta+\rx),
    \text{~where~}
    g_\alpha(\theta) = \frac{(1+\Delta/2)\theta}
    {1+\alpha(\Delta/2)\theta/\pi}.
  \end{align}
  From
  \[
    g'_\alpha(\theta) = \frac{1+\Delta/2}
    {[1+\alpha(\Delta/2)\theta/\pi]^2},
  \]
  $g_\alpha(\theta)$ is strictly increasing and concave.  From the
  choice of $\alpha_0$ and $\theta_0$ in \eqref{e:pars},
  $\alpha\theta/\pi \ge \alpha_0\theta_0/\pi>1/2$.  Then
  \[
    g'_\alpha(\theta)\le g'_\alpha(\theta_0)
    <\frac{1+\Delta/2}{1+\alpha\Delta\theta_0/\pi} <1.
  \]
  Then $g_\alpha(\theta) = \theta+\rx$ has at most one solution in
  $(\theta_0, \infty)$.  Suppose there is a solution $\theta_*$.  Then
  $g_\alpha(\theta)\ge \theta+\rx$ on $[\theta_0,\theta_*]$ and
  $g_\alpha(\theta) \le \theta+\rx$ on $[\theta_*,\infty)$.  For every
  $\theta_i \le \theta_*$, $\theta_{i+1} =  \theta_i + \rx$.  As a
  result, the number of $\theta_i$'s in $[\theta_0, \theta_*]$ is
  $O(1/\rx) = O(1)$.  On the other hand, if $\theta_i$ and
  $\theta_{i+1}$ are in $[\theta_*, \pi)$, then $\theta_{i+1} =
  g_\alpha(\theta_i)$, giving $\frac{\theta_{i+1}}{\pi -
    \alpha\theta_{i+1}} = (1+\frac\Delta 2) \frac{\theta_i}{\pi -
    \alpha\theta_i}$.  By induction, if $\theta_{i+k}\in
  [\theta_*,\pi)$, then $\frac{\theta_{i+k}}{\pi - \alpha\theta_{i+k}}
  = (1+\frac\Delta 2)^k \frac{\theta_i}{\pi - \alpha\theta_i}\ge
  (1+\frac\Delta 2)^k \frac{\theta_0}{\pi - \alpha\theta_0}$.
  However, when $\frac\theta{\pi - \alpha\theta}$ is increasing on
  $[\theta_0, \pi]$, and when $\theta=\pi$, $\frac\theta{\pi -
    \alpha\theta} = \nth\delta$.  As a result, $\nth\delta
  \ge(1+\frac\Delta 2)^k \frac{\theta_0}{\pi -\alpha\theta_0}$.  Then
  the number of $\theta_i$ in $[\theta_*, \pi)$ is $O(\ln(1/\delta))$.
  Therefore the number of $\theta_i$'s in $[\theta_0, \pi)$ is
  $O(\ln(1/\delta))$.  The case where $g_\alpha(\theta) = \theta+\rx$
  has no solution on $(\theta_0,\infty)$ can be similarly dealt with.

  Finally, since $\theta_\nal$ is the last $\theta_i$ in $[\theta_0,
  \pi)$, then $g_\alpha(\theta_\nal)\ge\pi$ and $\theta_\nal+\rx
  \ge\pi$.  Then it is easy to get $\XIAT[\pi]/\XIAT\le1+\Delta$.  The
  proof is complete.
\end{proof} 

To prove \Cref{p:N-J-sup}, we need the following two lemmas.
\begin{lemma} \label{l:convex}
  Let $c\ge0$ be a constant and $F(x)>0$ be a continuous function on
  an interval $I$.  Let $I_c = \{x\in I: F(x)\ge c\}$.
  \begin{subenum}
  \item If $F$ is increasing on $I$, then $F^c e^{-F}$ is decreasing
    on $I_c$.
  \item If $F$ is convex, then $F^c e^{-F}$ is log-concave on any
    interval in $I_c$. 
  \end{subenum}
\end{lemma}

\begin{proof}
  Let $f(v) = c\ln v - v$.  Then $f$ is concave and decreasing on
  $[c,\infty)$.  Since $\ln (F^c e^{-F}) = f\circ F$, 1) is plain.
  For 2), fix any
  $[x,y]\subset I_c$ and $a\in (0,1)$.  Let $z=(1-a)x + a y$.  Then
  $c\le F(z) \le (1-a) F(x) + a F(y)$, giving $f(F(z)) \ge
  f((1-a) F(x) + a F(y)) \ge (1-a) f(F(x)) + a f(F(y))$.  As a result,
  $f\circ F$ is concave.  Then $F^c e^{-F}$ is log-concave.
\end{proof} 

\begin{lemma} \label{l:integral}
  Let $f>0$ be a function on an interval $I$ with $0<\int_I
  f<\infty$, and $g$ a strictly increasing and differentiable
  convex function on an interval $J$ with $g(J)=I$.  Then
  \[
    \frac{\int_{J\cap (-\infty,x]} f\circ g}{\int_{I\cap (-\infty,
        g(x)]} f}
    \ge
    \nth{g'(x)} \ge
    \frac{\int_{J\cap [x,\infty)} f\circ g}{\int_{I\cap [g(x),
        \infty)} f} 
  \]
  for any non-boundary point $x\in J$.
\end{lemma}

\begin{proof}
  Since $x$ is in the inner part of $J$, then $g(x)$ is in the inner
  part of $I$, so $\int_{I\cap (-\infty, g(x)]} f>0$ and $\int_{I\cap
    [g(x), \infty)} f>0$.  Then $\int_{I\cap (-\infty, g(x)]} f =
  \int_{J\cap (-\infty, x]} (f\circ g) g' \le g'(x) \int_{J\cap
    (-\infty, x]} f\circ g$ and $\int_{I\cap [g(x), \infty} f =
  \int_{J\cap [x, \infty)} (f\circ g) g'\ge g'(x) \int_{J\cap
    [x,\infty)} f\circ g$, yielding the inequalities.
\end{proof} 

\begin{proof}[Proof of \Cref{p:N-J-sup}] \phlabel{p:p:N-J-sup}
  Let $\theta\in [\vartheta_z,\pi)$.  As in \eqref{e:calpha}, let
  $\tau = z\HAT$.  Since $z\JAT\ge1$, by \eqref{e:JH}, $\tau\ge1$.
  Then $(1/\tau)\ln2\le \ln(1+1/\tau) \le1/\tau$, so from the
  expression of $\psi_\alpha$ in \eqref{e:kappa-psi-alpha},
  $\kappa_{\alpha,4}\le \lfrac{\psi_\alpha(\tau)}{\tau^\alpha}
  \le\kappa_{\alpha,3}$.  Define $D\alz(\theta) =
  \kappa_{\alpha,3}\tau^\alpha e^{-\tau}$.  Then from \eqref{e:Dpsi2},
  \begin{align} \label{e:K-D-ratio}
    \frac{\kappa_{\alpha,4}}
    {\kappa_{\alpha,3}}
    \le \frac{Q\alz(\theta)}{D\alz(\theta)}
    = \frac{\psi_\alpha(\tau)}{\kappa_{\alpha,3} \tau^\alpha}
    \le1,
    \quad
    \theta\in [\vartheta_z, \pi),
  \end{align}
  in particular, $Q\alz(\theta) \le D\alz(\theta)$.  From \cite
  {gonzalez:23:arxiva}, $z\HAT$ is strictly increasing and 
  convex.  Then by \Cref{l:convex}, $D\alz(\theta)$ is decreasing and
  log-concave on $[\vartheta_z, \pi)$.  Therefore, by
  \eqref{e:log-concave}, for any $A\ge \int^\pi_{\vartheta_z} D\alz$,
  \begin{align} \label{e:K-D-log-concave}
    D\alz(\theta) \le 2 A \Vrx_s(\theta - \vartheta_z)
    \quad\text{with~} s = D\alz(\vartheta_z)/A.
  \end{align}
  From \Cref{l:integral},
  \begin{align*}
    \int^\pi_{\vartheta_z} D\alz
    \le
    \frac{\kappa_{\alpha,3}}{z\dHAT[\vartheta_z]}
    \int^\infty_{z\HAT[\vartheta_z]} s^\alpha e^{-s}\,\dd s
    &\le
    \frac{\kappa_{\alpha,3}}{z\dHAT[\vartheta_z]}
    \intzi s^\alpha e^{-s}\,\dd s
    \\
    &=
    \frac{\kappa_{\alpha,3}\Gamma(1+\alpha)}{z\dHAT[\vartheta_z]}
    \le \frac{\kappa_{\alpha,3}}{z\dHAT[\vartheta_z]}.
  \end{align*}
  Since $z\dHAT[\vartheta_z] = z\HAT[\vartheta_z] (\ln
  H_\alpha)'(\vartheta_z) = \tau_z \varrho_z$, then 
  letting $A=\kappa_{\alpha,3}/[z\dHAT[\vartheta_z]]$ in \eqref
  {e:K-D-log-concave} gives \eqref{e:z-star}.  Since $z\HAT[\vartheta_z]
  \ge z\JAT[\vartheta_z]\ge1$, then from \Cref{l:integral} and 
  \eqref{e:K-D-ratio}
  \[
    \nth{z\dHAT[\vartheta_z]}
    \le
    \frac{\int^{\vartheta_z}_0 D\alz}
    {\kappa_{\alpha,3} \int^{z\HAT[\vartheta_z]}_{z
        \HAT[0]} s^\alpha e^{-s}\,\dd s}
    \le \frac{\int^{\vartheta_z}_0 Q\alz}{\kappa_{\alpha,4}\int^1_z
      s^\alpha e^{-s}\,\dd s}. 
  \]
  Since $s^\alpha>s$ for $s\in (z,1)$, then \eqref{e:K-sup} follows.

  Finally, for $\alpha\ge\alpha_0$, since $\alpha_0>1/2$, then
  $2\alpha-1\ge0$, so
  \begin{align*}
    \frac{\kappa_{\alpha,3}}{\kappa_{\alpha,4}}
    &=
    \frac
    {(2\alpha-1)/(1-\alpha) + (\ln2)^{-\alpha}(2+e^{-1}) +
      [(1-\alpha)/\alpha]^\alpha}
    {(\ln2)^{1-\alpha}(2\alpha-1)/(1-\alpha) + 2+e^{-1}}
    \\
    &\le
    \frac
    {(2\alpha-1)/(1-\alpha) + (\ln2)^{-\alpha}(2+e^{-1}) +
      [(1-\alpha)/\alpha]^\alpha}
    {(\ln2)^{1-\alpha}(2\alpha-1)/(1-\alpha)}
    \\
    &=
    (\ln 2)^{\alpha-1} + \frac{2+e^{-1}}{\ln
      2}\frac{1-\alpha}{2\alpha-1} +
    \Grp{\frac{1-\alpha}\alpha}^\alpha
    \frac{1-\alpha}{2\alpha-1}.  \tag*{\qedhere}
  \end{align*}
\end{proof}

For the proof of \Cref{p:N-J-sandwich}, we need the following.
\begin{lemma} \label{l:psi-alpha}
  The function  $\psi_\alpha(s)=\kappa_{\alpha,1}[\ell(1/s)]^\delta s +
  \kappa_{\alpha,2}[\ell(1/s)]^{-\alpha} + 1$ in
  \eqref{e:kappa-psi-alpha} is increasing on $(0,\infty)$ and for
  $a\ge1$, $0<\psi_\alpha(as)\le a\psi_\alpha(s)$.
\end{lemma}
\begin{proof}
  Since $\ell(t) = \ln(1+t)$ is concave in $t>0$, $\ell(t)/t =
  [\ell(t)-\ln(0)]/t$ is decreasing.  Letting $t=1/s$, $s\ell(1/s)$ is
  increasing, so $[\ell(1/s)]^\delta s = [s\ell(1/s)]^\delta s^\alpha$
  is increasing.  Then it is easy to see that $\psi_\alpha(s)$ is
  increasing.  Let $a\ge1$.  Clearly, $[\ell(1/(as))]^\delta(as)\le
  [\ell(1/s)]^\delta (as)$.  By concavity, $\ell(1/(as))\ge
  (1/a)\ell(1/s)$, so $[\ell(1/(as))]^{-\alpha}
  \le[(1/a)\ell(1/s)]^{-\alpha} < a [\ell(1/s)]^{-\alpha}$.  Then it
  follows that $\psi_\alpha(as)\le a\psi_\alpha(s)$.
\end{proof} 

\begin{proof}[Proof of \Cref{p:N-J-sandwich}]
  \phlabel{p:p:N-J-sandwich}
  From \eqref{e:Dpsi2}, \eqref{e:JH}, and \Cref{l:psi-alpha}, 
  \[
    Q\alz(\theta) = \psi_\alpha(z\HAT) e^{-z\HAT} \le
    \psi_\alpha(z\HAT) \le (1+\Delta)\psi_\alpha(z\JAT).
  \]
  On the other hand,  since $z\JAT<1$, $z\HAT<1+\Delta$, yielding
  \[
    Q\alz(\theta) \ge 
    e^{-1-\Delta} \psi_\alpha(z\HAT) \ge e^{-1-\Delta} 
    \psi_\alpha(z\JAT).  \eqno{\qedhere}
  \]
\end{proof}

\begin{proof}[Proof of \Cref{l:N-J-mid}]\phlabel{p:l:N-J-mid}
  Put $\phi(\theta) =\psi_\alpha(z\JAT)$, where again $\psi_\alpha$ is
  given in \eqref{e:Dpsi2}.  We will show that for each $n$, provided
  $I_n = [\theta_n\wedge\vartheta_z, \theta_{n+1}\wedge \vartheta_z)\ne
  \emptyset$,
  \begin{align} \label{e:cP}
    v_n P_n(t) \le \laz\#\phi(t)\le P_n(t),  \quad t\in \laz(I_n) =
    (c_n, d_n],
  \end{align}
  where $c_n$ and $d_n$ are as in \eqref{e:lambda-I}.  Once
  \eqref{e:cP} is established, it yields
  \[
    v_n \laz^{-1}\# P_n(\theta) \le \phi(\theta) \le \laz^{-1}\#
    P_n(\theta),
    \quad \theta\in I_n,
  \]
  which combined with \eqref{e:psi-phi} gives \eqref{e:Q-P}.

  To start, on the interval $[\theta_0, \vartheta_z)$, by $\laz(\theta)=
  \ell\Grp{\nth{z\JAT}}$,  
  \[
    \phi(\theta)
    = \kappa_{\alpha,1}[\laz(\theta)]^\delta
    (e^{\laz(\theta)}-1)^{-1} +
    \kappa_{\alpha,2}[\laz(\theta)]^{-\alpha} + 1.
  \]
  Since $\laz(\theta)$ is strictly decreasing and piecewise
  differentiable, for $t\in\laz([\theta_0, 
  \vartheta_z))$,
  \begin{align} \label{e:lambda-pw-phi}
    \laz\#\phi(t)
    = \phi(\laz^{-1}(t)) |(\laz^{-1})'(t)|
    = [\kappa_{\alpha,1} t^\delta(e^t-1)^{-1} + \kappa_{\alpha,2}
    t^{-\alpha} + 1]\times |(\laz^{-1})'(t)|.
  \end{align}
  We have  $\laz^{-1}(t) = \IJAT[z^{-1}(e^t-1)^{-1}]$.  Put
  \begin{align} \label{e:t-theta-s}
    \theta = \laz^{-1}(t), \quad s=[z(e^t-1)]^{-1}.
  \end{align}
  Then $s = \JAT$, $t=\ell(1/(zs))$, and by the chain rule, 
  \begin{align} \label{e:lambda}
    |(\laz^{-1})'(t)| = \dIJAT[s]\times \frac{e^t}{z(e^t-1)^2}.
  \end{align}
  Let $t\in \laz(I_n) = (c_n,
  d_n]$.  Then $\theta\in [\theta_n \wedge \vartheta_z,
  \theta_{n+1}\wedge\vartheta_z) = [\theta_n, \theta_{n+1})\cap [0,
  \vartheta_z)\ne\emptyset$.  By \eqref{e:JAT-upper},
  \begin{align} \label{e:JAT-inv}
    s = \XIAT[\theta_n]\Sbr{1+\frac{\delta\pi}{\alpha (\pi -
        \theta)}}^{1/\delta}, \quad
    \theta= \IJAT[s] = \pi-\frac{\delta\pi/\alpha}
    {[\frac s{\XIAT[\theta_n]}]^\delta-1}.
  \end{align}
  As a result,
  \begin{align} \label{e:JAT-inv-diff}
    \dIJAT[s]
    &= \nth{[\XIAT[\theta_n]]^\delta}
    \frac{(\delta\pi/\alpha) \delta s^{\delta-1}}
    {([\frac s{\XIAT[\theta_n]}]^\delta-1)^2}
    \\ \label{e:JAT-inv-diff2}
    &= \frac{(\alpha/\pi) s^{-\alpha}}{[\XIAT[\theta_n]]^\delta}(\pi
    - \theta)^2.
  \end{align}
  On the other hand, from \eqref{e:mid-range-consts}, $z s =
  z \JAT < z\JAT[(\theta_{n+1}\wedge \vartheta_z)-] = a_n$.   Then
  \begin{align} \label{e:expt}
    e^t - 1 = e^t(1 - e^{-t})  = e^t(1+zs)^{-1} \ge e^t/(1 + a_n).
  \end{align}
  Then from \eqref{e:t-theta-s} and \eqref{e:lambda},
  \begin{align} \label{e:inverse-diff}
    \begin{split}
      z^{-1}[\dIJAT[s]] e^{-t}<
      &
      |(\laz^{-1})'(t)|< z^{-1}(1+a_n)^2[\dIJAT[s]] e^{-t}, \\
      z^{-1} e^{-t} <
      &s < (1+a_n) z^{-1} e^{-t}.
    \end{split}
  \end{align}

  First, suppose $\theta_n\le (1-\delta/\alpha)\pi$.  Clearly,
  $\pi-\theta \le \pi-\theta_n$.  On the other hand, by
  \Cref{p:nodes}, $\theta_{n+1} = f_\alpha(\theta_n)$, so from
  \eqref{e:falpha},
  \[
    \pi-\theta \ge \pi - (\theta_{n+1}\wedge \vartheta_z)
    \ge\pi - g_\alpha(\theta_n)
    =\pi - \frac{(1+\Delta/2)\theta_n}{1 + \alpha
      (\Delta/2)\theta_n/\pi}
    = \frac{\pi - \theta_n - \delta(\Delta/2)\theta_n}{1+ \alpha
      (\Delta/2) \theta_n/\pi}.
  \]
  From $\delta(\Delta/2) \theta_n \le (\alpha-\delta)(\Delta/2)(\pi -
  \theta_n)$ and $\alpha(\Delta/2) \theta_n/\pi\le (\Delta/2)(\alpha -
  \delta)$, the \rhs is at least
  \[
    \frac{(\pi - \theta_n)[1- (\alpha - \delta)\Delta/2]}{
      1+ (\alpha - \delta)\Delta/2}
    \ge
    \frac{(\pi - \theta_n)(1- \Delta/2)}{1+ \Delta/2}.
  \]
  Then by \eqref{e:JAT-inv-diff2} and $[(1-\Delta/2)/(1+\Delta/2)]^2
  >(1-\Delta)/(1+\Delta)$,
  \[
    \Grp{\frac{1-\Delta}{1+\Delta}}
    \frac{(\alpha/\pi) s^{-\alpha}}{[\XIAT[\theta_n]]^\delta}
    (\pi - \theta_n)^2
    \le
    \dIJAT[s]\le \frac{(\alpha/\pi) s^{-\alpha}}
    {[\XIAT[\theta_n]]^\delta} (\pi - \theta_n)^2.
  \]
  This combined with \eqref{e:lambda-pw-phi}, \eqref{e:t-theta-s},
  \eqref{e:expt}, and \eqref{e:inverse-diff} yields two inequalities.
  First,
  \begin{align*}
    \laz\#\phi(t)
    &\le
    [\kappa_{\alpha,1} t^\delta(e^t-1)^{-1} +\kappa_{\alpha,2}
    t^{-\alpha} + 1] 
    \times \frac{(1+a_n)^2(\alpha/\pi) s^{-\alpha}}{z
      [\XIAT[\theta_n]]^\delta} (\pi - \theta_n)^2 e^{-t}
    \\
    &\le
    [(1+a_n)\kappa_{\alpha,1} t^\delta e^{-t}+ \kappa_{\alpha,2}
    t^{-\alpha} + 1] \times
    \frac{(1+a_n)^2(\alpha/\pi) z^\alpha e^{\alpha t}} 
    {z [\XIAT[\theta_n]]^\delta}(\pi - \theta_n)^2 e^{-t}
    \\
    &=
    \frac{(1+a_n)^2(\alpha/\pi)} {[z\XIAT[\theta_n]]^\delta}
    (\pi - \theta_n)^2 
    \times [(1+a_n)\kappa_{\alpha,1} t^\delta e^{-t}  +
    \kappa_{\alpha,2} t^{-\alpha} + 1] e^{-\delta t}.
  \end{align*}
  By \eqref{e:mid-range-P1}, the \rhs is $P_n(t)$.  On the other hand,
  \begin{align*}
    \laz\#\phi(t)
    &\ge
    [\kappa_{\alpha,1} t^\delta(e^t-1)^{-1} +\kappa_{\alpha,2}
    t^{-\alpha} + 1] 
    \times \Grp{\frac{1-\Delta}{1+\Delta}}
    \frac{(\alpha/\pi) s^{-\alpha}}{z[\XIAT[\theta_n]]^\delta}
    (\pi - \theta_n)^2 e^{-t}
    \\
    &\ge
    (\kappa_{\alpha,1} t^\delta e^{-t}+\kappa_{\alpha,2} t^{-\alpha}+
    1) \times\Grp{\frac{1-\Delta}{1+\Delta}}
    \frac{(\alpha/\pi) z^\alpha e^{\alpha t}} 
    {(1+a_n)^\alpha z[\XIAT[\theta_n]]^\delta}(\pi - \theta_n)^2
    e^{-t}
    \\
    &=\Grp{\frac{1-\Delta}{1+\Delta}}
    \frac{(\alpha/\pi)} {(1+a_n)^\alpha [z\XIAT[\theta_n]]^\delta}
    (\pi - \theta_n)^2 \times (\kappa_{\alpha,1} t^\delta e^{-t}  +
    \kappa_{\alpha,2} t^{-\alpha} + 1) e^{-\delta t}.
  \end{align*}
  The \rhs is at least $\Grp{\frac{1 - \Delta}{1 +\Delta}} (1+
  a_n)^{-\alpha-3} P_n(t) = v_n P_n(t)$, where $v_n$ is defined as in
  \Cref{l:N-J-mid}.  Then \eqref{e:cP} holds in the case where
  $\theta_n \le (1-\delta/\alpha)\pi$.

  Next, suppose $\theta_n>(1-\delta/\alpha)\pi$.  By
  \eqref{e:JAT-inv},
  \[
    \Sbr{\frac s{\XIAT[\theta_n]}}^\delta
    =1 + \frac{\delta\pi}{\alpha(\pi-\theta)}
    \ge1 + \frac{\delta\pi}{\alpha(\pi-\theta_n)}
    = \frac{\pi - \alpha\theta_n}{\alpha(\pi-\theta_n)},
  \]
  and so
  \[
    \Sbr{\frac s{\XIAT[\theta_n]}}^\delta -1
    \ge\Sbr{\frac s{\XIAT[\theta_n]}}^\delta
    \Cbr{1-\frac{\alpha(\pi-\theta_n)}{\pi - \alpha\theta_n}}
    = b_n \Sbr{\frac s{\XIAT[\theta_n]}}^\delta,
  \]
  where $b_n$ is given in \eqref{e:mid-range-consts}; note that
  $b_n>1/2$.  Then by \eqref{e:JAT-inv-diff},
  \[
    \nth{[\XIAT[\theta_n]]^\delta}
    \frac{(\delta\pi/\alpha) \delta s^{\delta-1}}
    {[\frac s{\XIAT[\theta_n]}]^{2\delta}}
    \le
    \dIJAT[s]
    \le
    \nth{[\XIAT[\theta_n]]^\delta}
    \frac{(\delta\pi/\alpha) \delta s^{\delta-1}}
    {b^2_n [\frac s{\XIAT[\theta_n]}]^{2\delta}},
  \]
  which combined with \eqref{e:t-theta-s} gives
  \[
    (\delta^2\pi/\alpha)[\XIAT[\theta_n]]^\delta
    [z(e^t-1)]^{\delta+1}
    \le
    \dIJAT[s]
    \le
    \nth{b^2_n}
    (\delta^2\pi/\alpha)[\XIAT[\theta_n]]^\delta
    [z(e^t-1)]^{\delta+1}.
  \]
  This combined with \eqref{e:lambda-pw-phi} and
  \eqref{e:inverse-diff} yields two inequalities.  First,
  \begin{align*}
    \laz\#\phi(t)
    &\le
    [\kappa_{\alpha,1} t^\delta(e^t-1)^{-1} +\kappa_{\alpha,2}
    t^{-\alpha} + 1] 
    \times
    \frac{(1+a_n)^2}{z b^2_n}(\delta^2\pi/\alpha)
    [\XIAT[\theta_n]]^\delta
    [z(e^t-1)]^{\delta+1} e^{-t}
    \\
    &=
    \frac{(1+a_n)^2}{b^2_n}(\delta^2\pi/\alpha)
    [z\XIAT[\theta_n]]^\delta
    [\kappa_{\alpha,1} t^\delta(e^t-1)^\delta +(\kappa_{\alpha,2}
    t^{-\alpha} + 1) (e^t-1)^{\delta+1}] e^{-t}
    \\
    &\le
    \frac{(1+a_n)^2}{b^2_n}(\delta^2\pi/\alpha)
    [z\XIAT[\theta_n]]^\delta
    [\kappa_{\alpha,1} t^\delta e^{-\alpha t} +(\kappa_{\alpha,2}
    t^{-\alpha} + 1) e^{\delta t}].
  \end{align*}
  From \eqref{e:mid-range-P2}, the \rhs is $P_n(t)$.  On the other
  hand,
  \begin{align*}
    \laz\#\phi(t)
    &\ge
    [\kappa_{\alpha,1} t^\delta(e^t-1)^{-1} +\kappa_{\alpha,2}
    t^{-\alpha} + 1] 
    \times
    z^{-1}(\delta^2\pi/\alpha)[\XIAT[\theta_n]]^\delta
    [z(e^t-1)]^{\delta+1} e^{-t}
    \\
    &=
    (\delta^2\pi/\alpha)[z\XIAT[\theta_n]]^\delta
    [\kappa_{\alpha,1} t^\delta(e^t-1)^\delta +(\kappa_{\alpha,2}
    t^{-\alpha} + 1) (e^t-1)^{\delta+1}] e^{-t}.
  \end{align*}
  Since $\theta<\vartheta_z$, $z\JAT<1$.  Then by \eqref{e:t-theta-s},
  $e^t = 1+\nth{z\JAT}\ge2$, giving $e^t-1\ge e^t/2$.  Then
  \begin{align*}
    \laz\#\phi(t)
    &\ge
    (\delta^2\pi/\alpha)[z\XIAT[\theta_n]]^\delta
    [\kappa_{\alpha,1} (t/2)^\delta e^{\delta t} +(\kappa_{\alpha,2}
    t^{-\alpha} + 1) (e^t/2)^{\delta+1}] e^{-t}
    \\
    &=
    2^{-\delta}
    (\delta^2\pi/\alpha)[z\XIAT[\theta_n]]^\delta
    [\kappa_{\alpha,1} t^\delta e^{-\alpha t} +(\kappa_{\alpha,2}
    t^{-\alpha} + 1) e^{\delta t}/2].
  \end{align*}
  The \rhs is at least $\frac{b^2_n}{2^{1+\delta} (1+a_n)^2} P_n(t) =
  v_n P_n(t)$.  Then \eqref{e:cP} again holds. 
\end{proof}

\begin{proof}[Proof of \Cref{p:N-J-mid2}] \phlabel{p:p:N-J-mid2}
  If
  $\theta_n\le (1-\delta/\alpha)\pi$, then by \eqref{e:mid-range-P1},
  for $t\in (c_n, d_n]$,
  \[
    P_n(t)
    =
    \pi_n
    [(1+a_n)\kappa_{\alpha,1} t^\delta e^{-(1+\delta) t}
    + \kappa_{\alpha,2} t^{-\alpha} e^{-\delta t} +  e^{-\delta t}]
    \cf{c_n<t\le d_n}
  \]
  From the definition of $g_{a,b,c}$ in \eqref{e:gabc},
  \begin{align*}
    &
    t^\delta e^{-(1+\delta) t}\cf{c_n < t\le d_n}
    \\
    &=(1+\delta)^{-\delta} [(1+\delta) t]^\delta
    e^{-(1+\delta)t}\cf{(1+\delta)c_n < (1+\delta) t\le (1+\delta)
      d_n} \\
    &=(1+\delta)^{-\delta} g_{(1+\delta)c_n,
      (1+\delta) d_n, 1+\delta}((1+\delta) t)
  \end{align*}
  and
  \[
    t^{-\alpha} e^{-\delta t}\cf{c_n < t\le d_n}
    = \delta^\alpha (\delta t)^{-\alpha} e^{-\delta t} \cf{\delta
      c_n < \delta t \le \delta d_n}
    = \delta^\alpha g_{\delta c_n, \delta
      d_n,\delta}(\delta t).
  \]
  Then by \eqref{e:icgamma-env}, $P_n(t)\cf{c_n < t\le d_n} \le
  P^*_n(t)$ and
  \[
    \nth{2e} \int 
    \laz^{-1}\#P^*_n
    =
    \nth{2e} \int 
    P^*_n
    \le
    \int P_n(t)\cf{c_n < t\le d_n}
    = \int_{I_n}\laz^{-1}\#P.
  \]
  This combined with \eqref{e:Q-P} yields the proof of
  \eqref{e:mid-range-P-env} when $\theta_n\le (1-\delta/\alpha)\pi$.

  On the other hand, if $\theta_n> (1-\delta/\alpha)\pi$, then by
  \eqref{e:mid-range-P2}, for $t\in (c_n, d_n]$,
  \[
    P_n(t)=\pi_n (\kappa_{\alpha,1} t^\delta e^{-\alpha t} +
    \kappa_{\alpha,2} t^{-\alpha}e^{\delta t} + e^{\delta t})
    \cf{c_n<t\le d_n}.
  \]
  Similar to the above argument, for $c_n<t\le d_n$, $t^\delta
  e^{-\alpha t} =\alpha^{-\delta} g_{\alpha c_n, \alpha d_n,
    1+\delta}(\alpha t)$ and $t^{-\alpha} e^{\delta t} =\delta^\alpha
  \varphi_{\delta c_n, \delta d_n,\delta}(\delta t)$, where
  $\varphi_{a,b,c}$ is defined in \eqref{e:varphi-def}.   Then the
  proof of \eqref{e:mid-range-P-env} in this case follows from 
  \eqref{e:varphi*}, \eqref{e:icgamma-env}, and
  \eqref{e:mid-range-P2}.
\end{proof}

\begin{proof}[Proof of \Cref{l:lambda-inv}] \phlabel{p:l:lambda-inv}
  The first expression in \eqref{e:lambda-inv} is a repeat of
  \eqref{e:JAT-inv}, and the second one follows from
  \eqref{e:JAT-upper} and direct calculation.
\end{proof}
\end{document}